\documentclass[a4paper,11pt,abstracton]{scrartcl}
\usepackage[latin1]{inputenc}
\usepackage{amsfonts}
\usepackage{amsmath}
\usepackage[pdftex]{graphics}
\usepackage{dsfont}
\usepackage{tikz}
\usepackage[pdftex, pdftitle={Recursive construction of OPE coefficients}, pdfauthor={Jan Holland and Stefan Hollands}, colorlinks=true, citecolor=blue]{hyperref}
\usepackage[ amsmath, thmmarks, hyperref]{ntheorem}
\usepackage{microtype}
\usepackage{mathrsfs}
\usepackage{mathptmx}
\usepackage{enumerate}

\def\ben{\begin{equation}}
\def\een{\end{equation}}
\def\bena{\begin{eqnarray}}
\def\eena{\end{eqnarray}}

\def\d{{\rm d}}

\def\C{{\mathcal C}}

\def\L{{\mathcal L}}
\def\D{D}

\def\mr{{\mathbb R}}

\def\G{{\mathcal G}}
\def\H{{\mathcal H}}
\def\F{{\mathcal F}}

\def\T{{\mathbb T}}

\newcommand{\mn}{{\mathbb N}}

\renewcommand{\D}{{\mathcal D}}
\newcommand{\R}{{\mathcal R}}

\renewcommand{\S}{{\mathscr S}}
\renewcommand{\O}{{\mathcal O}}

\newcommand{\e}{{\rm e}}

\renewcommand{\atop}[2]{\genfrac{}{}{0pt}{}{#1}{#2}}
\newcommand{\varp}{\frac{\delta}{\delta\varphi}}

\newcommand{\LscO}{\mathcal{L}^{0, \Lambda_{0}}}
\newcommand{\LscIr}{\mathcal{L}^{\Lambda_{0}, \Lambda_{0}}}
\def\bra{\langle }
\def\ket{\rangle}

\renewcommand{\bigotimes}{\otimes}

\newcommand{\La}{\Lambda}
\newcommand{\Lao}{\Lambda_0}

\newcommand{\vp}{\varphi}
\newcommand{\pa}{\partial}
\newcommand{\de}{\delta}

\newcommand{\vecp}{\vec{p}}

\newcommand{\m}{\mathcal{X}}
\newcommand{\Pol}{\mathbf{P}}
\newcommand{\In}{\mathfrak{L}\,}

\makeatletter
\newcommand{\openbox}{\leavevmode
  \hbox to.77778em{%
  \hfil\vrule
  \vbox to.675em{\hrule width.6em\vfil\hrule}%
  \vrule\hfil}}

\newcommand{\proofname}{Proof.}
\newcounter{proof}%
\newenvironment{proof}[1][\proofname]{
  \th@nonumberplain
  \def\theorem@headerfont{\itshape}%
  \normalfont
  \@thm{proof}{proof}{#1}}%
  {\@endtheorem}
\makeatother

\theoremstyle{plain}
\theoremseparator{:}
\theoremheaderfont{\normalfont\sffamily\bfseries}
 \theorembodyfont{\itshape}

\newtheorem{thms}{Theorem}

\newtheorem{bounds}{Bound}

\newtheorem{lemm}{Lemma}

\newtheorem{propos}{Proposition}

\newtheorem{result}{Theorem}

\newtheorem{cor}{Corollary}

\newtheorem{definitions}{Definition}

 \theorembodyfont{\normalfont}

\newtheorem{remark}{Remark}

\newenvironment{defn}[1][]{
  \begin{definitions}[#1]\label{def:#1}%
}{\end{definitions}}

\newenvironment{thm}[1][]{
  \begin{thms}\label{thm:#1}%
}{\end{thms}}

\newenvironment{prop}[1][]{
  \begin{propos}%
}{\end{propos}}

\newenvironment{corollary}[1][]{
  \begin{cor}%
}{\end{cor}}

\newenvironment{lemma}[1][]{
  \begin{lemm}%
}{\end{lemm}}

\newenvironment{rem}[1][]{
  \begin{remark}%
}{\end{remark}}

\title{Recursive construction of\\ operator product expansion coefficients}

\author{Jan Holland$^{1}$\thanks{\tt holland@cpht.polytechnique.fr}\: and 
Stefan Hollands$^{2}$\thanks{\tt stefan.hollands@uni-leipzig.de}\:
\\ \\
{\it ${}^{1}$Centre de Physique Th\'eorique, CNRS, UMR 7644} \\
{\it  \'Ecole Polytechnique, F-91128 Palaiseau, France} \medskip \\
{\it ${}^{2}$Institut f\"ur Theoretische Physik, Universit\"at Leipzig
} \\
{\it Br\"uderstr. 16, Leipzig, D-04103, Germany
}  \\
}

\begin{document}
\maketitle

\begin{abstract}
We derive a novel formula for the derivative of operator product expansion (OPE) coefficients with respect to a coupling constant.  
The formula only involves the OPE coefficients themselves, and no further input, and is in this sense self-consistent. Furthermore, unlike other formal identities of this general nature 
in quantum field theory (such as the formal expression for the Lagrangian perturbation of a correlation function), our formula is completely well-defined from the start, i.e. 
requires no further UV-renormalization. This feature is a result of a cancelation of UV-divergences between various terms in our identity. Our proof, and an analysis of 
the features, of our identity is given for the example of massive, Euclidean $\varphi^4$ theory in 4 dimensional Euclidean space, and relies heavily on the framework of the renormalization group flow equations. It is valid to arbitrary, but finite orders in perturbation theory. The final formula, however, makes no explicit reference to the renormalization group flow, nor to perturbation theory, and we conjecture that it also holds non-perturbatively. The identity can be applied constructively because it gives recursive algorithm for the computation of OPE coefficients to arbitrary (finite) 
perturbation order in terms of the zeroth order coefficients corresponding to the underlying free field theory, which in turn are trivial to obtain. We briefly illustrate the relation of this 
method to more standard methods for computing the OPE in some simple examples.
\end{abstract}

\section{Introduction}

The \emph{operator product expansion} (OPE)~\cite{Wilson:1969ub, Zimmermann:1973wp} is the statement that any product of local quantum fields $\O_{A_{1}},\ldots, \O_{A_{N}}$ admits a short distance expansion of the form
\ben\label{OPEintro}
\O_{A_{1}}(x_{1})\cdots \O_{A_{N}}(x_{N}) \sim \sum_{B} \C_{A_{1}\ldots A_{N}}^{B}(x_{1},\ldots,x_{N})\, \O_{B}(x_{N})\, ,
\een
where $A_{1},\ldots,A_{N},B$ are labels for the composite fields of the theory under investigation. The \emph{OPE coefficients} $\C_{A_{1}\ldots A_{N}}^{B}(x_{1},\ldots,x_{N})$ are distributions with singularities located at the diagonals $x_{i}=x_{j}$ for $i\neq j$. The relation \eqref{OPEintro} is understood in the weak sense, i.e. as an insertion into an arbitrary quantum state (or into a correlation function), and it is further usually understood to be an \emph{asymptotic expansion}: If the sum on the right side is carried out to a sufficiently large but finite order, then the remainder goes to zero fast as $x_{i}\to x_{N}$ for all $i\in\{1,\ldots, N\}$. Within perturbative quantum field theory on Euclidean space, the expansion even converges~\cite{Hollands:2011gf,Holland:2012vw, Holland:2013we}. In 
the same framework, one can also prove that the OPE coefficients fulfill stringent consistency conditions that essentially express a kind of associativity law~\cite{Holland:2012vw, Holland:2013we}.  In this sense, 
the OPE should be viewed as encoding the ``algebraic'' content of a quantum field theory. 

The OPE is by now a well-established tool in quantum field theory, and has found many practical uses, for example in the analysis of asymptotically free non-abelian 
gauge theories in 4 dimensions, or in the analysis of conformal field theories, especially in 2 dimensions. Our main interest in the OPE arises from the fact that it may be viewed, in a sense, 
as the defining structure of a quantum field theory. This viewpoint becomes particularly compelling if one considers quantum field theory on general curved spacetime manifolds~\cite{Hollands:2008vq,Hollands:2008wr}. As is well-known~\cite{Wald199411}, there is in gerneral no preferred ``vacuum state'' on a general curved spacetime, and it is therefore not fruitful to formulate quantum field theory via entities depending explicitly
or implicitly on such a concept, such as correlation functions, S-matrices, or functional integrals. The main advantage of the OPE is precisely that it is independent of any arbitrary choices of states --
it holds when inserted into {\em any} (sufficiently regular) state. 

If we accept as a possible viewpoint that the OPE is a fundamental, defining structure of quantum field theory, then there should be ways to generate examples of quantum field theories via 
the direct construction of their OPE's. In 2 dimensional conformal field theories, such constructions of essentially algebraic nature are indeed possible, e.g . via the concept of 
``vertex algebras'', which essentially encode the OPE of a theory. A rich class of models can thereby be produced, see e.g.~\cite{FrenkelLepowskyMeurman198903,Lepowsky:2004vn,Huang:1997yt,Borcherds1986}. Unfortunately, direct constructions of this type do not 
seem possible in higher dimensions or for non-conformally invariant theories. In practice, one proceeds mostly by starting from a free field theory, and computes the coefficients perturbatively. The 
perturbative corrections are computed essentially by inserting the OPE into suitable matrix elements. While this procedure is consistent, it would be much more natural, from our viewpoint, 
to have a procedure for determining the OPE that does not depend on auxiliary matrix elements. 

One way to achieve this could be to focus attention directly the consistency relations that must be satisfied by the OPE coefficients~\cite{Holland:2012vw}. However, it is unclear how to handle these 
relations efficiently in order to generate solutions, either directly, or indirectly (e.g. via perturbation theory, symmetry arguments etc.). Instead, we will derive in the present paper a novel formula for the derivative of the OPE coefficients with respect to the coupling constant of the theory, which we present concretely for the $\varphi^4$-model in 4 dimensions. More precisely, our main result is
%
\begin{thm}[OPE deformation]
\label{thmint1}
The derivative of the OPE coefficients w.r.t. the coupling constant $g$ in massive Euclidean $\varphi^{4}$-theory with BPHZ renormalization conditions can be expressed as
\ben\label{eq:thmint1}
\begin{split}
& \partial_{g}\, \C_{A_{1}\ldots A_{N}}^{B}(x_{1},\ldots, x_{N})= \frac{-1}{4!}\int\d^{4}y \bigg[ \C_{\In A_{1}\ldots A_{N} }^{B}(y,x_{1},\ldots,x_{N}) \\
& \qquad -\sum_{i=1}^{N}\sum_{[C]\leq [A_{i}]} \!\!\!\C_{\In A_{i}}^{C}(y,x_{i})\, \C_{A_{1}\ldots \widehat{A_{i}}\,C\ldots A_{N}}^{B}(x_{1},\ldots, x_{N})- \sum_{[C]< [B]} \C_{A_{1}\ldots A_{N}}^{C}(x_{1},\ldots,x_{N})\, \C_{\In C }^{B}(y,x_{N}) \bigg]\, ,
\end{split}
\een
where the index $\In$ corresponds to the ''coupling operator'' $\O_{\In}=\varphi^{4}$ and where $\widehat{A_{i}}$ denotes omission of the corresponding index. $[A]$ stands 
for the dimension of the composite field $\O_A$. 
\end{thm}
We note the following implications of this formula:
\begin{description}
\item[Recursion scheme:] Equation \eqref{eq:thmint1} allows for a recursive construction of OPE coefficients to any order in perturbation theory: Expanding both sides of the equation as formal power series in the coupling constant, we note that, due to the derivative $\partial_{g}$ on the left hand side, we obtain a formula expressing coefficients of order $r\in\mathbb{N}$ in terms of lower order ones (see corollary \ref{cor} below). The only \emph{initial data} necessary for this construction are the zeroth-order OPE coefficients, which are quite easy to obtain.

\item[Features of the scheme:] The method of computation for OPE coefficients based on eq.\eqref{eq:thmint1} is quite different from customary methods, which generally rely on certain short distance/large momentum expansions  of particular Feynman diagrams~\cite{collins,Weinberg:1996kr} associaated with matrix elements of operator products. Our method, by contrast, is formulated entirely in terms of OPE coefficients, and is in this sense entirely self-consistent. In particular, the fundamental state-independence of these coefficients  is evident in this scheme, 
simply because no state enters the recursion formula. 

\item[Deformation theory:] In the language of ordinary, finite dimensional algebra, perturbations of the algebra product are usually referred to as \emph{deformations}. One can see in equation \eqref{OPEintro} that, formally, the OPE coefficients play a role similar to the structure constants of an algebra. In the light of this analogy, theorem \ref{thmint1} can be interpreted as a formula describing the \emph{deformation of the OPE algebra} caused by the $\varphi^{4}$-interaction. It would be very interesting to pursue this analogy further.

\item[Self-consistency:]  Unlike other hierarchies of equations, such as e.g. the Dyson-Schwinger equation, our scheme does not require an additional ``renormalization'' procedure, i.e. all aspects of the equation are completely well-defined from the outset! For example, one may suspect, based on dimensional analysis, that the $y$-integral would diverge when $y \to x_i$, and this would indeed be the case for the 
individual terms under the integral. However, as we will see, the divergences of the individual terms cancel out precisely, i.e. the integrand is an integrable function in the variable $y$! 
The same is shown to occur for potential divergences of the integral for large $y$. This supports our viewpoint that it should be possible to view the OPE as a fundamental property of quantum field theory. 

\item[Beyond perturbation theory:] Although our deformation formula \eqref{eq:thmint1} is derived within the framework of renormalized perturbation theory, its final form 
no longer makes any reference to perturbation theory. Eq.~\eqref{eq:thmint1} is a first order differential equation in $g$, and since $g=0$
corresponds to the free field theory, the ``initial values'' of all OPE-coefficients are known at $g=0$. Thus, it is conceivable that one could actually show that a unique solution to~\eqref{eq:thmint1} must exist (beyond the level of formal power series). We view this as a promising approach to a non-perturbative definition of the OPE-coefficients.  
\end{description}

It does not seem straightforward to motivate our  deformation formula \eqref{eq:thmint1} by formal methods such as path integrals. 
The first term on the right hand side resembles Lagrangian perturbation theory, since we essentially have an additional ``insertion'' of the interaction $\In$ which is integrated over the ``insertion point'' $y$. This is very similar to formulas for the Lagrangian perturbation of correlation functions, which are easy to motivate formally via a path integral. However, it seems unclear how to motivate the other terms by formal methods. 

Our derivation of eq.~\eqref{eq:thmint1} is, by contrast, 
mathematically rigorous, and relies on the renormalization group flow equation approach to quantum field theory~\cite{Polchinski:1983gv,Wilson:1971bg,Wilson:1971dh,Wegner:1972ih,Keller:1990ej,Keller:1991bz,Wetterich}. We will briefly review this framework in section \ref{sec:FE}, where we also define various quantities of relevance for the purpose of this paper. In section \ref{sec:proof} we give the proof of the theorem, followed in section \ref{sec:iteration} by an application of the recursion scheme mentioned above. In appendix \ref{APbounds} we derive bounds on Schwinger functions with operator insertions, which are used in the proof of theorem \ref{thmint1}. These estimates constitute a slight improvement over previously known bounds in the flow equation framework in the short distance regime, so the appendix might be a side result of some interest in its own right.

\paragraph{Notation and conventions:}

The convention for the Fourier transform in $\mr^4$ used in this paper is
\ben
f(x) = \int_p  \hat f(p)\, \e^{ipx} := \int_{\mr^4} \frac{\d^4 p}{(2\pi)^4}\, 
\e^{ipx} \hat f(p) \, .
\een
We use a
standard
multi-index notation.
Our multi-indices are elements $w = (w_1, \dots, w_n) \in \mn^{4n}$,
where each  $w_i \in \mn^4$
is a four-tuple with components $w_{i,\mu} \in \mn$
and $\mu=1,\dots,4$. For $f(p_{1},\ldots,p_{n})$ a smooth function on $\mr^{4n}$, we use the shorthand $f(\vec p)$ and we set
\ben\label{multider}
\pa^{w} f(\vec p) = \prod_{i,\mu}
\left( {\pa \over \pa p_{i,\mu}} \right)^{w_{i,\mu}} f(\vec p)
\een
and
\ben
w! = \prod_{i,\mu} w_{i,\mu}! \, , \quad |w|=\sum_{i,\mu} w_{i,\mu} \, .
\een
If a function $f(\vec{x};\vec{p})$ depends on two sets of variables, $(\vec{x},\vec{p})\in\mathbb{R}^{4n_{1}}\times\mathbb{R}^{4n_{2}}$, then we write $\partial_{\vec{p}}^{w}$ to indicate that the partial derivatives are taken with respect to the variables $(p_{1},\ldots,p_{n_{2}})$ as in \eqref{multider}. Derivatives $\partial^w$ of a product of
functions
$f_1 \cdots f_r$ are distributed over the factors using the Leibniz rule, which results
in the sum of all terms of the form $c_{\{v_i\}} \ \partial^{v_1} f_1
\cdots \partial^{v_r} f_r$. Here
each $v_i$ is now a $4n$-dimensional multi-index, where
$v_1+\ldots+v_r=w$,
and where
\ben
c_{\{v_i\}} = \frac{(v_1+\ldots+v_r)!}{v_1! \cdots v_r!} \le r^{|w|}
\een
is the associated multi-nomial weight factor. 

Given a set of momenta $(p_{1},\ldots, p_{n})\in\mathbb{R}^{4n}$, we agree on the shorthand notation
\ben\label{pshort}
\vec{p}:= (p_{1},\ldots, p_{n})\quad , \quad |\vec{p}|_{n}:= \sup_{J\subseteq \{1,\ldots, n\}}\, \Big|\sum_{i\in J} p_{i} \Big|\quad , \quad \vec{p}_{n+2}:= (\vec{p},k,-k)
\een
Later we will often simply write $|\vec{p}|$ instead of $|\vec{p}|_{n}$.

If $F(\varphi)$ is a differentiable function (in the Frechet space
sense)
of the
Schwartz space function $\varphi \in \S(\mr^4)$, we denote its
functional
derivative as
\ben
\frac{\d}{\d t} F(\varphi + t\psi) |_{t=0} = \int \d^4 x \
\frac{\delta F(\varphi)}{\delta \varphi(x)} \ \psi(x) \ ,
\quad \psi \in \S(\mr^4)\ ,
\een
where the right side is understood in the sense of distributions in
$\S'(\mr^4)$. Multiple functional derivatives are
denoted in a similar way and define in general distributions on
multiple Cartesian copies of $\mr^4$.

By $M_{N}$ we denote the spacetime domain
\ben\label{MN}
M_{N}:= \{ (x_{1},\ldots,x_{N})\in\mathbb{R}^{4N} \, |\, x_{i}\neq x_{j}\,\text{ for all } 1\leq i<j\leq N \}\, .
\een
We use the convention that $\mathbb{N}$ are the natural numbers without $0$, that $\mathbb{N}_{0}$ are the natural numbers including zero and that $\mathbb{R}_{+}$ are the real numbers greater than $0$. The scaling degree of a function $u \in C^\infty_0(M_N)$ at the diagonal is defined as
\ben\label{sddef}
\operatorname{sd} (u) :=\inf\Big\{p \in\mathbb{R}\, :\,  \lim_{\epsilon\to 0^+} \epsilon^{p} \ u(\epsilon x_1, \dots, \epsilon x_N) =0\, \text{ for all }(x_{1},\ldots,x_{N})\in M_{N}\Big\}\, ,
\een
where the limit is required to be uniform on compact sets $K \subset M_N$. We similarly define the scaling degree with respect to a subset of points $(x_1, \dots, x_M), M<N$, i.e. 
at a subdiagaonal $M_M \subset M_N$, which is denoted as $\operatorname{sd}_{\{1,...,M\}}$.

\section{The flow equation framework}\label{sec:FE}

The model studied in this paper is the hermitian scalar field theory with self-interaction $g \varphi^4$ and
mass $m > 0$ on flat 4-dimensional Euclidean space. The quantities of
interest in this (perturbative) quantum field theory will be defined
in this section via the flow equation (FE) method \cite{Polchinski:1983gv, Wegner:1972ih, Wilson:1971bg,Wilson:1971dh}. 
We will give a brief outline of the general formalism with a focus on objects of relevance to our study of the OPE, following closely~\cite{Hollands:2011gf}.
The original presentation of the particular method used here can be found in~\cite{Keller:1990ej}, and for more detailed reviews we refer the reader to \cite{Muller:2002he} and \cite{Kopper:1997vg} (in German).

Before we are ready to give the definition of the OPE coefficients $\C_{A_{1}\ldots A_{N}}^{B}$ in section \ref{sec:OPEcoefs}, we first introduce the basics of the flow equation approach in sections \ref{subsec:CAGs} and \ref{subsec:CAGint}. In sections \ref{sec:regCAG} and \ref{subsec:subdiv} we then discuss the regularization of short distance singularities of operator products within this approach and we define versions of Zimmermann's ''normal products'', which are intimately related to the operator product expansion.

\subsection{Connected amputated Green's functions (CAG's)}\label{subsec:CAGs}

We first formulate our quantum field theory with finite infrared (IR) and ultraviolet (UV) cutoffs, called $\Lambda$ and $\Lambda_{0}$ respectively, which can be removed in the end. 
 In the following, we always assume
\ben
\label{ka}
0<\Lambda\ , \quad  \sup(\La,m) < \Lao \ .
\een 
As we are dealing with a massive theory, an infrared cutoff is of course not actually necessary. It is introduced in the flow equation framework as a technical device, which will later allow us to derive the name-giving differential equations. The theory is defined in terms of
\begin{enumerate}

\item the propagator $C^{\Lambda,\Lambda_0}$, which reads in momentum space:
\ben\label{propreg}
C^{\La,\Lao}(p)\,:=\, {1 \over  p^2+m^2}
\left[ \exp \left(- {p^2+m^2 \over \Lao^2} \right) - \exp
\left(- {p^2+m^2 \over \La^2} \right) \right] 
\een
Removing the cutoffs corresponds to taking the limits $\La \to 0$ and $\Lao \to \infty$, which recovers the full propagator $1/(p^{2}+m^{2})$. Other choices of regularization than \eqref{propreg} are equally legitimate. The definition (\ref{propreg}) has the advantage of being analytic in $p^2$
for $\Lambda>0$. The propagator \eqref{propreg} defines a corresponding
Gaussian measure
$\mu^{\La,\Lao}$, whose covariance is $\hbar C^{\La,\Lao}$.
Here the factor $\hbar$ is introduced in order to obtain a consistent
\emph{loop expansion}\footnote{If one considers the usual Feynman diagram expansion of the quantities of interest defined below, then every closed loop  yields  a power of $\hbar$.} in the following.

\item the interaction Lagrangian, including renormalization counter terms (we also require the symmetry $\varphi\to-\varphi$, which causes odd powers of the basic field to vanish):
\ben
L^{\Lambda_0}(\varphi) = \int \d^4 x \ \bigg( a^{\Lambda_0}
\, \varphi(x)^2
+b^{\Lambda_0} \, \partial \varphi(x)^2+c^{\Lambda_0}
\, \varphi(x)^4 \bigg) 
\label{ac}
\een
Here the \emph{basic field} $\varphi \in \S(\mr^4)$ is any Schwartz space function. 
The counter terms
$a^{\Lambda_0} = O(\hbar),\  b^{\Lambda_0} = O(\hbar^2)$
and $c^{\Lambda_0} = \frac{g}{4!} +
O(\hbar)$ will be
adjusted--and actually diverge--when $\Lambda_0 \to \infty$,
in order to obtain a well
defined limit of the quantities of interest. This has been anticipated 
by making them ``running couplings'', i.e. functions of the ultra
violet cutoff $\Lambda_0$.
\end{enumerate}
The correlation ($=$ Schwinger- $=$ Green's- $=$ $n$-point-) functions of $n$ basic fields with
cutoff are defined by the expectation values
\ben\label{pathint}
\begin{split}
 \langle \varphi(x_1) \cdots \varphi(x_n) \rangle &\equiv  \mathbb{E}_{\mu^{\Lambda,\Lambda_{0}}} \bigg[\exp \bigg( -\frac{1}{\hbar}
L^{\Lambda_0}\bigg) \, \varphi(x_1) \cdots \varphi(x_n) \bigg] \bigg/ Z^{\Lambda,\Lambda_0} \\
& =
\int \d\mu^{\Lambda,\Lambda_0} \ \exp \bigg( -\frac{1}{\hbar}
L^{\Lambda_0}\bigg) \, \varphi(x_1) \cdots \varphi(x_n)\bigg/ Z^{\Lambda,\Lambda_0} \, .
\end{split}
\een
This expression is simply the standard Euclidean path-integral, but with the free part in the Lagrangian absorbed into the Gaussian measure $\d\mu^{\La,\Lao}$.
The normalization factor $Z^{\Lambda,\Lambda_0}$
is chosen so that $\langle 1 \rangle = 1$. To keep this factor finite one actually has to impose an additional volume cutoff, but the infinite volume limit can be taken without difficulty once we pass to perturbative
connected correlation functions, which we shall do in a moment.
For more details on this limit see~\cite{Kopper:2000qm,Muller:2002he}.
 The correct behavior of the running couplings $a^{\Lambda_{0}},b^{\Lambda_{0}},c^{\Lambda_{0}}$ is determined by deriving
first a differential equation in the parameter $\Lambda$ for the Schwinger functions, see eq.\eqref{fe},
and by then defining the couplings implicitly through the boundary
conditions for this equation given below in eqs.~\eqref{CAGBC1} and \eqref{CAGBC2}.

These differential equations, referred to from now on as flow equations, are written more conveniently in
terms of the hierarchy of ``connected, amputated
Green's functions'' (CAG's),
whose generating functional is given by the following 
convolution\footnote{The convolution is defined in general by
$(\mu^{\Lambda,\Lambda_0} \star F)(\varphi) =
\int \d\mu^{\Lambda,\Lambda_0}(\varphi') \ F(\varphi+\varphi')$.}
of the Gaussian measure with the exponentiated interaction,
\ben\label{CAGdef}
-L^{\Lambda, \Lambda_0} := \hbar \, \log \, \mu^{\Lambda,\Lambda_0}
\star \exp \bigg(-\frac{1}{\hbar} L^{\Lambda_0}
\bigg)- \hbar \log  Z^{\La,\Lao} \ .
\een
One can expand the functionals $L^{\Lambda,\Lambda_0}$ as formal power series in terms of Feynman diagrams with $l$ loops, $n$ external legs and propagator $C^{\Lambda,\Lambda_{0}}(p)$. One can show that, indeed, only connected diagrams contribute,
and the (free) propagators on the external legs are removed. While we will 
not use diagrammatic decompositions in terms of Feynman diagrams here, we will also analyze the functional (\ref{CAGdef})
in the sense of formal power series in $\hbar$ (''loop expansion''),
\ben\label{genfunc}
L^{\Lambda, \Lambda_0}(\varphi) := \sum_{n>0}^\infty
\sum_{l=0}^\infty {\hbar^l}
\int \d^4x_1 \dots \d^4 x_n\ \L^{\Lambda,\Lambda_0}_{n,l}(x_1, \dots, x_n)
\,
\varphi(x_1) \cdots \varphi(x_n) \, .
\een
No
statement
is made about the
convergence of the series in $\hbar$.

Translation invariance of the connected amputated functions in position space implies that their Fourier transforms, denoted
$\L^{\La,\Lao}_{n,l}(p_1, \dots, p_n)$, are supported
at $p_1+\ldots+p_n=0$. Therefore, we can write, by
abuse of notation
\ben
\L^{\Lambda,\Lambda_0}_{n,l}(p_1, \dots, p_n) = \delta^{4}{(\sum_{i=1}^n
p_i)}
\, \L^{\Lambda,\Lambda_0}_{n,l}(p_1, \dots, p_{n-1}) \, ,
\een
i.e. the momentum variable $p_{n}$ is determined in terms of the remaining $n-1$
independent momenta by momentum conservation. One can show that, as
functions of these remaining independent momenta, the connected
amputated
Green's functions
are smooth for $\La_{0}<\infty$, $\L^{\Lambda,\Lambda_0}_{n,l}(p_1, \dots, p_{n-1})
\in C^\infty(\mr^{4(n-1)})$.

To obtain the flow equations for the CAG's, we take the $\La$-derivative of
eq.\eqref{CAGdef}:
\ben
\partial_{\La} L^{\La,\Lao} \,=\,
\frac{\hbar}{2}\,
\langle\frac{\delta}{\delta \vp},\dot {C}^{\La}\star
\frac{\delta}{\delta \vp}\rangle L^{\La,\Lao}
\,-\,
\frac{1}{2}\, \langle \frac{\delta}{\delta
  \vp} L^{\La,\Lao},
\dot {C}^{\La}\star
\frac{\delta}{\delta \vp} L^{\La,\Lao}\rangle  +
\hbar \partial_\Lambda \log Z^{\La,\Lao} \ .
\label{fe}
\een
Here we use the following notation:
 We write $\,\dot {C}^{\La}\,$ for the derivative 
$\partial_{\La} {C}^{\La,\Lao}\,$, which, as we note,
does not depend on $\Lao$. Further, by $\langle\ ,\  \rangle$ we denote the standard scalar product in
$L^2(\mathbb{R}^4, \d^4 x)\,$, and $\star$ stands for
convolution in $\mr^4$. As an example, 
\ben
\langle\frac{\delta}{\delta \vp},\dot {C}^{\La}\star
\frac{\delta}{\delta \vp}\rangle = \int \d^4x\, \d^4y \ \dot {C}^{\La}(x-y)
\frac{\delta}{\delta \varphi(x)} \frac{\delta}{\delta \varphi(y)}
\een
is the ``functional Laplace operator''. 
%
The CAG's are defined uniquely as a solution to the differential equation \eqref{fe} only after we impose suitable boundary
conditions. These are\footnote{We restrict to BPHZ renormalization in this paper. Other choices are of course possible, and equally legitimate.}, using
the multi-index convention introduced above in ``Notations and Conventions'':
\ben\label{CAGBC1}
\partial^w_{\vec{p}} \L^{0,\Lambda_0}_{n,l}(\vec 0) = \de_{w,0}\ \de_{n,4}\
 \de_{l,0}\ \frac{g}{4!} \quad \text{for $n+|w|
\le 4$,}
\een
 as well as
\ben\label{CAGBC2}
\partial^w_{\vec{p}} \L^{\Lambda_0,\Lambda_0}_{n,l}(\vec p) = 0 \quad \text{for
$n+|w| > 4$.}
\een
Here $\delta_{a,b}$ is the Kronecker-delta. 
The  CAG's are then determined by integrating the
flow equations subject to these boundary conditions, see e.g.~\cite{Keller:1990ej,Muller:2002he}.

\subsection{Insertions of composite fields}\label{subsec:CAGint}

In the previous section we have defined Schwinger functions of products of the basic field. We now turn to the \emph{composite operators} (or ''composite fields''), which are given by the monomials
\ben\label{compop}
\O_{A}= \partial^{w_{1}}\varphi\cdots \partial^{w_{n}}\varphi\, , \quad A=\{n,w\}\, .
\een
Here $w=(w_{1},\ldots,w_{n})\in\mathbb{N}^{4n}$ is a multi-index (see also our notation and conventions section), and we denote the canonical dimension of such a field by
\ben
[A]:= n+\sum_{i}|w_{i}| \, .
\een
The Schwinger functions with insertions of composite operators are obtained by replacing the action $L^{\Lambda_0}$ with an
action containing additional sources, expressed through smooth functionals.
Particular examples of such functionals are {\em local} ones. Any
such local functional can by definition be
written as
\ben
F (\varphi)= \sum_A \int \d^4 x \  \O_A(x) \ f^A(x) \,\, ,
\quad f^A \in C^\infty_0(\mr^4) \, ,
\een
where the composite operators $\O_A$ are as in eq.~\eqref{compop} and
where the sum is finite. Recall that we may restrict attention to composite fields~\eqref{compop} with
an even number of factors of $\varphi$ as a result of our symmetry requirement $\varphi\to-\varphi$.
We now modify the action $L^{\Lambda_0}$ by adding sources $f^A$ as follows:
\ben
L^{\Lambda_0}\to L^{\Lao}_F:=L^{\Lambda_0}- F - \sum_{j=0}^\infty
B^{\Lambda_0}_j(\underbrace{F \otimes \cdots \otimes F}_j) 
\een
Here the last term represents the counter terms which are needed to eliminate the additional divergences arising from composite field insertions in the limit $\Lambda_0 \to \infty$. 
For each $j$ it is a linear functional\footnote{$C^\infty(\S(\mr^4))$ denotes the space
of smooth (in the Frechet sense) functionals. All our functionals are actually
formal power series in $\hbar$, so we should write more accurately $C^\infty(\S(\mr^4))[[\hbar]]$ for the space appearing below.}
\ben
B_j^{\Lambda_0}: [C^\infty(\S(\mr^4))]^{\otimes j}
\to C^\infty(\S(\mr^4))
 \  ,
\een
that is symmetric, and of order $O(\hbar)$. These functionals will be defined implicitly below through a flow 
equation and boundary conditions, see eqs.~\eqref{BCL1} and \eqref{BCL2}. To obtain the Schwinger functions with insertions of $r$~composite operators we now simply take functional
derivatives with respect to the sources, setting the sources
$f^{A_i} =0$ afterwards:
\ben
\langle \O_{A_1}(x_1) \cdots \O_{A_r}(x_r)  \rangle :=\hbar^r
\frac{\delta^r}{\delta f^{A_1}(x_1) \dots \delta f^{A_r}(x_r)}
\ (Z^{\La,\Lao})^{-1} \int \d\mu^{\Lambda,\Lambda_0}
\exp \bigg(-\frac{1}{\hbar}  L^{\Lambda_0}_F(\varphi)
\bigg)\biggr|_{ f^{A_i}=0} 
\een
Note that the CAG's discussed in the previous section are a special case
of this equation; there we
take $F = \int \d^4 x \ f(x) \ \varphi(x)$, and we have
$B^{\Lambda_0}_j(F^{\otimes j})=0$, because no extra counter terms
are required for this insertion. As above, we can define
a corresponding effective action
as
\ben\label{genfuncins}
-L^{\Lambda,\Lambda_0}_F := \hbar \, \log \, \mu^{\Lambda,\Lambda_0}
\star \exp \bigg(-\frac{1}{\hbar} ( L^{\Lambda_0}
- F - \sum_{j=0}^\infty B^{\Lambda_0}_j(F^{\otimes j})) \bigg)
- \log Z^{\La,\Lao}
\een
which now depends on the sources $f^{A_i}$, as well as on $\varphi$. From this modified effective action we determine the generating functionals of the CAG's with
$r$ operator insertions:
\ben
L^{\Lambda,\Lambda_0}(\O_{A_1}(x_1) \otimes \dots \otimes \O_{A_r}(x_r))
:=
\frac{\delta^r \ L^{\Lambda,\Lambda_0}_F}{\delta f^{A_1}(x_1) \dots \delta
  f^{A_r}(x_r)}
\,  \Bigg|_{f^{A_i} =  0} \, .
\een
The CAG's with insertions defined this way are multi-linear, as indicated by the
tensor product notation, and symmetric in the insertions. We can also expand the CAG's with insertions in $\varphi$ and $\hbar$ again (in momentum space):
\ben\label{genfunctinsert}
L^{\Lambda,\Lambda_0} \bigg( \bigotimes_{i=1}^r \O_{A_i}(x_i) \bigg) =\sum_{n,l \ge 0} {\hbar^l} \int \d^4p_1\dots \d^4p_n \
\L^{\Lambda,\Lambda_0}_{n,l}\bigg( \bigotimes_{i=1}^r \O_{A_i}(x_i);
p_1,
\dots, p_n \bigg)
\prod_{j=1}^n \hat{\vp}(p_j) 
\een
Due to the insertions in $\L_{n,l}^{\La,\Lao} (\otimes_j
\O_{A_j}(x_j),
\vec p)$, there is no restriction on the momentum set
$\vec p$ in this case. Translation invariance, however, implies that the CAG's 
with insertions at a translated set of points $x_j + y\,$ are
obtained
from those
at $y=0\,$ through multiplication by $\e^{iy\sum_{i=1}^n p_i}$, i.e.
\ben\label{CAGtrans}
\L^{\Lambda,\Lambda_0}_{n,l}\bigg( \bigotimes_{i=1}^r \O_{A_i}(x_i+y);
p_1,
\dots, p_n \bigg)=
\e^{iy\sum_{i=1}^n p_i}\, \L^{\Lambda,\Lambda_0}_{n,l}\bigg( \bigotimes_{i=1}^r \O_{A_i}(x_i);
p_1,
\dots, p_n \bigg)\, .
\een
Note also that only moments of CAG's with an even number $n$ are non-vanishing, again by our $\mathbb{Z}_{2}$-symmetry requirement.
The flow equation for the CAG's with insertions reads:
\ben\label{FEN}
\begin{split}
\partial_{\Lambda}L^{\Lambda,\Lambda_{0}}(\bigotimes_{i=1}^{N}\O_{A_{i}})=&\frac{\hbar}{2}\bra \varp \, ,\, \dot{C}^{\Lambda}\star \varp  \ket\, L^{\Lambda,\Lambda_{0}}(\bigotimes_{i=1}^{N}\O_{A_{i}})\\
-&\frac{1}{2}
\sum_{\atop{I_1 \cup I_2 = \{1,...,N \} } {I_{1}\cap I_{2}=\emptyset }}\bra \varp  L^{\Lambda,\Lambda_{0}}(\bigotimes_{i\in I_{1}}\O_{A_{i}}) \, ,\, \dot{C}^{\Lambda}\star \varp L^{\Lambda,\Lambda_{0}}(\bigotimes_{j\in I_{2}}\O_{A_{j}}) \ket\, ,
\end{split}
\een
In the second line it is understood that in the case $I=\emptyset$ we obtain the CAG's without insertions, i.e. $L^{\Lambda,\Lambda_{0}}(\bigotimes_{i\in I=\emptyset}\O_{A_{i}}):=L^{\Lambda,\Lambda_{0}}$. We also suppressed the coordinate space variables $(x_{1},\dots, x_{N})$ by writing $\O_{A_{i}}$ instead of $\O_{A_{i}}(x_{i})$. This convention will also be used regularly in the following for the sake of brevity. 

%
Note that the flow equation for the CAG's with $N\geq 2$ insertions involves inhomogeneities (called \emph{source terms} in the following) in the last line, which are quadratic in the CAG's with less than $N$ insertions. Therefore, we have to ascend in the number
of insertions if we want to integrate the flow equations \eqref{FEN}.
To complete the definition of the CAG's with insertions, we again have to specify boundary conditions on the corresponding flow equation. The simplest choice in the case of $N\geq 2$ insertions is
\ben\label{BCunsub}
\partial^{w}_{\vec{p}}\LscIr_{n,l}(\bigotimes_{i=1}^{N}\O_{A_{i}}(x_{i}); \vec{p})=0\quad \text{for all } w,n,l.
\quad
\een
For CAG's with one insertion we choose again BPHZ renormalization conditions\footnote{ See~\cite{Keller:1991bz,Keller:1992by} for a more detailed motivation of these boundary conditions. It should be mentioned that our definition of the functionals $L^{\Lambda,\Lambda_{0}}(\O_{A})$ differs from the one given in those papers by a minus sign. }
\ben\label{BCL1}
\partial^{w}_{\vec{p}}\LscO_{n,l}(\O_{A}(0); \vec{0})= i^{|w|}w! \delta_{w,w'}\delta_{n,n'}\delta_{l,0} \quad \text{ for }n+|w|\leq [A]
\een
\ben\label{BCL2}
\partial^{w}_{\vec{p}}\LscIr_{n,l}(\O_{A}(0); \vec{p})=0\quad \text{ for }n+|w|>[A] \quad .
\een
Although the connected amputated Green's functions (CAG's) with insertions can be used as the basic building blocks of the correlation functions, it will turn out to be useful to also consider certain non-connected versions of these, called ''AG's with insertions'' in the following. They are defined as
\ben\label{GD}
G^{\Lambda,\Lambda_{0}}(\bigotimes_{i=1}^{N}\O_{A_{i}})
:=\sum_{\alpha=1}^{N} \sum_{\atop{ I_{1}\cup \ldots \cup I_{\alpha}=\{1,\ldots, N\} }{I_{i}\neq \emptyset, I_{i}\cap I_{j}=\emptyset} }\prod_{i=1}^{\alpha}\,(-\hbar)^{N-\alpha} L^{\Lambda, \Lambda_{0}}(\bigotimes_{j\in I_i}\O_{A_{j}})\, .
\een
Note that the case $N=1$ just reduces to the CAG's with one insertion, i.e. $G^{\Lambda,\Lambda_{0}}(\O_{A})=L^{\Lambda,\Lambda_{0}}(\O_{A})$. Again, we also consider the expanded quantities in $\hbar$ and $\hat\varphi$; these are denoted in
the present case as ${\mathcal G}^{\La,\Lao}_{n,l}(\otimes_{i=1}^N \O_{A_i}, \vec p)$, where
as usual, $l$ indicates the power of $\hbar$, and $n$ the power of $\hat\varphi$. As the name suggests, these are the
amputated versions of the Schwinger (=Green's) functions\footnote{Strictly speaking, the functionals $G^{\Lambda,\Lambda_{0}}(\bigotimes_{i=1}^{N}\O_{A_{i}})$ do not generate \emph{all} the amputated Feynman diagrams with operator insertions, since connected pieces without any operator insertion are excluded, see also eq.\eqref{AGGreens}. For lack of a better name, we will however continue to refer to these functionals as amputated Green's functions with insertions by a slight abuse of language. },
\ben\label{AGGreens}
\begin{split}
&\Big\bra\prod_{i=1}^N \O_{A_i}(x_i) \  \prod_{j=1}^n \hat \varphi(p_j) \Big\ket
\ \prod_{k=1}^n (C^{\Lambda,\Lambda_{0}}(p_{k}))^{-1}
\\
&= \sum_{j=1}^{n}\sum_{\substack{I_{1}\cup\ldots\cup I_{j}=\{1,\ldots, n\}\\ I_{i}\cap I_{j}=\emptyset \\ l_{1}+\ldots+l_{j}=l\geq 0 }}\, \hbar^{n+l+1-j}\,  {\mathcal G}^{\La,\Lao}_{|I_{1}|,l_{1}}(\otimes_{i=1}^N \O_{A_i}(x_i), \vec p_{I_{1}}) \, \bar{\L}_{|I_{2}|,l_{2}}^{\Lambda,\Lambda_{0}}(\vec{p}_{I_{2}}) \cdots  \bar{\L}_{|I_{j}|,l_{j}}^{\Lambda,\Lambda_{0}}(\vec{p}_{I_{j}})
\end{split}
\een
where $\bar\L^{\La,\Lambda_{0}}_{n,l}$ are the expansion coefficients of the generating functional $\bar{L}^{\La,\Lambda_{0}}(\varphi)=-L^{\La,\Lambda_{0}}(\varphi)+\frac{1}{2}\bra \varphi,\, (C^{\La,\Lambda_{0}})^{-1}\star\varphi \ket$ without the momentum conservation delta functions taken out.

By contrast to the CAG's, the AG's satisfy linear {\em homogeneous} flow equations,
\ben\label{GFE}
\partial_{\Lambda}G^{\Lambda,\Lambda_{0}}(\bigotimes_{i=1}^{N}\O_{A_{i}})=\frac{\hbar}{2}\bra \varp , \dot{C}^{\Lambda}\star\varp \ket\, G^{\Lambda,\Lambda_{0}}(\bigotimes_{i=1}^{N}\O_{A_{i}})- \bra \varp G^{\Lambda,\Lambda_{0}}(\bigotimes_{i=1}^{N}\O_{A_{i}}), \dot{C}^{\Lambda}\star\varp L^{\Lambda,\Lambda_{0}} \ket \ .
\een
This property is a welcome simplification, which is unfortunately
counterbalanced by the fact that the boundary conditions for the AG's are more complicated.
Therefore, as a compromise between simple flow equation and simple boundary conditions,  we will not work with the full AG's in the following, but instead define the slightly modified objects
 \ben\label{GTild}
\hbar F^{\Lambda,\Lambda_{0}}(\bigotimes_{i=1}^{N}\O_{A_{i}}):=G^{\Lambda,\Lambda_{0}}(\bigotimes_{i=1}^{N}\O_{A_{i}})- \prod_{i=1}^{N} L^{\Lambda,\Lambda_{0}}(\O_{A_{i}})\, .
\een
Using the definitions of the CAG's given above, these functionals are seen to obey the flow equation 
\ben\label{GTildFE}
\begin{split}
\partial_{\Lambda}F^{\Lambda,\Lambda_{0}}(\bigotimes_{i=1}^{N}\O_{A_{i}})=&\frac{\hbar}{2}\bra \varp , \dot{C}^{\Lambda}\star\varp \ket\, F^{\Lambda,\Lambda_{0}}(\bigotimes_{i=1}^{N}\O_{A_{i}})- \bra \varp F^{\Lambda,\Lambda_{0}}(\bigotimes_{i=1}^{N}\O_{A_{i}}), \dot{C}^{\Lambda}\star\varp L^{\Lambda,\Lambda_{0}} \ket\\
+&\sum_{1\leq i<j\leq N} \bra \varp L^{\Lambda,\Lambda_{0}}(\O_{A_{i}}), \dot{C}^{\Lambda}\star\varp L^{\Lambda,\Lambda_{0}}(\O_{A_{j}})   \ket \prod_{r\in\{1,\ldots,N\}\setminus\{i,j\}} L^{\Lambda,\Lambda_{0}}(\O_{A_{r}})
\end{split}
\een
and the trivial boundary conditions
\ben\label{BCFbare}
\partial^{w}_{\vec{p}} \F^{\Lambda_{0},\Lambda_{0}}_{n,l}(\bigotimes_{i=1}^{N}\O_{A_{i}}; \vec{p})=0\quad \text{for all }n,l,w,
\een
with a calligraphic letter $\F^{\La,\Lao}_{n,l}$ denoting as usual the objects appearing in the expansion
of $F^{\La,\Lao}$ in powers of $\hbar, \varphi$. In terms of Feynman diagrams we may interpret these functionals as follows: As mentioned above, the $G$-functionals correspond to the (not necessarily connected) amputated Feynman graphs with $N$ extra vertices corresponding to the operator insertions. On the other hand, the $F$-functionals correspond to the subset of these diagrams where \emph{at least two} of the operator insertions belong to the same connected component of the graph. 
Like the CAG's with multiple insertions, the $F$-functionals are divergent on the partial diagonals (when $\Lambda_0 \to \infty$), i.e. whenever two or more spacetime arguments coincide. Since the CAG's with one insertion are smooth in the spacetime argument [see equation \eqref{CAGtrans}], the decomposition \eqref{GTild} separates the contributions to $G$ which are regular in the spacetime arguments from those which are singular at short distances. We also note that translation invariance again implies
\ben\label{Ftrans}
\F^{\Lambda,\Lambda_{0}}_{n,l}(\bigotimes_{i=1}^{N}\O_{A_{i}}(x_{i}); p_{1},\ldots,p_{n})=\e^{i y (p_{1}+\ldots+p_{n})}\F^{\Lambda,\Lambda_{0}}_{n,l}(\bigotimes_{i=1}^{N}\O_{A_{i}}(x_{i}-y); p_{1},\ldots,p_{n})\, .
\een

\subsection{Regularization of Schwinger functions with insertions}\label{sec:regCAG}

The purpose of introducing a UV-cutoff is that as long as we keep $\Lao$ finite, the CAG's with insertions depend smoothly on the points
$x_1, \dots, x_N$, as well as on the momenta $p_1, \dots, p_n$. In the limit $\Lao \to \infty$, smoothness in
the $x_i$'s however is lost, and the CAG's develop singularities for configurations such that some of the points $x_i$
coincide.  This is of course not a problem, nor unexpected--the Green's functions in quantum field theory are usually singular for coinciding points--reflecting the singular nature of the operators themselves. 
 In the following we will discuss certain \emph{regularized} (sometimes also called \emph{oversubtracted}) versions of the Green's functions with insertions defined in the previous section, which possess a higher degree of regularity in the spacetime arguments. As we will see later, these regularized Green's functions with insertions play a crucial role in the definition and application of the operator product expansion. Similar objects were defined by Zimmermann under the name \emph{normal products} of quantum fields in the diagrammatic approach to perturbation theory~\cite{zimbrand}.
 
A method for improving regularity of Green's functions with operator insertions was developed, in the context of the present framework, in \cite{Keller:1992by} (up to two insertions), \cite{Holland:2012vw} (3 insertions) and~\cite{Holland:2013we} ($N$ insertions). For the purposes of the present paper, the following definition will be useful:
\begin{defn}[Regularized AG's]\label{def:GDfct}
The \emph{amputated Green's functions (AG's) with insertions and regularization} are defined for any $D\geq -1$ as
 \ben\label{GDdef}
{G}_{D}^{\Lambda,\Lambda_{0}}(\bigotimes_{i=1}^{N}\O_{A_{i}}):=
\hbar F_{D}^{\Lambda,\Lambda_{0}}(\bigotimes_{i=1}^{N}\O_{A_{i}})+\prod_{i=1}^{N} L^{\Lambda,\Lambda_{0}}(\O_{A_{i}})\ ,
\een
where the functionals $F^{\Lambda,\Lambda_{0}}_{D}$ are required to satisfy the flow equation \eqref{GTildFE} and the boundary conditions
\ben\label{Gbound1}
\partial^{w}_{\vec{p}}\F^{0,\Lambda_{0}}_{D,n,l}(\bigotimes_{i=1}^{N}\O_{A_{i}}(x_{i}); \vec{0})\Big|_{x_{N}=0}=0\quad \text{ for }n+|w|\leq D
\een
\ben\label{Gbound2}
\partial^{w}_{\vec{p}}\F^{\Lambda_{0},\Lambda_{0}}_{D,n,l}(\bigotimes_{i=1}^{N}\O_{A_{i}}(x_{i}); \vec{p})\Big|_{x_{N}=0}=0\quad \text{ for }n+|w|>D\, .
\een
\end{defn}
Evidently, $F_{D=-1}^{\Lambda,\Lambda_{0}}(\bigotimes_{i=1}^{N}\O_{A_{i}})=F^{\Lambda,\Lambda_{0}}(\bigotimes_{i=1}^{N}\O_{A_{i}})$ are the functionals without regularization. Note that in the $N=2$ case, $F_{D}$ reduces to the CAG with two insertions, i.e.
\ben\label{CAGFeq}
F^{\Lambda,\Lambda_{0}}_{D}(\O_{A}(x)\otimes\O_{B}(0))=- L^{\Lambda,\Lambda_{0}}_{D}(\O_{A}(x)\otimes\O_{B}(0))
\een
since both sides of the equation share the same flow equation and boundary conditions. For $N\geq 3$, however, such a simple relation does not seem to exist.

\paragraph{Properties of regularized AG's with insertions:} 
It follows from the bounds derived in corollary \ref{corbound} (see page \pageref{corbound}) that, up to any order in the loop expansion (i.e. expansion in $\hbar$), the scaling degree (see
eq.~\eqref{sddef}) of the regularized AG's with insertions, is bounded by\footnote{The same result is also 
true for the functionals $F^{\Lambda,\Lambda_{0}}_{D}(\otimes_{i=1}^{N}\O_{A_{i}})$.}
\ben\label{SD1}
\operatorname{sd}(G_{D}^{\Lambda,\Lambda_{0}}(\bigotimes_{i=1}^{N}\O_{A_{i}}))\leq [A_{1}]+\ldots+[A_{N}]-D-1\, ,
\een 
for all $\Lambda,\Lambda_0$, including the case of main interest, $\Lambda_0=\infty, \Lambda=0$, i.e. upon removal of the cutoffs. 
This confirms the role of $D$ as a regularization parameter for the singularity on the total diagonal $x_{1}=\ldots=x_{N}$. The result is consistent with the common opinion that Schwinger functions with insertions $\O_{A_{1}}(\delta x_{1}),\ldots, \O_{A_{N}}(\delta x_{N})$ and without additional regularization (i.e. $D=-1$) should scale as $\delta^{-([A_{1}]+\ldots+[A_{N}])}\cdot (\log\delta)^{n}$ as $\delta\to 0$ (uniformly as $\Lambda_0 \to \infty$). To our knowledge, this property has not been rigorously established in the present framework before\footnote{In~\cite{Keller:1992by} the scaling behavior \eqref{SD1} has been anticipated, but only informal arguments for its validity were given. }, so the result may be of interest in its own right. 

Concerning the infrared (i.e. large distance) behavior of the functionals $F_{D}^{\Lambda,\Lambda_{0}}(\bigotimes_{i=1}^{N}\O_{A_{i}})$, we will show below in theorem~\ref{thmbound} (see page \pageref{thmbound}) that for
 any $R\in\mathbb{R_{+}}$ 
and any $s\in\mathbb{N}_{0}$
\ben\label{FIR}
\begin{split}
|\F_{D,2n,l}^{\Lambda,\Lambda_{0}}(\bigotimes_{i=1}^{N}\O_{A_{i}}(R x_{i});\vec{p})| &\leq  R^{D-D'-s}\cdot   (m+\La)^{D-s-2n}\,   \Pol_{1}\left(\log\frac{\Lambda+m}{m}\right) \Pol_{2}\left(\frac{|\vec{p}|}{\Lambda+m}\right) \\
&\times \frac{  \max\limits_{1\leq i\leq N}|x_{i}|^{D+1} \cdot (m \min\limits_{1\leq i<j\leq N} |x_{i}-x_{j}|)^{-s} }{\min\limits_{1\leq i<j\leq N}|x_{i}-x_{j}|^{D'+1}}\, ,
\end{split}
\een
where $\Pol_{i}(p)$ are polynomials in $p$ with positive coefficients. In other words, if we scale the spacetime arguments by a large factor $R$, then the $F^{\La,\Lambda_{0}}$-functionals decay more rapidly than any inverse power of $R$. This property is of course a consequence of the fact that we are dealing with a massive theory.

We also note for later that the amputated Green's functions with operator insertions satisfy versions of the so called \emph{Lowenstein rules}:
\ben\label{Low1}
\partial_{x}^{v}\, L^{\Lambda,\Lambda_{0}}(\O_{A}(x)) = L^{\Lambda,\Lambda_{0}}(\partial_{x}^{v}\O_{A}(x))
\een
\ben
\partial_{x_{a}}^{v}\, G^{\Lambda,\Lambda_{0}}_{D}(\otimes_{i=1}^{N}\O_{A_{i}}(x_{i}))\Big|_{x_{N}=0} = G^{\Lambda,\Lambda_{0}}_{D}\left(\partial_{x_{a}}^{v}\otimes_{i=1}^{N}\O_{A_{i}}(x_{i})\right)\Big|_{x_{N}=0} , \quad 1\leq a <  N
\een
\ben\label{Low2}
\left(\partial_{x_{1}}+\ldots+\partial_{x_{N}} \right)^{v} G^{\Lambda,\Lambda_{0}}_{D}(\otimes_{i=1}^{N}\O_{A_{i}}(x_{i})) = G^{\Lambda,\Lambda_{0}}_{D+|v|}\left(\left(\partial_{x_{1}}+\ldots+\partial_{x_{N}} \right)^{v}\otimes_{i=1}^{N}\O_{A_{i}}(x_{i})\right) ,
\een
where $v\in\mathbb{N}^{4}$. These relations can be verified by comparing the flow equation and boundary conditions satisfied by either side of the respective equation (see~\cite{Keller:1991bz,Keller:1992by,Hollands:2011gf,Holland:2013we} for more details).

\subsection{Regularization of subdivergences}\label{subsec:subdiv}

In the previous section we have outlined a procedure that allows us to improve the total scaling degree of the amputated Green's functions ${G}^{\Lambda,\Lambda_{0}}(\bigotimes_{i=1}^{N}\O_{A_{i}})$. In other words, we are able to control the singular behavior of these functionals with respect to the total diagonal $x_{1}=\ldots=x_{N}$.  Below, however, we want to remove only divergences associated to the partial diagonals of a subset of the spacetime arguments $x_{1},\ldots,x_{N}$. In the present section we will define this regularization of subdivergences. 

It is a priori far from clear how to generalize the strategy of the previous section to subdivergences. The following lemma provides a decomposition of the AG's that will be helpful for this purpose:

\begin{lemma}\label{lemdecomp}
For any $N\geq2$ and $M<N$ the following decomposition holds:
\ben\label{lemdecompeq}
\begin{split}
{G}^{\Lambda,\Lambda_{0}}(\bigotimes_{i=1}^{N}\O_{A_{i}})&={G}^{\Lambda,\Lambda_{0}}(\bigotimes_{i=1}^{M}\O_{A_{i}})\ {G}^{\Lambda,\Lambda_{0}}(\bigotimes_{i=M+1}^{N}\O_{A_{i}})\\
&+\hbar  H^{\Lambda,\Lambda_{0}}(\bigotimes_{i=1}^{M}\O_{A_{i}} ; \bigotimes_{i=M+1}^{N}\O_{A_{i}})
\end{split}
\een
Here the functionals $H^{\Lambda,\Lambda_{0}}$ are defined through the flow equation
\ben\label{HFE}
\begin{split}
\partial_{\Lambda}H^{\Lambda,\Lambda_{0}}(\bigotimes_{i=1}^{M}\O_{A_{i}} ; \bigotimes_{i=M+1}^{N}\O_{A_{i}})=\frac{\hbar}{2}\bra \varp , \dot{C}^{\Lambda}\star\varp \ket\, H^{\Lambda,\Lambda_{0}}(\bigotimes_{i=1}^{M}\O_{A_{i}} ; \bigotimes_{i=M+1}^{N}\O_{A_{i}})&\\
- \bra \varp H^{\Lambda,\Lambda_{0}}(\bigotimes_{i=1}^{M}\O_{A_{i}} ; \bigotimes_{i=M+1}^{N}\O_{A_{i}}), \dot{C}^{\Lambda}\star\varp L^{\Lambda,\Lambda_{0}} \ket&\\
+\bra \varp {G}^{\Lambda,\Lambda_{0}}(\bigotimes_{i=1}^{M}\O_{A_{i}}), \dot{C}^{\Lambda}\star\varp {G}^{\Lambda,\Lambda_{0}}(\bigotimes_{i=M+1}^{N}\O_{A_{i}})   \ket &
\end{split}
\een
and the boundary conditions
\ben\label{BCH}
\partial^{w}_{\vec{p}} \H^{\Lambda_{0},\Lambda_{0}}_{n,l}(\bigotimes_{i=1}^{M}\O_{A_{i}} ; \bigotimes_{i=M+1}^{N}\O_{A_{i}};\vec{p})=0\quad \text{for all }n,l,w.
\een
\end{lemma}
%
\begin{rem}
The lemma can also be understood diagrammatically. On the l.h.s. of equation \eqref{lemdecompeq} we have all (amputated) diagrams with insertion of extra vertices corresponding to the composite operators $\O_{A_{1}},\ldots,\O_{A_{N}}$. The first term on the r.h.s. stands for the factorized contributions, where the diagrams containing the $\O_{A_{1}},\ldots,\O_{A_{M}}$ vertices are disconnected from the diagrams containing the $\O_{A_{M+1}},\ldots,\O_{A_{N}}$ vertices. The second term on the r.h.s. then contains all contributions where at least one pair of vertices $\O_{A_{i}},\O_{A_{j}}$ with $1\leq i\leq M<j\leq N$ belong to the same connected diagram.
\end{rem}
To prove lemma \ref{lemdecomp} one simply checks that both sides of eq.\eqref{lemdecompeq} satisfy the same flow equation and boundary conditions. The explicit proof can be found in~\cite{Holland:2013we}.

The decomposition provided by lemma \ref{lemdecomp} suggests the following definition:
\begin{defn}[Partially regularized AG's]\label{def:AGpart} Let $M<N\geq 2$.
We denote the amputated Green's functions with operator insertions $\O_{A_{1}}(x_{1}),\ldots,\O_{A_{N}}(x_{N})$, regularized to degree $D\leq [A_{1}]+\ldots+[A_{M}]$ in the coordinates $x_{1},\ldots,x_{M}$, by $G^{\Lambda,\Lambda_{0}}([\bigotimes_{i=1}^{M}\O_{A_{i}}]_{D};\bigotimes_{M+1}^{N}\O_{A_{i}})$. These functionals are defined as
\ben\label{AGpartdef}
\begin{split}
G^{\Lambda,\Lambda_{0}}\left([\bigotimes_{i=1}^{M}\O_{A_{i}}]_{D}; \bigotimes_{M+1}^{N}\O_{A_{i}}\right)&:={G}_{D}^{\Lambda,\Lambda_{0}}(\bigotimes_{i=1}^{M}\O_{A_{i}})\ {G}^{\Lambda,\Lambda_{0}}(\bigotimes_{i=M+1}^{N}\O_{A_{i}})\\
&+\hbar H^{\Lambda,\Lambda_{0}}\left([\bigotimes_{i=1}^{M}\O_{A_{i}}]_{D} ; \bigotimes_{i=M+1}^{N}\O_{A_{i}}\right)
\end{split}
\een
where $H^{\Lambda,\Lambda_{0}}\left([\bigotimes_{i=1}^{M}\O_{A_{i}}]_{D} ; \bigotimes_{i=M+1}^{N}\O_{A_{i}}\right)$ is defined through the flow equation \eqref{HFE} with $G^{\Lambda,\Lambda_{0}}(\bigotimes_{i=1}^{M}\O_{A_{i}})$ replaced by $G_{D}^{\Lambda,\Lambda_{0}}(\bigotimes_{i=1}^{M}\O_{A_{i}})$, subject to the boundary conditions \eqref{BCH}. 
\end{defn}
\paragraph{Properties of partially regularized AG's:}
Using an inductive scheme based on the flow equations, we show in appendix \ref{appartB} that the parameter $D$ does indeed allow us to improve regularity on the partial diagonal $x_{1}=\ldots=x_{M}$, while the behavior on the other diagonals remains unaffected. More precisely, it follows from corollary \ref{corboundH} (on page \pageref{corboundH}) that the 
scaling degree (see eq.~\eqref{sddef}) of these functionals is bounded by
\ben\label{SDpartbad}
\operatorname{sd}\left(G^{\Lambda,\Lambda_{0}}\left([\bigotimes_{i=1}^{M}\O_{A_{i}}]_{D};\bigotimes_{M+1}^{N}\O_{A_{i}}\right)\right)\leq [A_{1}]+\ldots+[A_{N}]\, ,
\een 
including at $\Lambda_0=\infty, \Lambda=0$, i.e. upon removal of the cutoffs.
This is the same estimate as for the AG's without regularization [see \eqref{SD1}]. On the other hand, the scaling degree with respect to the partial diagonal $x_{1}=\ldots=x_{M}$
satisfies the bound (see again corollary \ref{corboundH})
\ben\label{SDpart}
\operatorname{sd}_{\{1,\ldots, M\}}\left(G^{\Lambda,\Lambda_{0}}\left([\bigotimes_{i=1}^{M}\O_{A_{i}}]_{D};\bigotimes_{M+1}^{N}\O_{A_{i}}\right)\right)\leq [A_{1}]+\ldots+[A_{M}]-D-1\, ,
\een
which confirms the role of $D$ as a regularization parameter with respect to the partial diagonal $x_{1}=\ldots=x_{M}$. We finally also note that the $H^{\Lambda,\Lambda_{0}}$-functionals decay rapidly for large separation of the set of points $x_{1},\ldots,x_{M}$ from the set of points $x_{M+1},\ldots, x_{N}$, i.e. we have for $|y|\gg \max_{i}|x_{i}|$
\ben\label{Hinfrared}
\begin{split}
 &\left| \H_{2n,l}^{\Lambda,\Lambda_{0}}\left([\bigotimes_{i=1}^{M}\O_{A_{i}}(x_{i}+y)]_{D} ; \bigotimes_{i=M+1}^{N}\O_{A_{i}}(x_{i});\vec{p}\right) \right|\\
 & \leq \min_{1\leq i\leq M<j\leq N}|x_{i}+y-x_{j}|^{-D-s}\cdot (\La+m)^{-2n-|w|-s-1}\,  \bar{\m}(\vec{x})\, \Pol_{1}\left(\log\frac{\Lambda+m}{m}\right) \Pol_{2}\left(\frac{|\vec{p}|}{\Lambda+m}\right)
\end{split}
\een
for any $s\in\mathbb{N}_{0}$, which follows from theorem \ref{thmboundH} (see page \pageref{thmboundH}). Here $\bar{\m}(\vec{x})$ is again a rational function in the distances $|x_{i}-x_{j}|$, where $1\leq i <j\leq N$, (but independent of $y$) and $\Pol_{i}(p)$ are polynomials in $p$ with positive coefficients (see theorem \ref{thmboundH} for more details). Finally, we also note that due to translation invariance
\ben\label{Htrans}
\begin{split}
&\H_{2n,l}^{\Lambda,\Lambda_{0}}\left([\bigotimes_{i=1}^{M}\O_{A_{i}}(x_{i})]_{D} ; \bigotimes_{i=M+1}^{N}\O_{A_{i}}(x_{i});\vec{p}\right)\\
&= \e^{i y (p_{1}+\ldots+p_{2n})}\H_{2n,l}^{\Lambda,\Lambda_{0}}\left([\bigotimes_{i=1}^{M}\O_{A_{i}}(x_{i}-y)]_{D} ; \bigotimes_{i=M+1}^{N}\O_{A_{i}}(x_{i}-y);\vec{p}\right)
\end{split}
\een
holds.

\subsection{OPE coefficients}\label{sec:OPEcoefs}

We are now ready to give the definition of the  OPE coefficients in the present framework. To have a more compact notation, let us define the operator $\D^{A}$ acting on differentiable functionals $F(\varphi)$ of Schwartz space functions $\varphi\in\S(\mathbb{R}^{4})$ by
\ben\label{defD}
\D^{A} F(\varphi) = \left. \frac{(-i)^{|w|}}{n!\,w!}\, \partial_{\vec{p}}^{w}\frac{\delta^{n}}{\delta\hat\varphi(p_{1})\cdots\delta\hat\varphi(p_{n})}\, F(\varphi)\, \right|_{\hat\varphi=0, \vec{p}=0}\quad ,
\een
where $A=\{n,w\}$. Further, let us also define the multivariate Taylor expansion operator 
through
\ben\label{defT}
\T^{j}_{\vec{x}\to\vec{y}}\, f(\vec{x})=\T^{j}_{(x_{1},\ldots,x_{N})\to(y_{1},\dots, y_{N})} \, f(x_{1},\ldots, x_{N})=\sum_{|w|=j}\, \frac{(\vec{x}-\vec{y})^{w}}{w!}\, \partial^{w}f(\vec{y})
\een
where $\vec{x}=(x_{1},\ldots,x_{N})$ and
where $f$ is a sufficiently smooth function on $\mathbb{R}^{4N}$. For expansions around zero we will use the shorthand $\T^{j}_{\vec{x}\to\vec{0}}=:\T^{j}_{\vec{x}}$\ .
Then the OPE coefficients are defined as follows~\cite{Keller:1992by,Hollands:2011gf,Holland:2012vw}:
\begin{defn}[OPE coefficients]\label{defOPE}
Let $\Delta:=[B]-([A_{1}]+\ldots+[A_{N}])$. The OPE coefficients are defined in terms of the regularized AG's with insertions as
\ben\label{OPEhigh}
\C_{A_{1}\ldots A_{N}}^{B}(x_{1},\ldots,x_{N-1},0)\,
:=\, \D^{B}\left\{ G^{0,\Lambda_{0}}_{[B]-1}\left((1-\sum_{j < \Delta}\T^{j}_{\vec{x}}\,) \bigotimes_{i=1}^{N}\O_{A_{i}}(x_i)\right) \right\}\, ,
\een
 where it is understood that $x_N=0$.
\end{defn}
\begin{rem}
Note that the OPE coefficients are translation invariant, so we may e.g. put the last point
to zero by a translation, as we have done above to get a simpler formula.
\end{rem}
To provide some motivation for this definition, we note that the remainder of the operator product expansion can be conveniently expressed in the form
\ben\label{OPER}
\begin{split}
& \Big| \Big\bra \O_{A_{1}}(x_{1})\cdots\O_{A_{N}}(x_{N})\, \varphi(f_{p_{1}})\cdots\varphi(f_{p_{n}}) \Big\ket- \hspace{-.3cm}\sum_{[C]-D'\leq \Delta}\hspace{-.3cm} \C_{A_{1}\ldots A_{N}}^{C}(x_{1},\ldots,x_{N})\ \Big\bra \O_{C}(x_{N})\, \hat{\varphi}(p_{1})\cdots\hat{\varphi}(p_{n}) \Big\ket  \Big|\\
 &=\hspace{-.3cm}\sum_{\substack{I_{1}\cup\ldots\cup I_{j}=\{1,\ldots, n\}\\ I_{i}\cap I_{j}=\emptyset \\ l_{1}+\ldots+l_{j}=l  }} \hspace{-.5cm}\hbar^{n+l+1-j}  
 \R^{\La,\Lao}_{D,|I_{1}|,l_{1}}(\bigotimes_{i=1}^{N}\O_{A_{i}}; \vec{p}_{I_{1}}) \bar\L^{\La,\Lambda_{0}}_{|I_{2}|, l_{2}}(\vec{p}_{I_{2}})\cdots \bar\L^{\La,\Lambda_{0}}_{|I_{j}|, l_{j}}(\vec{p}_{I_{j}}) \prod_{i=1}^{n} C^{\La,\Lambda_{0}}(p_{i})
\end{split}
\een
where the generating functional of the remainder functions $ \R^{\La,\Lao}_{D,n,l}(\bigotimes_{i=1}^{N}\O_{A_{i}}; \vec{p}) $ satisfies~\cite{Hollands:2011gf,Holland:2012vw,Holland:2013we}
\ben\label{RGD}
\begin{split}
R_{D}^{\Lambda,\Lambda_{0}}(\bigotimes_{i=1}^{N}\O_{A_{i}}(x_{i})):=&G^{\La,\Lambda_{0}}(\bigotimes_{i=1}^{N}\O_{A_{i}})-\sum_{[C]\leq D} \C_{A_{1}\ldots A_{N}}^{C} L^{\La,\Lambda_{0}}(\O_{C})\\
=& (1-\sum_{j\leq\Delta}\T^{j}_{\vec{x}\to(x_{N},\dots, x_{N})})\,  {G}^{\Lambda,\Lambda_{0}}_{D} (\bigotimes_{i=1}^{N}\O_{A_{i}}(x_{i}))
\end{split}
\een
with $\Delta=D-\sum_{i=1}^{N}[A_{i}]$. We conclude from these equations that our OPE acts as generally expected: It first tempers short distance singularities (on the total diagonal $x_{1}=\ldots=x_{N}$) by increasing the value of the regularization parameter $D$ in the remainder term \eqref{RGD}. Once the resulting functions are regular enough, it then takes a Taylor expansion around $\vec{x}=(x_{N},\ldots,x_{N})$. 

One can further use the identity \eqref{RGD} to estimate the remainder of the OPE. This way, it was shown in~\cite{Hollands:2011gf,Holland:2012vw, Holland:2013we} that, in the present model, the OPE is not only an asymptotic expansion, as was generally believed, but that it actually converges in the limit $D\to\infty$ for arbitrary configurations $(x_{1},\ldots,x_{N})\in M_{N}$ of the spacetime arguments. 

Finally, we will later also need the relation~\cite{Holland:2012vw, Holland:2013we}
\ben\label{eq:partrem}
\begin{split}
&G^{\Lambda,\Lambda_{0}}(\bigotimes_{i=1}^{N}\O_{A_{i}}(x_{i}))-\!\! \sum_{[C]\leq D}\!\! \C_{A_{1}\ldots A_{M}}^{C}(x_{1},\ldots,x_{M})\, G^{\Lambda,\Lambda_{0}}(\O_{C}(x_{M})\bigotimes_{j=M+1}^{N} \O_{A_{j}}(x_{j}))\\
&\quad=(1-\sum_{j=0}^{\Delta} \T^{j}_{(x_{1},\ldots,x_{M})\to (x_{M},\ldots,x_{M})})\ G^{\Lambda,\Lambda_{0}}\left([\bigotimes_{i=1}^{M}\O_{A_{i}}(x_{i})]_{D};\bigotimes_{j=M+1}^{N}\O_{A_{j}}(x_{j})\right) \ ,
\end{split}
\een
where $G^{\Lambda,\Lambda_{0}}([\bigotimes_{i=1}^{M}\O_{A_{i}}]_{D};\bigotimes_{j=M+1}^{N}\O_{A_{j}})$, defined in section \ref{subsec:subdiv}, are the AG's with regularization on the partial diagonal $x_{1}=\ldots=x_{M}$.

\section{Derivation of the deformation formula}\label{sec:proof}

While the definition of the perturbative OPE coefficients given in the previous section is very clear from a conceptual standpoint, it is somewhat dissatisfying that we have to rely on secondary objects (i.e. regularized AG's with insertions) in order to determine the OPE coefficients in perturbation theory. It would be desirable to be able to construct the perturbed OPE coefficients just in terms of the zeroth perturbation order ones, without reference to any other quantities.  This would yield support to the viewpoint that no data other than the OPE coefficients and one point functions are needed to define a quantum field theory. 

In the following we are going to show that such a construction is indeed possible. Starting from our definition of the OPE coefficients in terms of amputated Green's functions with insertions, see def.\ref{defOPE}, we will derive theorem \ref{thmint}, which is a formula that allows us to express the coefficients at a given order $r$ in terms of (an integral over) lower order ones. Our derivation of this formula is from first principles, i.e. we do not have to make any additional assumptions. 

To obtain the mentioned perturbation formula, we will first study the effect of taking a derivative with respect to the coupling constant $g$ of Green's functions with and without insertions, see section \ref{sec:vargAGs}. In section \ref{sec:OPEint} we will put these results to use and come to the actual derivation of the perturbation formula for the OPE coefficients, see theorem~\ref{thmint}.

\subsection[Variation of Green's functions w.r.t. the coupling constant]{Variation of Green's functions with respect to the coupling constant}\label{sec:vargAGs}

In the familiar diagrammatic framework of quantum field theory, increasing the perturbation order is represented by additional insertions of interaction vertices, corresponding in our case to $\varphi^{4}$ insertions, into the Feynman diagrams. This relation between insertions of the interaction operator  and the order of perturbation theory takes on a very simple form in our framework. Namely, one can show\footnote{The propositions derived in this section are understood to hold in the sense of formal power series in $\hbar$, i.e. they hold up to arbitrary finite ''loop order''.}:

\begin{prop}
\label{thmAP}
{\normalfont\sffamily \bfseries (M\"uller \cite{Muller:2002he})} The derivative with respect to the coupling constant of the CAG's without insertion, which were defined in section \ref{subsec:CAGs}, can be expressed as
\ben\label{AP1}
\partial_{g} L^{\Lambda,\Lambda_{0}}=\frac{1}{4!}\int\d^{4}y \, L^{\Lambda,\Lambda_{0}}(\varphi^{4}(y))\, ,
\een
where we have the CAG's with insertion of the composite operator $\varphi^{4}$ on the right hand side (see section \ref{subsec:CAGint} for the definition of operator insertions).
%
%
\end{prop}

\begin{proof}
We give a slightly different version of the proof compared to the one presented in~\cite{Muller:2002he}, which is more in the spirit of the present paper. Namely, we note that $L^{\Lambda,\Lambda_{0}}$ is defined through the following conditions:
\begin{enumerate}
\item Flow Equation \eqref{fe}
\item Boundary conditions \eqref{CAGBC1}, \eqref{CAGBC2}
\item Translation invariance
\end{enumerate}
Taking a $\Lambda$-derivative\footnote{Note that we can use the formal power series expansion \eqref{genfunctinsert}  and the translation properties of the CAG's with one insertion, eq.\eqref{CAGtrans}, in order to write the $\Lambda$ derivative of the r.h.s. of eq.\eqref{AP1} as $\partial_{\Lambda}\L_{2n,l}^{\Lambda,\Lambda_{0}}(\varphi^{4}(0);\vec{p})\frac{1}{4!}\int\d^{4}y \,\exp[i y (p_{1}+\ldots+p_{2n})] $. We can thereby apply the $\La$-derivative to $\L_{2n,l}^{\Lambda,\Lambda_{0}}(\varphi^{4}(0);\vec{p})$ without having to exchange the order of integration and differentiation.\label{footnote8}} on both sides of equation \eqref{AP1} and substituting the flow equations \eqref{fe} and \eqref{FEN}, we find that both expressions indeed obey the same linear homogeneous flow equation. Concerning the boundary conditions, we apply the $g$-derivative to eqs.\eqref{CAGBC1} and \eqref{CAGBC2}, which yields the conditions \eqref{BCL1} and \eqref{BCL2} with $A=\{w'=0, n'=4\}$. Fina
 lly, both sides of equation \eqref{AP1} are evidently translation invariant.
\end{proof}
The proposition can be generalized to the $g$-derivative of CAG's with insertions.
\begin{prop}
\label{thmAP2}
The CAG's with one insertion satisfy the identity
%
\ben\label{propLOgder}
\partial_{g} L^{\Lambda,\Lambda_{0}}(\O_{A}(x))=\frac{1}{4!}\int\d^{4}y \, L^{\Lambda,\Lambda_{0}}_{D=[A]}(\O_{A}(x)\otimes\varphi^{4}(y))\, .
\een
The $y$-integral converges absolutely, uniformely in the cutoffs, including $\Lambda=0, \Lambda_0=\infty$.
\end{prop}
%
\begin{proof}
The general strategy of the proof is similar to the one used in the proof of proposition \ref{thmAP} above: We want to establish the equality by showing that both sides of the equation satisfy the same flow equations and boundary conditions. 

First, however, we have to show that the integral on the right hand side of eq.\eqref{propLOgder} even exists. In the case of proposition \ref{thmAP} it was easy to show that the integral simply gives a momentum space delta function once we expanded the corresponding CAG-functionals in terms of $\hbar$ and $\hat{\varphi}$ [see footnote \ref{footnote8}]. In the case at hand, however, the situation is more complicated: The CAGs with two insertions are not smooth in the spacetime arguments (upon removal of the UV cutoff $\Lambda_0 \to \infty$), in contrast to the case of one insertion. Therefore, we have to take  more care in the present proof and, as a first step, establish convergence of the $y$-integral, uniformly as $\Lambda_0 \to \infty$.

\begin{description}
\item[Convergence of the $y$-integral:] We use the formal power series expansion \eqref{genfunctinsert} in order to write the integrand in terms of the moments $\L^{\Lambda,\Lambda_{0}}_{D=[A],n,l}(\O_{A}(x)\otimes\varphi^{4}(y); \vec{p})$, which can be estimated with the help of the bound \eqref{eqboundN=2} [see page \pageref{eqboundN=2}]. It follows from this bound, choosing the parameters $w=0,t=0,r=1$ and any $s\in\mathbb{N}$ in \eqref{eqboundN=2}, that the modulus of our integrand, $|\L^{\Lambda,\Lambda_{0}}_{D=[A],n,l}(\O_{A}(x)\otimes\varphi^{4}(y); \vec{p})|$, is smaller than a function which depends on the spacetime arguments as $|x-y|^{-4-s}$. We conclude that the integrand falls off rapidly for large $|y|$, i.e. in the \emph{infrared region}, and that the integral over that region converges absolutely. On the other hand, choosing $w=0,t=0,s=0$ and $r\in\mathbb{N}$ arbitrary, the bound \eqref{eqboundN=2} also implies that the modulus of our integrand is smaller than
  a function which behaves as $|x-y|^{-3-\frac{1}{r}}$. Therefore, also the $y$-integral over the ultra-violet region, $y\approx x$, is absolutely convergent. In summary, we conclude from the bound \eqref{eqboundN=2} that the integrand (expanded in powers of $\hbar,\hat{\varphi}$) on the r.h.s. of eq.\eqref{propLOgder} satisfies the bound
\ben\label{integrandbound1}
\begin{split}
|\L^{\Lambda,\Lambda_{0}}_{D=[A],n,l}(\O_{A}(x)\otimes\varphi^{4}(y); \vec{p})|&\leq \min\left( \frac{(\La+m)^{-r}}{   |x-y|^{3+r}}, \frac{(\La+m)^{-\frac{1}{r}}}{|x-y|^{3+\frac{1}{r}}} \right)\\
&\times \left[(\Lambda+m)^{[A]-n+1}\, \Pol_{1}\left(\log\frac{\Lambda+m}{m}\right) \Pol_{2}\left(\frac{|\vec{p}|}{\Lambda+m}\right)\right]
\end{split}
\een
for some some $r>1$ and suitable polynomials $\Pol_{i}$ with positive coefficients. Since the function on the r.h.s. of this inequality is integrable over $y\in\mathbb{R}^{4}$, we conclude that the integral on the r.h.s. of \eqref{propLOgder} is absolutely convergent, even as we remove the cutoffs.

\item[Flow equations:] Next we would like compare the flow equations for both sides of eq.\eqref{propLOgder}. Taking the $g$-derivative of the flow equation for $L^{\Lambda,\Lambda_{0}}(\O_{A})$, eq.\eqref{FEN}, we obtain for the left hand side
\ben\label{FELgder}
\begin{split}
\partial_{\Lambda}\partial_{g}L^{\Lambda,\Lambda_{0}}(\O_{A})&=\frac{\hbar}{2}\bra \varp \, ,\, \dot{C}^{\Lambda}\star \varp  \ket\, \partial_{g}L^{\Lambda,\Lambda_{0}}(\O_{A})-\bra \varp  \partial_{g} L^{\Lambda,\Lambda_{0}}(\O_{A}) \, ,\, \dot{C}^{\Lambda}\star \varp L^{\Lambda,\Lambda_{0}}(\varphi) \ket\\
&-\bra \varp  L^{\Lambda,\Lambda_{0}}(\O_{A}) \, ,\, \dot{C}^{\Lambda}\star \varp \frac{1}{4!}\int\d^{4}y\, L^{\Lambda,\Lambda_{0}}(\varphi^{4}(y)) \ket\, ,
\end{split}
\een
where we used proposition \ref{thmAP} in the last line. In order to determine a flow equation for the right side of eq.\eqref{propLOgder}, we take a $\La$-derivative of $\frac{1}{4!}\int\d^{4}y \, L^{\Lambda,\Lambda_{0}}_{D=[A]}(\O_{A}(x)\otimes\varphi^{4}(y))$.  We would now like to exchange the order of the $y$-integral with the $\La$-derivative, which would allow us to use the flow equation \eqref{FEN}. This \emph{differentiation under the integral sign} has to be justified of course: We have to show that also $\frac{1}{4!}\int\d^{4}y \,\partial_{\La} L^{\Lambda,\Lambda_{0}}_{D=[A]}(\O_{A}(x)\otimes\varphi^{4}(y))$ converges absolutely. This is again achieved with the help of the bound \eqref{eqboundN=2}, which implies
\ben\label{integrandbound2}
\begin{split}
|\partial_{\La}\L^{\Lambda,\Lambda_{0}}_{D=[A],n,l}(\O_{A}(x)\otimes\varphi^{4}(y); \vec{p})|&\leq \min\left( \frac{(\La+m)^{-r}}{   |x-y|^{3+r}}, \frac{(\La+m)^{-\frac{1}{r}}}{|x-y|^{3+\frac{1}{r}}} \right)\\
&\times \left[(\Lambda+m)^{[A]-n}\, \Pol_{1}\left(\log\frac{\Lambda+m}{m}\right) \Pol_{2}\left(\frac{|\vec{p}|}{\Lambda+m}\right)\right]
\end{split}
\een
if we choose $t=1$ in that bound. Thus, we are allowed to exchange the order of the $y$-integral and the $\La$-derivative. Using eq.\eqref{FEN}, it is then easy to check that the right hand side of eq.\eqref{propLOgder} obeys the same flow equation, eq.\eqref{FELgder}, as the left hand side.

\item[Boundary conditions:]  Taking a $g$-derivative of eqs.\eqref{BCL1} and $\eqref{BCL2}$, we find that $\partial_{g}L^{\Lambda,\Lambda_{0}}(\O_{A})$ is subject to the following boundary conditions:
\begin{align}\label{BCgder1}
\partial^{w}_{\vec{p}}\partial_{g}\LscO_{n,l}(\O_{A}(0); \vec{0})&= 0 \quad \text{ for }n+|w|\leq [A] \\
\partial^{w}_{\vec{p}}\partial_{g}\LscIr_{n,l}(\O_{A}(0); \vec{p})&=0\quad \text{ for }n+|w|>[A] \quad . \label{BCgder2}
\end{align}
By definition, the right hand side of equation \eqref{propLOgder} satisfies the same boundary conditions, which can be seen from eqs.\eqref{Gbound1} and \eqref{Gbound2} and by recalling also that $F^{\Lambda,\Lambda_{0}}_{D}(\O_{A}(x)\otimes\O_{B}(0))=- L^{\Lambda,\Lambda_{0}}_{D}(\O_{A}(x)\otimes\O_{B}(0))$ from eq.\eqref{CAGFeq}.

\end{description}
\noindent
To summarize, we have established that a) the integral over $y$ in eq.\eqref{propLOgder} converges, that b) both sides of the equation satisfy the same flow equations and that c) both sides are subject to the same boundary conditions. This establishes the equality and finishes the proof.
\end{proof}

Below we will also be interested in the $g$-derivative of the amputated Green's functions with insertions, $G^{\Lambda,\Lambda_{0}}(\bigotimes_{i=1}^{N}\O_{A_{i}})$.
\begin{prop}\label{thmAGAP}
The $g$-derivative of the amputated Green's functions with insertions can be expressed as
\ben\label{AGAP}
\begin{split}
&\hbar\, \partial_{g}G^{\Lambda,\Lambda_{0}}(\bigotimes_{i=1}^{N}\O_{A_{i}}(x_{i}))=\\
&\frac{-1}{4!}\int\d^{4} y\Big[ G^{\Lambda,\Lambda_{0}}(\bigotimes_{i=1}^{N}\O_{A_{i}}(x_{i})\otimes \varphi^{4}(y))-G^{\Lambda,\Lambda_{0}}(\bigotimes_{i=1}^{N}\O_{A_{i}}(x_{i})) L^{\Lambda,\Lambda_{0}}(\varphi^{4}(y))\\
&- \sum_{j=1}^{N}\sum_{[C]\leq [A_{j}]} \C_{\In A_{j}}^{C}(y,x_{j})\, G^{\Lambda,\Lambda_{0}}(\bigotimes_{\atop{i=1}{i\neq j}}^{N}\O_{A_{i}}(x_{i})\otimes\O_{C}(x_{j})) \Big]
\end{split}
\een
where $\In:=\{n=4,w=0\}$, i.e. $\O_{\In}=\varphi^{4}$ is the interaction operator. The $y$-integral converges absolutely, uniformely in the cutoffs, including $\Lambda=0, \Lambda_0=\infty$.
\end{prop}
\begin{rem}
In the case $N=1$ this reduces to proposition \ref{thmAP2}.
\end{rem}
\begin{proof}
We follow the same basic strategy as in the proof of proposition \ref{thmAP2}, i.e. we first establish convergence of the integral on the right hand side of the equation before we compare flow equations and boundary conditions. 
\begin{description}
\item[Convergence of the $y$-integral:] Again we argue that the integral on the right hand side - as well as its $\Lambda$-derivative - converges absolutely. Let us first discuss the infrared behavior of the integral, i.e. the case of $|y|$ being large. Note that the first two terms on the right hand side of eq.\eqref{AGAP} combine to
\ben
G^{\Lambda,\Lambda_{0}}(\bigotimes_{i=1}^{N}\O_{A_{i}}\otimes \varphi^{4})-G^{\Lambda,\Lambda_{0}}(\bigotimes_{i=1}^{N}\O_{A_{i}}) L^{\Lambda,\Lambda_{0}}(\varphi^{4})= \hbar  H^{\Lambda,\Lambda_{0}}(\bigotimes_{i=1}^{N}\O_{A_{i}} ; \varphi^{4})
\een
by definition [see eq.\eqref{lemdecompeq}]. In view of the bound derived in theorem \ref{thmboundH} (see page \pageref{thmboundH}), we know that the moments $\H^{\Lambda,\Lambda_{0}}_{n,l}$ of these functionals decay more rapidly than\\ $\min_{1\leq i\leq  N}|x_{i}-y|^{-s}$ for large $|y|$ and for any $s\in\mathbb{N}$. A similar bound also holds for the functionals with a $\La$-derivative (the case $t=1$ in theorem \ref{thmboundH}).
Thus, the integral over these terms converges absolutely in the infrared. To estimate the IR-behavior of the OPE coefficients $\C_{\In A_{j}}^{C}(y,x_{j})$ with $[C]\leq [A_{j}]$, we recall their definition in terms of AG's with insertions from def.~\ref{defOPE}. 
The factorized contribution to these coefficients, $\D^{C}\{ L^{0,\Lambda_{0}}(\varphi^{4}) L^{0,\Lambda_{0}}(\O_{A_{j}})  \}$,  vanishes in the case $[C]<4+[A_{j}]$ due to the boundary conditions of the CAG's with one insertion, eq.\eqref{BCL1}. The remaining contribution,  $\D^{C}L^{0,\Lambda_{0}}_{[C]-1}(\varphi^{4}(y)\otimes\O_{A_{j}}(x_{j}))$, is found to decay rapidly for large $|y|$ by the same arguments as in proposition \ref{thmAP2} [i.e. using the bound \eqref{eqboundN=2}].

Let us now discuss the UV behavior of the integral, i.e. the regions where $y$ is close to one of the $x_{j}$. The two contributions which are potentially singular in this region are (other contributions are smooth)
\ben\label{gGderUV}
\begin{split}
G^{\Lambda,\Lambda_{0}}(\bigotimes_{i=1}^{N}\O_{A_{i}}(x_{i})\otimes \varphi^{4}(y))&- \sum_{[C]\leq [A_{j}]} \C_{\In A_{j}}^{C}(y,x_{j}) G^{\Lambda,\Lambda_{0}}(\bigotimes_{\atop{i=1}{i\neq j}}^{N}\O_{A_{i}}(x_{i})\otimes\O_{C}(x_{j}))\\
&= G^{\Lambda,\Lambda_{0}}\big([\O_{A_{j}}(x_{j})\otimes\varphi^{4}(y)]_{[A_{j}]}\, ; \bigotimes_{\atop{i=1}{i\neq j}}^{N}\O_{A_{i}}(x_{i})  \big)
\end{split}
\een
where we used equation \eqref{eq:partrem} in the second line. To see that we can safely integrate over the region $y\approx x_{j}$, we note that corollary \ref{corboundH} [see page \pageref{corboundH}] implies that the right hand side (as well as its $\La$-derivative) is bounded by a function which scales as
 $|x_{j}-y|^{-3-\frac{1}{r}}$ for any $r\in\mathbb{N}$ as $y$ approaches $x_{j}$. We conclude that the $y$-integral in eq.\eqref{AGAP} also converges absolutely in the regions where $y$ is close to any of the spacetime arguments $x_{j}$, and the same is true if one takes the $\La$-derivative of the integrand.

In summary, combining all these bounds for different spacetime regions, we conclude that the integral on the right hand side of equation \eqref{AGAP}, as well as its $\La$-derivative, is in fact absolutely convergent.

\item[Flow equations and boundary conditions:] In order to establish the equality \eqref{AGAP}, it remains to show that both sides of the equation satisfy the same flow equation and boundary conditions, making use also of the fact that we may exchange the order of $\La$-differentiation and $y$-integral on the right hand side. This straightforward, but somewhat lengthy, proof can be found in appendix \ref{ap:propgder}. 
\end{description}

\end{proof}

\subsection[Variation of OPE coefficients w.r.t. the coupling constant]{Variation of OPE coefficients with respect to the coupling constant}\label{sec:OPEint}

The OPE coefficients have been defined in def. \ref{defOPE} in terms of amputated Green's functions with insertions. The results of the previous section can be used to derive our main formula for the deformation of the OPE algebra:
%
\setcounter{thms}{0}

\begin{thm}[OPE deformation]
\label{thmint}
 Let $\In:=\{n=4, v=0\}$, i.e. $\O_{\In}=\varphi^{4}$. The derivative of the OPE coefficients w.r.t. the coupling constant $g$ can be expressed as

\ben\label{eq:thmint}
\begin{split}
&\hbar \partial_{g}\, \C_{A_{1}\ldots A_{N}}^{B}(x_{1},\ldots, x_{N})= \frac{-1}{4!}\int\d^{4}y \bigg[ \C_{\In A_{1}\ldots A_{N} }^{B}(y,x_{1},\ldots,x_{N}) \\
& \qquad -\sum_{i=1}^{N}\sum_{[C]\leq [A_{i}]} \!\!\!\C_{\In A_{i}}^{C}(y,x_{i})\, \C_{A_{1}\ldots \widehat{A_{i}}\,C\ldots A_{N}}^{B}(x_{1},\ldots, x_{N})- \sum_{[C]< [B]} \C_{A_{1}\ldots A_{N}}^{C}(x_{1},\ldots,x_{N})\, \C_{\In C }^{B}(y,x_{N}) \bigg]\, .
\end{split}
\een
Here $\widehat{A_{i}}$ denotes omission of the corresponding index. The relation holds to arbitrary finite perturbation order in massive Euclidean $g\varphi^{4}$-theory with BPHZ renormalization conditions. 
\end{thm}
%
 \begin{rem}\label{remdefrom} The loop parameter $\hbar$ is only of auxiliary nature and can be set to $1$, if one is not interested in a \emph{loop expansion}. A few further observations are in order:
 \begin{enumerate}
  
 \item In the proof below, convergence of the integral on the right hand side of eq.\eqref{eq:thmint} follows from convergence of the integrals in propositions \ref{thmAP2} and \ref{thmAGAP}. It is instructive, however, to try and understand why the integral in eq.\eqref{eq:thmint} converges just in terms of properties of the OPE coefficients. From this perspective, the expressions which appear in the second line of equation \eqref{eq:thmint} may be interpreted as ''counter terms'', i.e. they cancel possible UV- and IR-divergent contributions from the first term on the right hand side in the theory without cutoffs. More precisely, if the integration variable $y$ is close to one of the arguments $x_{j}$, then we can factorize the first coefficient on the right hand side (see~\cite{Holland:2012vw, Holland:2013we} for a proof of this identity)
 \ben
 \C_{\In A_{1}\ldots A_{N} }^{B}(y,x_{1},\ldots,x_{N})=\sum_{C} \C_{\In A_{j}}^{C}(y,x_{j})  \C_{A_{1}\ldots \widehat{A_{j}}\,C\ldots A_{N}}^{B}(x_{1},\ldots, x_{N})\, .
 \een
The corresponding counter term subtracts all contributions from the sum over $C$ with $[C]\leq [A_{j}]$. Recall that the OPE coefficients were given in def.~\ref{defOPE} in terms of the AG's with operator insertions, whose short distance scaling behavior is estimated in corollary \ref{corbound} in the appendix. This corollary implies that the remaining terms in the sum over $C$, which contain coefficients $\C_{\In A_{j}}^{C}(y,x_{j})$ with $[C]>[A_{j}]$, diverge at most like $|x_{j}-y|^{-3-\frac{1}{r}}$ (for any $r\in\mathbb{N}$) and are thus indeed integrable on the domain with $y$ close to $x_{j}$. Similarly, if $|y|$ is large compared to the $|x_{i}|$, then we can factorize the first term on the right hand side of \eqref{eq:thmint}:
  \ben
 \C_{\In A_{1}\ldots A_{N} }^{B}(y,x_{1},\ldots,x_{N})=\sum_{C}  \C_{A_{1}\ldots A_{N}}^{C}(x_{1},\ldots,x_{N})\, \C_{\In C }^{B}(y,x_{N})
 \een
Here the corresponding counter term in \eqref{eq:thmint} subtracts all summands with $[C]<[B]$. One can check, using  arguments from the derivation of proposition  \ref{thmAGAP}, that the remaining terms decay faster than any power $|y-x_{N}|^{-s}$ for arbitrary $s\in\mathbb{N}$. 
 
We will observe this cancellation of divergences in a concrete example in section \ref{sec:iteration}. 
 
 \item Composite operators in perturbative quantum field theory are generally subject to renormalization ambiguities, which are in one-to-one correspondence, in the present framework, to the freedom to modify the boundary conditions \eqref{BCL1}, see e.g.~\cite{Keller:1991bz} for further discussion of this point. It can be shown that a change in the boundary conditions, i.e. renormalization 
prescription, in turn can be absorbed by a field redefintion of the form $\O_{A}' = \sum_{A'} Z_{A}^{A'}\, \O_{A'}$, where $Z_{A}^{A'}$ is an invertible `mixing matrix' of complex numbers, with vanishing
entries for $[A'] < [A]$. The OPE-coefficients clearly change accordingly under such a field redefinition; in fact, if the coefficients associated with the new prescription are called 
$\C_{AB}^{\prime C}$ (e.g. for $N=2$), then
 \ben\label{eq:thmintchange}
\C^{\prime C}_{AB} = \sum_{A',B',C'} Z_A^{A'} Z_B^{B'} (Z^{-1})^C_{C'} \ \C_{A'B'}^{C'} \, .
\een
By substituting this formula (and its analogs for $N>2$) into eq.~\eqref{eq:thmint}, one will obtain a corresponding equation for the OPE-coefficients associated with the new 
prescription. If the mixing matrix $Z_A^{A'}$ is independent of the coupling constant $g$, then the resulting equation will have the same form, since the mixing matrix will then 
simply drop out. On the other hand, if $Z_A^{A'}$ depends on $g$ (as is normally the case), then there will be another term on the left side of eq.~\eqref{eq:thmint} involving
$\sum_B (Z^{-1})_{A'}^B \partial_g Z_B^{C'}$. 

\item In this paper, we are dealing with the massive theory, $m>0$. It is interesting to ask whether eq.\eqref{eq:thmint} would also hold for $m=0$. 
In the massless case, the BPHZ-type renormalization conditions that we employ for our composite fields [see eq. \eqref{BCL1}] are not appropriate, and one has to use some other, reasonable, 
set of conditions. In the light of the previous remark, one would therefore expect that eq.~\eqref{eq:thmint} is somewhat modified in the massless case. 
A deeper study of the massless case may be an interesting topic for future research.
 \end{enumerate}
\end{rem}
\begin{proof}[Proof of theorem \ref{thmint}:]
Combining our definition of the OPE coefficients, eq.\eqref{OPEhigh}, with the remainder formula \eqref{RGD}, it follows that (suppressing the dependence on the spacetime arguments for the moment)
\ben\label{OPEalt}
\C_{A_1 \ldots A_N}^B=\D^{B}\left\{ G^{0,\Lambda_{0}}(\bigotimes_{i=1}^{N}\O_{A_{i}}) - \sum_{[{C}]<[B]} \C_{A_1 \ldots A_N}^{{C}} L^{0,\Lambda_{0}}(\O_{{C}})  \right\}\, .
\een
Applying the $g$-derivative to both sides of the equation and using propositions \ref{thmAP2} and \ref{thmAGAP} for the $g$-derivatives of the functionals $L^{0,\Lambda_{0}}(\O_{{C}})$ and $G^{0,\Lambda_{0}}(\bigotimes_{i=1}^{N}\O_{A_{i}})$ respectively, we obtain the formula
\ben\label{inteideq1}
\begin{split}
\hbar\partial_{g}&\C_{A_1 \ldots A_N}^B=\ \D^{B}\frac{-1}{4!}\int\d^{4}y \Big[ G^{0,\Lambda_{0}}(\bigotimes_{i=1}^{N}\O_{A_{i}}\otimes\varphi^{4})- G^{0,\Lambda_{0}}(\bigotimes_{i=1}^{N}\O_{A_{i}})L^{0,\Lambda_{0}}(\varphi^{4}) \\
-& \sum_{j=1}^{N}\sum_{[C]\leq [A_{j}]}\C_{\In A_{j}}^{C}G^{0,\Lambda_{0}}(\bigotimes_{\atop{i=1}{i\neq j}}^{N}\O_{A_{i}}\otimes \O_{C})+\hbar \sum_{[C]< [B]} \C_{A_{1}\ldots A_{N}}^{C} L^{0,\Lambda_{0}}_{[C]}(\varphi^{4}\otimes\O_{C})\Big]\\
-\hbar &\D^{B}\sum_{[{C}]<[B]} \left(\partial_{g}\C_{A_1 \ldots A_N}^{{C}}\right) L^{0,\Lambda_{0}}(\O_{{C}})\, . 
\end{split}
\een
These propositions also give absolute convergence of the $y$-integral, uniformly in $\Lambda,\Lambda_0$. 
Our aim is now to express the right hand side of this equation in terms of OPE coefficients and CAG's with one insertion only. For that purpose the following relation, which follows from the ''remainder formula''  \eqref{RGD}, will be useful:
\ben\label{OPEderconv}
\begin{split}
\D^{B}\left[ G^{0,\Lambda_{0}}(\bigotimes_{i=1}^{N}\O_{A_{i}})-\sum_{[C]\leq [B]} \C_{A_{1}\ldots A_{N}}^{C} L^{0,\Lambda_{0}}(\O_{C}) \right]&=\D^{B} R_{[B]}^{0,\Lambda_{0}}(\bigotimes_{i=1}^{N}\O_{A_{i}})\\
&=\D^{B} (1-\sum_{j\leq\Delta}\T^{j})\, G_{[B]}^{0,\Lambda_{0}}(\bigotimes_{i=1}^{N}\O_{A_{i}}) =0
\end{split}
\een
Here $\Delta=[B]-([A_{1}]+\ldots+[A_{N}])$.
To show that the expression in the last line indeed vanishes, we made use of the decomposition ${G}_{D}^{\Lambda,\Lambda_{0}}(\bigotimes_{i=1}^{N}\O_{A_{i}})=
\hbar F_{D}^{\Lambda,\Lambda_{0}}(\bigotimes_{i=1}^{N}\O_{A_{i}})+\prod_{i=1}^{N} L^{\Lambda,\Lambda_{0}}(\O_{A_{i}})$, see \eqref{GDdef}. The contribution from the $F$-functionals then vanishes due to the boundary conditions \eqref{Gbound1}. To see that the factorized contribution vanishes as well, we apply the integral formula for the remainder of the Taylor expansion:
\ben
(1-\sum_{j\leq\Delta}\T^{j}_{\vec{x}}) \prod_{i=1}^{N} L^{0,\La_{0}}(\O_{A_{i}}(x_{i})) =\sum_{|v|=\Delta+1} \int_{0}^{1}\d\tau\, \frac{\vec{x}^{v}|v|}{v!} (1-\tau)^{\Delta} \partial_{\tau\vec{x}}^{v} \prod_{i=1}^{N} L^{0,\La_{0}}(\O_{A_{i}}(\tau x_{i})) 
\een
We can pull the $\vec{x}$-derivatives into the CAG's, using the Lowenstein rule \eqref{Low1}. It is then not hard to verify, using the boundary conditions for the CAG's with one insertion, eq.\eqref{BCL1}, that the expression vanishes once we apply the differential operator $\D^{B}$.

Equation \eqref{OPEderconv} allows us to replace the AG's with multiple insertions in eq.\eqref{inteideq1} (i.e. the first three terms on the right side) by finite sums over OPE coefficients multiplied by CAG's with one insertion. Thus,
\ben\label{inteideq1b}
\begin{split}
\hbar\partial_{g}&\C_{A_1 \ldots A_N}^B=\ \D^{B}\frac{-1}{4!}\int\d^{4}y \Big[ \sum_{[C]\leq [B]} \C_{\In A_{1}\ldots A_{N}}^{C} L^{0,\La_{0}}(\O_{C}) - \sum_{[C]\leq [B]} \C_{ A_{1}\ldots A_{N}}^{C} L^{0,\La_{0}}(\O_{C})\, L^{0,\Lambda_{0}}(\varphi^{4}) \\
&- \sum_{j=1}^{N}\sum_{[C']\leq [A_{j}]}\C_{\In A_{j}}^{C'}\sum_{[C]\leq [B]} \C_{\In A_{1}\ldots \widehat{A_{j}}C'\ldots A_{N}}^{C} L^{0,\La_{0}}(\O_{C})+\hbar \sum_{[C]< [B]} \C_{A_{1}\ldots A_{N}}^{C} L^{0,\Lambda_{0}}_{[C]}(\varphi^{4}\otimes\O_{C})\Big]\\
&-\hbar \D^{B}\sum_{[{C}]<[B]} \left(\partial_{g}\C_{A_1 \ldots A_N}^{{C}}\right) L^{0,\Lambda_{0}}(\O_{{C}})\, ,
\end{split}
\een
where, as in the statement of the theorem,  $\widehat{A_{i}}$ denotes omission of the corresponding index. Next we note that the expression $\D^{B}\{L^{0,\La_{0}}(\O_{C})\, L^{0,\Lambda_{0}}(\varphi^{4})\}$ in fact vanishes for $[C]=[B]$, which follows once more in view of the boundary conditions \eqref{BCL1} for the CAG's with one insertion. Hence, we can restrict the sum over $C$ in the second term on the r.h.s. of eq.\eqref{inteideq1b} to values satisfying $[C]<[B]$. This allows us to combine the second and the fourth term on the r.h.s. of eq.\eqref{inteideq1b}, using the relation
\ben
L^{0,\Lambda_{0}}(\varphi^{4}) L^{0,\Lambda_{0}}(\O_{C})- \hbar L^{0,\Lambda_{0}}_{[C]}(\varphi^{4}\otimes\O_{C})=G^{0,\Lambda_{0}}_{[C]}(\varphi^{4}\otimes\O_{C})\,, 
\een
which follows from a combination of eqs.\eqref{CAGFeq} and \eqref{GDdef}. Hence, we arrive at the form
\ben\label{inteideq1c}
\begin{split}
\hbar\partial_{g}&\C_{A_1 \ldots A_N}^B=\ \D^{B}\frac{-1}{4!}\int\d^{4}y \Big[ \sum_{[C]\leq [B]} \C_{\In A_{1}\ldots A_{N}}^{C} L^{0,\La_{0}}(\O_{C}) - \sum_{[C]< [B]} \C_{ A_{1}\ldots A_{N}}^{C} G^{0,\Lambda_{0}}_{[C]}(\varphi^{4}\otimes\O_{C})  \\
&- \sum_{j=1}^{N}\sum_{[C']\leq [A_{j}]}\C_{\In A_{j}}^{C'}\sum_{[C]\leq [B]} \C_{\In A_{1}\ldots \widehat{A_{j}}C'\ldots A_{N}}^{C} L^{0,\La_{0}}(\O_{C})\Big]-\hbar \D^{B}\sum_{[{C}]<[B]} \left(\partial_{g}\C_{A_1 \ldots A_N}^{{C}}\right) L^{0,\Lambda_{0}}(\O_{{C}})\, ,
\end{split}
\een
Now we can use the remainder formula \eqref{RGD} as well as the relation \eqref{OPEderconv} once more to replace the second term on the r.h.s of eq.\eqref{inteideq1c} by (using also the fact that $[C]\leq [B]$)
\ben
\begin{split}
&\D^{B} G^{0,\Lambda_{0}}_{[C]}(\varphi^{4}\otimes\O_{C})=\D^{B}\{ G^{0,\Lambda_{0}}(\varphi^{4}\otimes\O_{C})-\sum_{[C']\leq [C]} \C_{\In C}^{C'}\, L^{0,\Lambda_{0}}(\O_{C'}) \}\\
& = \D^{B}\{ \sum_{[C']\leq [B]} \C_{\In C}^{C'}\, L^{0,\Lambda_{0}}(\O_{C'})-\sum_{[C']\leq [C]} \C_{\In C}^{C'}\, L^{0,\Lambda_{0}}(\O_{C'}) \}  = \D^{B}\{ \sum_{[C]<[C']\leq [B]} \C_{\In C}^{C'}\, L^{0,\Lambda_{0}}(\O_{C'}) \} \, .
\end{split}
\een
Substituting the r.h.s of this equation into \eqref{inteideq1c} and bringing the last term on the r.h.s. of eq.\eqref{inteideq1c} to the left, we finally obtain the formula 
\ben\label{preinduct}
\begin{split}
&\D^{B}\sum_{[C]\leq [B]}\left(\hbar\partial_{g}\C_{A_{1}\ldots A_{N}}^{C}\right)  L^{0,\Lambda_{0}}(\O_{C})= 
\\
& \D^{B}\sum_{[C]\leq [B]}L^{0,\Lambda_{0}}(\O_{C})\frac{-1}{4!}\int\d^{4}y \left[ \C_{\In A_{1}\ldots A_{N}}^{C}- \sum_{j=1}^{N}\sum_{[C']\leq [A_{j}]}\C_{\In A_{j}}^{C'}\C_{A_{1}\ldots \widehat{A_{j}}C' \ldots A_{N}}^{C}- \sum_{[C'] < [C]} \C_{A_{1}\ldots A_{N}}^{C'}\C_{\In C'}^{C}\right]  \, .
\end{split}
\een
In the sum on the left we also made use of the fact that $\D^{B}L^{0,\Lambda_{0}}(\O_{C})=\delta_{B,C}$ for $[B]=[C]$, which follows from the boundary conditions \eqref{BCL1}. Since equation \eqref{preinduct} holds for any choice of index $B$, we can ascend inductively in [B]:

\begin{itemize}

\item Let $[B]=0$, i.e $B=\mathds{1}$. The boundary conditions for the CAG's with one insertion then imply $\D^{B}L^{0,\Lambda_{0}}(\O_{B})=1$, which immediately yields
\ben
\hbar\partial_{g}\C_{A_{1}\ldots A_{N}}^{B}=-\int\frac{\d^{4}y}{4!} \left[ \C_{\In A_{1}\ldots A_{N}}^{B}- \sum_{j=1}^{N}\sum_{[C']\leq [A_{j}]}\C_{\In A_{j}}^{C'}\C_{A_{1}\ldots \widehat{A_{j}}C\ldots A_{N}}^{B}- \sum_{[C'] < [B]} \C_{A_{1}\ldots A_{N}}^{C'}\C_{\In C'}^{B}\right] 
\een
in accordance with the theorem (here the sum over $[C']<[B]=0$ actually vanishes).

\item Assume the theorem holds for all $B$ with $[B]< D$. Pick a $B'$ with $[B']=D$. Then we obtain

\ben\label{deformindeq}
\begin{split}
&\D^{B'}\sum_{[C]\leq [B']}\left(\hbar\partial_{g}\C_{A_{1}\ldots A_{N}}^{C}\right)  L^{0,\Lambda_{0}}(\O_{C}) = \D^{B'}\sum_{[C]< [B']}L^{0,\Lambda_{0}}(\O_{C})\frac{-1}{4!}\int\d^{4}y \\
& \times\, \left[ \C_{\In A_{1}\ldots A_{N}}^{C}- \sum_{j=1}^{N}\sum_{[C']\leq [A_{j}]}\C_{\In A_{j}}^{C'}\C_{A_{1}\ldots \widehat{A_{j}}C\ldots A_{N}}^{C}- \sum_{[C'] < [C]} \C_{A_{1}\ldots A_{N}}^{C'}\C_{\In  C'}^{C}\right]  \\
&-\int \frac{\d^{4}y}{4!}  \left[ \C_{\In A_{1}\ldots A_{N}}^{B'}- \sum_{j=1}^{N}\sum_{[C']\leq [A_{j}]} \hspace{-.2cm} \C_{\In A_{j}}^{C'}\C_{A_{1}\ldots \widehat{A_{j}}C\ldots A_{N}}^{B'}- \hspace{-.2cm}\sum_{[C'] < [B']}\hspace{-.2cm} \C_{A_{1}\ldots A_{N}}^{C'}\C_{\In C'}^{B'}\right] 
\end{split}
\een
where we made use of the boundary conditions for the CAG's with one insertion in the last line, which imply for $[B]=[C]$ that $\D^{B}L^{0,\Lambda_{0}}(\O_{C})=\delta_{B,C}$, as mentioned above. By assumption, we can apply theorem \ref{thmint} on the left hand side for the terms with $[C]<[B']$, which yields exactly the same expressions as the sum over $[C]<[B']$ on the right hand side. Subtracting these contributions from both sides of equation \eqref{deformindeq}, we are again left with the claim of the theorem, which closes the induction.

\end{itemize}
Since the propositions derived in section \ref{sec:vargAGs} hold up to any finite loop order (and therefore also to finite perturbation order in $g$), the same holds true for theorem \ref{thmint}.
\end{proof}
Expanding the OPE coefficients as formal power series in the coupling constant, which we write as
\ben
\C_{A_{1}\ldots A_{N}}^{B}(x_{1},\ldots, x_{N}) =: \sum_{i=0}^{\infty} (\C_{i})_{A_{1}\ldots A_{N}}^{B}(x_{1},\ldots, x_{N})\ g^{i}\, ,
\een
and fixing our auxiliary loop parameter $\hbar=1$, the theorem implies:
\begin{corollary}\label{cor}
For any $r\in\mathbb{N}_{0}$, the OPE coefficients at perturbation order $r+1$ are given by
\ben\label{eq:thmintexp}
\begin{split}
  (\C_{r+1})_{A_{1}\ldots A_{N}}^{B}(x_{1},\ldots, x_{N})= \frac{-1}{4!\ (r+1)}\int\d^{4}y\ \bigg[ (\C_{r})_{\In A_{1}\ldots A_{N} }^{B}(y,x_{1},\ldots,x_{N})& \\
  -\sum_{i=1}^{N}\sum_{[C]\leq [A_{i}]} \sum_{s=0}^{r} (\C_{s})_{\In A_{i}}^{C}(y,x_{i})\, (\C_{r-s})_{A_{1}\ldots \widehat{A_{i}}\,C\ldots A_{N}}^{B}(x_{1},\ldots, x_{N})&\\
-\sum_{[C]< [B]} \sum_{s=0}^{r} (\C_{s})_{A_{1}\ldots A_{N}}^{C}(x_{1},\ldots,x_{N})\, (\C_{r-s})_{\In C }^{B}(y,x_{N})& \bigg]\, . 
\end{split}
\een

\end{corollary}
This equation allows us to determine the coefficients at order $(r+1)$ from those of lower perturbation order. In particular, given the OPE coefficients of the free theory, we can iterate this equation to construct the coefficients to arbitrary order in $g$. An explicit  application of this recursive algorithm can be found in the next section.

\section{Iterative construction of OPE coefficients: examples}\label{sec:iteration}

Here we illustrate our recursion scheme in order to construct OPE coefficients at low perturbation orders in some simple examples. The coefficients of the free theory, which are the initial data of this recursion process, are obtained quite easily using Wick's theorem (see for example~\cite{Olbermann:2012uf}  for a comprehensive analysis of the massless case). 
We can proceed to first perturbation order with the help of corollary~\ref{cor}. Note that, in this section we often use the alternative notation $\C_{\O_{A_{1}}\ldots \O_{A_{N}}}^{\O_{B}}$ instead of $\C_{A_{1}\ldots A_{N}}^{B}$ for OPE coefficients.

\paragraph{The coefficient $(\C_{1})_{\varphi\, \varphi}^{\varphi^{4}}$ :}
 Equation \eqref{eq:thmintexp} allows us to write this coefficient as
\ben\label{OPE1stex1}
\begin{split}
&  (\C_{1})_{\varphi\, \varphi}^{\varphi^{4}}(x_{1}, x_{2})= \frac{-1}{4!}\!\!\int\!\!\d^{4}y\ \bigg[ (\C_{0})_{\varphi^{4}\, \varphi\, \varphi }^{\varphi^{4}}(y,x_{1},x_{2})-\!\!\sum_{[C]\leq 1}  (\C_{0})_{\varphi^{4}\, \varphi}^{\O_{C}}(y,x_{1})\, (\C_{0})_{\O_{C}\,\varphi}^{\varphi^{4}}(x_{1},x_{2}) \\
&\quad-\sum_{[C]\leq 1}  (\C_{0})_{\varphi^{4}\, \varphi}^{\O_{C}}(y,x_{2})\, (\C_{0})_{\varphi\,\O_{C}}^{\varphi^{4}}(x_{1},x_{2})-\sum_{[C]< 4}  (\C_{0})_{\varphi\, \varphi}^{\O_{C}}(x_{1},x_{2})\, (\C_{0})_{\varphi^{4}\, \O_{C}}^{\varphi^{4}}(y,x_{2}) \bigg]\, .
\end{split}
\een
On the right hand side we can now substitute explicit expression for the zeroth order coefficients. With the help of Wick's theorem, one can show quite easily that many coefficients in these summations over $C$ vanish. Let ${A}=\{n_{A},w_{A}\}, {B}=\{n_{B},w_{B}\}$ and ${C}=\{n_{C},w_{C}\}$. If the condition $n_{A}+n_{B}-2k=n_{C}$ can not be fulfilled for some $k\leq \min(n_{A},n_{B})$, then
\ben\label{OPE0vanish}
(\C_{0})_{{A}\, {B} }^{{C}}(x_{1},x_{2})=0\, .
\een
It follows that $(\C_{0})_{\varphi^{4}\, \varphi}^{\O_{C}}(x)=0$ for $[C]\leq 1$.  For the sum over $[C]<4$ in the last term on the r.h.s. of \eqref{OPE1stex1} one finds, again in view of eq.\eqref{OPE0vanish}, that only the contributions with $\O_{C}\in\{\mathds{1},\varphi^{2},\varphi\partial_{\mu}\varphi\}$ are non-vanishing. Hence, we arrive at the simpler equation
\ben\label{OPE1stex1b}
\begin{split}
  (\C_{1})_{\varphi\, \varphi}^{\varphi^{4}}&(x_{1}, x_{2})= \frac{-1}{4!}\int\d^{4}y\ \bigg[ (\C_{0})_{\varphi^{4}\, \varphi\, \varphi }^{\varphi^{4}}(y,x_{1},x_{2})- (\C_{0})_{\varphi\, \varphi}^{\mathds{1}}(x_{1},x_{2})\, (\C_{0})_{\varphi^{4}\,\mathds{1}}^{\varphi^{4}}(y,x_{2}) \\
&-\sum_{\mu=1}^{4}  (\C_{0})_{\varphi\, \varphi}^{\varphi\partial_{\mu}\varphi}(x_{1},x_{2})\, (\C_{0})_{\varphi^{4}\,\varphi\partial_{\mu}\varphi}^{\varphi^{4}}(y,x_{2}) -  (\C_{0})_{\varphi\, \varphi}^{\varphi^{2}}(x_{1},x_{2})\, (\C_{0})_{\varphi^{4}\, \varphi^{2}}^{\varphi^{4}}(y,x_{2}) \bigg]\, .
\end{split}
\een
The computation of the zeroth order coefficients on the r.h.s. of this equation proceeds by Wick's theorem and is trivial. Let 
\ben
\hat{C}^{0,\infty}(x):=\int\frac{\d^{4}p}{(2\pi)^{4}} \frac{\e^{i px}}{p^{2}+m^{2}}
\een
be the position space propagator. The zeroth order coefficients are explicitly given by
\begin{align}
&(\C_{0})_{\varphi^{4}\, \varphi\, \varphi }^{\varphi^{4}}(y,x_{1},x_{2})=4\hat{C}^{0,\infty}(x_{1}-y)+4\hat{C}^{0,\infty}(x_{2}-y) +\hat{C}^{0,\infty}(x_{1}-x_{2})\\
& (\C_{0})_{\varphi\, \varphi}^{\mathds{1}}(x_{1},x_{2})\, (\C_{0})_{\varphi^{4}\,\mathds{1}}^{\varphi^{4}}(y,x_{2})=\hat{C}^{0,\infty}(x_{1}-x_{2})\\
& (\C_{0})_{\varphi\, \varphi}^{\varphi^{2}}(x_{1},x_{2})\, (\C_{0})_{\varphi^{4}\, \varphi^{2}}^{\varphi^{4}}(y,x_{2})=8\hat{C}^{0,\infty}(x_{2}-y) \\
&(\C_{0})_{\varphi\, \varphi}^{\varphi\partial_{\mu}\varphi}(x_{1},x_{2})\, (\C_{0})_{\varphi^{4}\,\varphi\partial_{\mu}\varphi}^{\varphi^{4}}(y,x_{2}) =  (x_{1}-x_{2})^{\mu}\cdot 4\,\partial_{\mu}\hat{C}^{0,\infty}(x_{2}-y) 
\end{align}
Our final result for the first order coefficient is therefore
\ben
\begin{split}
  (\C_{1})_{\varphi\, \varphi}^{\varphi^{4}}(x_{1},& x_{2})=- \frac{\int\d^{4}y}{4!}\ \bigg[ 4\hat{C}^{0,\infty}(x_{1}-y)+4\hat{C}^{0,\infty}(x_{2}-y) +\hat{C}^{0,\infty}(x_{1}-x_{2}) \\
&-\hat{C}^{0,\infty}(x_{1}-x_{2}) - 4 (x_{1}-x_{2})^{\mu} \partial_{\mu}\hat{C}^{0,\infty}(x_{2}-y)    -8\hat{C}^{0,\infty}(x_{2}-y) \Big]=0 
\end{split}
\een
where we have used the fact that $\int\d^{4} y\,  \partial_{\mu}\hat{C}^{0,\infty}(x_{2}-y) = 0$.

\paragraph{The coefficient $(\C_{1})_{\varphi\, \varphi}^{\varphi^{2}}$ :} We again use eq.\eqref{eq:thmintexp} to write
\ben\label{OPE1stex2}
\begin{split}
  (\C_{1})_{\varphi\, \varphi}^{\varphi^{2}}(x_{1}, x_{2})=
  & \frac{-1}{4!}\int\d^{4}y\ \bigg[ (\C_{0})_{\varphi^{4}\, \varphi\, \varphi }^{\varphi^{2}}(y,x_{1},x_{2})-\sum_{[C]\leq 1}  (\C_{0})_{\varphi^{4}\, \varphi}^{\O_{C}}(y,x_{1})\, (\C_{0})_{\O_{C}\,\varphi}^{\varphi^{2}}(x_{1},x_{2}) \\
&-\sum_{[C]\leq 1}  (\C_{0})_{\varphi^{4}\, \varphi}^{\O_{C}}(y,x_{2})\, (\C_{0})_{\varphi\,\O_{C}}^{\varphi^{2}}(x_{1},x_{2})-\sum_{[C]< 2}  (\C_{0})_{\varphi\, \varphi}^{\O_{C}}(x_{1},x_{2})\, (\C_{0})_{\varphi^{4}\, \O_{C}}^{\varphi^{2}}(y,x_{2}) \bigg]\, .
\end{split}
\een
Using equation \eqref{OPE0vanish}, one verifies that the sums over $C$ on the right hand side actually vanish. Our result for this coefficient is
\ben\label{Cphi21}
\begin{split}
  &(\C_{1})_{\varphi\, \varphi}^{\varphi^{2}}(x_{1}, x_{2})= -\int\frac{\d^{4}y}{4!} (\C_{0})_{\varphi^{4}\, \varphi\, \varphi }^{\varphi^{2}}(y,x_{1},x_{2}) =-\int\frac{\d^{4}y}{2} \hat{C}^{0,\infty}(x_{1}-y)\hat{C}^{0,\infty}(x_{2}-y)\\
 &=\frac{-1}{16\pi^{2}}K_{0}(\sqrt{(x_{1}-x_{2})^{2}m^{2}})=\frac{1}{16\pi^{2}} \left[ \log(\sqrt{(x_{1}-x_{2})^{2}m^{2}}/2)+ \Gamma_{\rm E} \right] + O((x_{1}-x_{2})^{2}) \, ,
\end{split}
\een
where $K_{0}$ is a modified Bessel function of the second kind and where $\Gamma_{\rm E}$ is the \emph{Euler-Mascheroni constant}. This particular coefficient is the standard example computed in the textbooks~\cite{collins,Weinberg:1996kr}. These authors present their result on Minkowski space, it is trivial to translate their results into the Euclidean context via Wick rotation. 
The explicit form of this coefficient given in the standard reference~\cite[p.~262]{collins} then corresponds to~\eqref{Cphi21}, up to a field redefinition [see remark~4, item 2)], which is 
clearly admissible.

\paragraph{The coefficient $(\C_{1})_{\varphi\, \varphi^{3}}^{\varphi^{2}}$ :} We use the same strategy as above.
\ben\label{OPE1stex3}
\begin{split}
&  (\C_{1})_{\varphi\, \varphi^{3}}^{\varphi^{2}}(x_{1}, x_{2})=-\!\!\int\! \frac{\d^{4}y}{4!} \bigg[ (\C_{0})_{\varphi^{4}\, \varphi\, \varphi^{3} }^{\varphi^{2}}(y,x_{1},x_{2})-\hspace{-.2cm} \sum_{[C]\leq 1}  (\C_{0})_{\varphi^{4}\, \varphi}^{\O_{C}}(y,x_{1})\, (\C_{0})_{\O_{C}\,\varphi^{3}}^{\varphi^{2}}(x_{1},x_{2}) \\
&\quad-\sum_{[C]\leq 3}  (\C_{0})_{\varphi^{4}\, \varphi^{3}}^{\O_{C}}(y,x_{2})\, (\C_{0})_{\varphi\,\O_{C}}^{\varphi^{2}}(x_{1},x_{2}) -\sum_{[C]< 2}  (\C_{0})_{\varphi\, \varphi^{3}}^{\O_{C}}(x_{1},x_{2})\, (\C_{0})_{\varphi^{2}\, \O_{C}}^{\varphi^{2}}(y,x_{2})\bigg]\, .
\end{split}
\een
Dropping vanishing terms in the summations, we obtain the formula
\ben
\begin{split}
  (\C_{1})_{\varphi\, \varphi^{3}}^{\varphi^{2}}(x_{1}, x_{2})= \frac{-1}{4!}\int\d^{4}y\ \bigg[ (\C_{0})_{\varphi^{4}\, \varphi\, \varphi^{3} }^{\varphi^{2}}(y,x_{1},x_{2}) &-  (\C_{0})_{\varphi^{4}\, \varphi^{3}}^{\varphi}(y,x_{2})\, (\C_{0})_{\varphi\,\varphi}^{\varphi^{2}}(x_{1},x_{2}) \\
&-  (\C_{0})_{\varphi^{4}\, \varphi^{3}}^{\varphi^{3}}(y,x_{2})\, (\C_{0})_{\varphi\,\varphi^{3}}^{\varphi^{2}}(x_{1},x_{2})\bigg]\, .
\end{split}
\een
Substituting the explicit form of the zeroth order coefficients on the right hand side, we arrive at the result
\ben\label{Cphi31}
\begin{split}
  (\C_{1})_{\varphi\, \varphi^{3}}^{\varphi^{2}}(x_{1}, x_{2})=& \frac{-1}{4!}\int\d^{4}y\bigg[  \hat{C}^{0,\infty}(y-x_{2})^{2} \Big( 24 \hat{C}^{0,\infty}(y-x_{2}) + 72 \hat{C}^{0,\infty}(y-x_{1})+ 36 \hat{C}^{0,\infty}(x_{1}-x_{2})  \Big) 
  \\
  &\qquad - \hat{C}^{0,\infty}(y-x_{2})^{2} \Big( 24 \hat{C}^{0,\infty}(y-x_{2})+108 \hat{C}^{0,\infty}(x_{1}-x_{2})  \Big) \bigg]\\
 =&- 3\int\d^{4}y\ \left[ (\hat{C}^{0,\infty}(y-x_{2}))^{2}\ \Big( \hat{C}^{0,\infty}(y-x_{1})-\hat{C}^{0,\infty}(x_{1}-x_{2}) \Big) \right]\\
  =&-\frac{3}{(2\pi)^{8}} \int\d^{4}p\,\d^{4}q \, \frac{\e^{i q (x_{1}-x_{2})}}{(p^{2}+m^{2})(q^{2}+m^{2})} \, \left( \frac{1}{(p+q)^{2}+m^{2}}- \frac{1}{p^{2}+m^{2}} \right)
 \end{split}
\een
An interesting point about this particular coefficient is that the integrals over the individual terms in \eqref{Cphi31} diverge in the ultra-violet (i.e. for $y\to x_{2}$ and for $p^{2}\to\infty$). It is not hard to check, however, that the divergent contributions cancel between the two terms, so that the total expression on the right side is indeed finite, as it should be. Here we see in a concrete example how the sums over $C$, which are subtracted on the right hand side of our perturbation formula \eqref{eq:thmintexp}, act as counter-terms, which guarantee (UV-) finiteness. 

Again, we can compare our result to the one obtained via customary Feynman diagram methods. Following the standard prescription for the computation of OPE coefficients \cite{collins}, one has to determine the Feynman diagram displayed in figure \ref{fig:Feynm}.

\begin{figure}[htbp]
\begin{center}
\begin{tikzpicture}[fill=gray!50]
\draw (.5,1) node[left]{$x_{1}$} -- (2,0);
\draw (.5,-1) -- node[below]{ $p_{1}$} (2,0);
\draw (2,0) .. controls (3.5,0.75) and (3.5,.75) .. (5,0);
\draw (2,0) .. controls (3.5,-0.75) and (3.5,-.75) .. (5,0);
\node[draw,circle,fill] (A1) at (5,0) {$x_{2}$};
\draw (A1) --  node[below]{ $p_{2}$} (6.5,-1);
\end{tikzpicture}
\caption{Feynman diagram for the computation of $(\C_{1})_{\varphi\, \varphi^{3}}^{\varphi^{2}}(x_{1}, x_{2})$}
\label{fig:Feynm}
\end{center}
\end{figure}
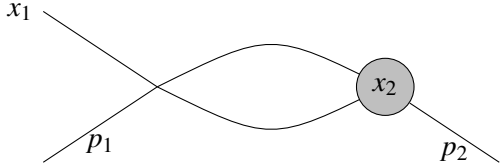

The grey circle in fig.\ref{fig:Feynm} corresponds to the vertex from the composite operator $\varphi^{3}$. The OPE coefficient $(\C_{1})_{\varphi\, \varphi^{3}}^{\varphi^{2}}(x_{1}, x_{2})$ is then obtained after amputating the external legs labelled by $p_{1}$ and $p_{2}$ form the diagram. Using BPHZ renormalization for the composite operator $\varphi^{3}$, one thereby recovers precisely the expression \eqref{Cphi31}.


\section{Conclusions and outlook}

In this paper, we have derived, within the setting of massive Euclidean $\varphi^{4}$-theory, a new formula expressing the coupling constant derivative of any OPE coefficient as a function of other OPE coefficients. Algebraically, the formula describes the deformation of the OPE algebra caused by the $\varphi^{4}$-interaction. The result can be used constructively: OPE coefficients can be computed to any perturbation order from the zeroth order ones by iterating our recursion formula. We have given a few examples of this procedure at first perturbation order. The construction method bears significance both from a conceptual as well as from a technical viewpoint, as it a) shows that OPE coefficients can be determined in a manifestly state independent manner without reference to any other quantities, such as for example correlation functions, and b) in that it provides a clear and concrete algorithm for the computation of any OPE coefficient in perturbation theory.

Concerning possible lines of future research, several generalizations and extensions of our recursion formula quickly come to mind. For example, it would be interesting to generalize the formula to massless theories, to other QFT models, to Minkowski space or possibly even to curved spacetimes. Further, one could investigate the dependence of the recursion formula on the renormalization conditions, i.e. one could study renormalization conditions other than BPHZ ones. A very exciting application of the recursion formula would be the possibility to give a non-perturbative definition of OPE coefficients, as discussed in the introduction.

\appendix

\section{Derivation of bounds on Schwinger functions with insertions}\label{APbounds}

In this appendix we derive bounds on the regularized amputated Green's functions (AG's) with operator insertions $\O_{A_{1}}(x_{1}),\ldots, \O_{A_{N}}(x_{N})$. We distinguish the cases of regularization w.r.t. the total diagonal $x_{1}=\ldots=x_{N}$ (related to the $F^{\Lambda,\Lambda_{0}}$-functionals, see section \ref{sec:regCAG}) and w.r.t. any partial diagonal $x_{1}=\ldots=x_{M<N}$ (related to the $H^{\Lambda,\Lambda_{0}}$-functionals, see section \ref{subsec:subdiv}). For the sake of simplicity, we set $g=1$ in the following.

\subsection{Regularization on the total diagonal}

As discussed in section \ref{sec:regCAG}, regularization on the total diagonal of Schwinger functions with insertions is related to the functionals $F^{\Lambda,\Lambda_{0}}_{D}(\bigotimes_{i=1}^{N}\O_{A_{i}})$. In the following we will derive inductive bounds which imply the estimates \eqref{SD1} and \eqref{FIR} on the short- and large distance behavior of the AG's in the spacetime arguments.

\begin{thm}\label{thmbound}
For any $D\leq D'=[A_{1}]+\ldots+[A_{N}]$, any $s\in\mathbb{N}_{0}$, any $r\in\mathbb{N}$ and $t\in\{0,1\}$, the bound
\ben\label{eqbound}
\begin{split}
\Big|\,\partial_{\La}^{t}\partial_{\vecp}^{w}\F^{\Lambda,\Lambda_{0}}_{D,2n,l}&(\bigotimes_{i=1}^{N}\O_{A_{i}}(x_{i});\vec{p})\, \Big|\leq (\Lambda+m)^{D-2n-|w|+1-t-\frac{1}{r}}\  \Pol_{1}\left(\log\frac{\Lambda+m}{m}\right) \Pol_{2}\left(\frac{|\vec{p}|}{\Lambda+m}\right)  \\
&\times \frac{  \max\limits_{1\leq i\leq N}|x_{i}-x_{N}|^{\max(|w|,D+1)} \cdot (m \min\limits_{1\leq i<j\leq N} |x_{i}-x_{j}|)^{-s} }{\min\limits_{1\leq i<j\leq N}|x_{i}-x_{j}|^{D'-D-1+\frac{1}{r}+\max(|w|,D+1)}}\sup\Big(1,(\Lambda+m)|x_{N}|\Big)^{|w|}
\end{split}
\een
holds, where $\Pol_{i}(x)$ are polynomials in $x$ with positive coefficients.
\end{thm}
\begin{rem}
For versions of these bounds with a stronger control over the numerical factors in the polynomials, see~\cite{Holland:2013we} and also~\cite{Hollands:2011gf,Holland:2012vw} (for the case $N=2$). The bound implies the following conclusions on the scaling behavior of Schwinger functions with insertions:
\begin{enumerate}

\item Setting $r=1$ and $w=0$, the bound clearly implies the scaling identity \eqref{FIR}.

\item In the case $N=2$, the bound simplifies somewhat due to the fact that
\ben
\max_{1\leq i\leq 2}|x_{i}-x_{2}|=|x_{1}-x_{2}|=\min_{1\leq i<j\leq 2}|x_{i}-x_{j}|\, .
\een
Recalling from eq.\eqref{CAGFeq} that in this case the $F$-functionals coincide with the CAG's with two insertions up to a sign, we therefore obtain the bound
\end{enumerate}
\ben\label{eqboundN=2}
\begin{split}
\Big|\,\partial_{\La}^{t}\partial_{\vecp}^{w}\L^{\Lambda,\Lambda_{0}}_{D,2n,l}(\O_{A_{1}}(x_{1})\otimes\O_{A_{2}}(x_{2});\vec{p})\, \Big|&\leq (\Lambda+m)^{D-2n-|w|+1-t-\frac{1}{r}}\  \Pol_{1}\left(\log\frac{\Lambda+m}{m}\right) \Pol_{2}\left(\frac{|\vec{p}|}{\Lambda+m}\right)  \\
&\times \frac{  m^{-s}}{|x_{1}-x_{2}|^{D'-D-1+\frac{1}{r}+s}}\, \sup\Big(1,(\Lambda+m)|x_{2}|\Big)^{|w|}
\end{split}
\een
\end{rem}
Before we come to the proof of the theorem, let us also note the following consequence:
\begin{cor}\label{corbound}
For any $D\leq D'=[A_{1}]+\ldots+[A_{N}]$, any $r\in\mathbb{N}$ and $t\in\{0,1\}$, the bound
\ben\label{eqboundcor}
\begin{split}
&\Big|\,\partial_{\La}^{t}\partial_{\vecp}^{w}\G^{\Lambda,\Lambda_{0}}_{D,2n,l}(\bigotimes_{i=1}^{N}\O_{A_{i}}(x_{i});\vec{p})\, \Big|\\
&\leq (\Lambda+m)^{D-2n-|w|+1-t-\frac{1}{r}}\  \Pol_{1}\left(\log\frac{\Lambda+m}{m}\right) \Pol_{2}\left(\frac{|\vec{p}|}{\Lambda+m}\right)\sup\Big(1,(\Lambda+m)|x_{N}|\Big)^{|w|}  \\
&\times \max\left(\frac{1}{\min\limits_{1\leq i<j\leq N}|x_{i}-x_{j}|^{D'-D-1+\frac{1}{r}}}, (\La+m)^{D'-D-1+\frac{1}{r}}  \right) \left(\frac{  \max\limits_{1\leq i\leq N}|x_{i}-x_{N}|  }{\min\limits_{1\leq i<j\leq N}|x_{i}-x_{j}|}\right)^{\max(|w|,D+1)} 
\end{split}
\een
holds, where $\Pol_{i}(x)$ are polynomials in $x$ with positive coefficients.
\end{cor}
\begin{rem}
If we scale the spacetime arguments by a small factor $\varepsilon>0$, the bound scales as $\varepsilon^{-D'-D-1+\frac{1}{r}}$. Thus, we conclude that the scaling degree (cf. eq.~\eqref{sddef})
is
\ben
\operatorname{sd}(G_{D}^{\Lambda,\Lambda_{0}}(\bigotimes_{i=1}^{N}\O_{A_{i}}))
=D'-D-1
\een
\end{rem}
In view of the decomposition \eqref{GDdef},  this corollary follows straightforwardly from the bounds stated in theorem \ref{thmbound} combined with the known bounds for the CAG's with one insertion, see \eqref{CAGoneinsert} below or also~\cite{Keller:1991bz,Muller:2002he,Hollands:2011gf} for the derivation of these bounds. We now proceed to the proof of theorem~\ref{thmbound}:

\begin{proof}[Proof of theorem \ref{thmbound}:]
We use an inductive scheme, based on the renormalization group flow equations to verify the bound \eqref{eqbound}, see~\cite{Keller:1991bz,Keller:1992by,Muller:2002he,Hollands:2011gf} for more details on this procedure. To begin with, let us recall this flow equation for the functionals under consideration. In the expanded form, these are
\ben\label{FE2}
\begin{split}
&\hspace{2cm}\partial_{\Lambda}\partial_{\vec{p}}^{w}\F^{\Lambda,\Lambda_{0}}_{D,2n,l}(\bigotimes_{i=1}^{N}\O_{A_{i}}; p_{1},\ldots,p_{2n})=\\
\vspace{0.4cm}\\
&= \left(\atop{2n+2}{2}\right) \, \int_{k} \dot{C}^{\Lambda}(k)\ \partial_{\vec{p}}^{w}\F^{\Lambda,\Lambda_{0}}_{D,2n+2,l-1}(\bigotimes_{i=1}^{N}\O_{A_{i}}; k, -k,  p_{1},\ldots,p_{2n})\\
&- \mathbb{S}\, \Bigg[\sum_{\substack{l_{1}+l_{2}=l \\ n_{1}+n_{2}=n+1}}\!\!\!\!\! 4n_{1}n_{2}\sum_{w_{1}+w_{2}+w_{3}=w}\!\!\!\!\!\!\! c_{\{w_{j}\}}  \partial_{\vec{p}}^{w_{1}}\F^{\Lambda,\Lambda_{0}}_{D,2n_{1},l_{1}}(\bigotimes_{i=1}^{N}\O_{A_{i}}; q,  p_{1},\ldots,p_{2n_{1}-1})\\
&\hspace{5cm}\times  \partial_{\vec{p}}^{w_{2}}\dot{C}^{\Lambda}(q)\,   \partial_{\vec{p}}^{w_{3}}\L^{\Lambda,\Lambda_{0}}_{2n_{2},l_{2}}(  p_{2n_{1}},\ldots,p_{2n}) \\
&-\sum_{\substack{l_{1}+\ldots+l_{N}=l \\ n_{1}+\ldots+n_{N}=n+1}} \sum_{1\leq a<b\leq N}4n_{a}n_{b} \partial_{\vec{p}}^{w} \int_{k}\, \L^{\Lambda,\Lambda_{0}}_{2n_{a},l_{a}}(\O_{A_{a}}; k,  p_{2n_{a-1}},\ldots,p_{2n_{a}-1})\,  \dot{C}^{\Lambda}(k)\\
&\quad \times \L^{\Lambda,\Lambda_{0}}_{2n_{b},l_{b}}(\O_{A_{b}}; -k,   p_{2n_{b-1}},\ldots,p_{2n_{b}-1})\!\!\!\!\! \prod_{c\in\{1,\ldots,N\}\setminus\{a,b\}} \!\!\!\!\!  \L^{\Lambda,\Lambda_{0}}_{2n_{c},l_{c}}(\O_{A_{c}};   p_{2n_{c-1}},\ldots,p_{2n_{c}-1})  \Bigg]
\end{split}
\een
with $q=p_{2n_{1}}+\ldots+p_{2n}$ and where $\mathbb{S}$ is a symmetrization operator acting on functions of the momenta $(p_{1},\ldots, p_{2n})$ by taking the mean value over all permutations $\pi$ of $1,\ldots, 2n$ satisfying $\pi(1)<\pi(2)<\ldots<\pi(2n_{1}-1)$ and $\pi(2n_{1})<\ldots<\pi(2n)$.

The induction procedure is to go up in $2n+l$, and for fixed $2n+l$ to ascend in $l$. In order to be able to make use of the boundary conditions \eqref{Gbound1} and \eqref{Gbound2}, we first derive the bound for $x_{N}=0$. The bound for $x_{N}\neq 0$ can then be obtained, in the very end, with the help of the translation properties \eqref{Ftrans}.
 We start considering (from here on we set $x_{N}=0$).
 
 When integrating eq.\eqref{FE2} over $\Lambda$, we have to distinguish three cases:

\begin{enumerate}[(A)]

\item  Contributions with $2n+|w|> D$ are referred to as \emph{irrelevant}. Here the boundary conditions are given at $\Lambda=\Lambda_{0}$, see eq.\eqref{Gbound2}. Therefore, we integrate over $\Lambda'$ from $\Lambda$ to $\Lambda_{0}$ in this case.

 \item  Contributions with $2n+|w|\leq D$ are referred to as \emph{relevant}.  The boundary conditions for relevant terms, eq.\eqref{Gbound1}, are given at $\Lambda=0$ and at vanishing external momentum, $\vec{p}=\vec{0}$. Thus, we will integrate over $\Lambda'$ from $0$ to $\Lambda$ in this case.

\item Contributions with $2n+|w|\leq D$ and $\vec{p}\neq \vec{0}$ will be obtained from (A),(B) with the help of a Taylor expansion in $\vec{p}$.
 \end{enumerate}

\paragraph{Irrelevant terms ($2n+|w|>D$):} 

\begin{itemize}

\item \underline{\textsf{First term on the r.h.s. of the flow equation:}}

For the sake of brevity, let us define the shorthand
\ben\label{defm}
\m(\vec{x},D',D,w):=\frac{\max\limits_{1\leq i\leq N}|x_{i}|^{\max(|w|,D+1)}\cdot (m \min\limits_{1\leq i<j\leq N} |x_{i}-x_{j}|)^{-s}}{\min_{1\leq i<j\leq N}|x_{i}-x_{j}|^{D'-D-1+\frac{1}{r}+\max(|w|,D+1)}}\, .
\een
Inserting the inductive bound \eqref{eqbound} for the first term on the r.h.s. of the flow equation, we obtain (recall the notation $|\vec{p}|$ from our conventions section on page \pageref{multider})
%
\ben\label{ineq1st}
\begin{split}
&\Big| \left(\atop{2n+2}{2}\right)\int\d^{4}k \ \dot{C}^{\Lambda}(k) \, \partial_{\vec{p}}^{w}\F_{D,2n+2,l-1}^{\Lambda,\Lambda_{0} }(\bigotimes_{i=1}^{N}\O_{A_{i}}; k,-k, p_{1},\ldots, p_{2n})\Big|\\
&\leq  \left(\atop{2n+2}{2}\right)\int\d^{4}k\,  (\Lambda+m)^{D-2n-1 -\frac{1}{r}-|w|} \, \m(\vec{x},D',D,w)\\
&\qquad\times \Pol_{1}\left(\log\frac{\Lambda+m}{m}\right) \Pol_{2}\left(\frac{|\vec{p}|_{2n+2}}{\Lambda+m}\right)\, \frac{2}{\Lambda^{3}} \e^{-\frac{k^{2}+m^{2}}{\Lambda^{2}} }\\
&\leq  \int\d^{4}k\, (\Lambda+m)^{D-2n-1-|w|-\frac{1}{r}} \Lambda^{-3}\ \m(\vec{x},D',D,w) \Pol_{1}\left(\log\frac{\Lambda+m}{m}\right) \Pol_{2}\left(\frac{|\vec{p}|+|k|}{\Lambda+m}\right) \e^{-\frac{k^{2}+m^{2}}{\Lambda^{2}} }\\
&\leq \int\d^{4}(k/\Lambda)\, (\Lambda+m)^{D-2n-|w|-\frac{1}{r}} \,\m(\vec{x},D',D,w) \Pol_{1}\left(\log\frac{\Lambda+m}{m}\right) \Pol_{2}\left(\frac{|\vec{p}|+|k|}{\Lambda+m}\right) \e^{-\frac{k^{2}+m^{2}}{\Lambda^{2}} }\\
&\leq  (\Lambda+m)^{D-2n-|w|-\frac{1}{r}}\, \m(\vec{x},D',D,w)\, \Pol_{1}\left(\log\frac{\Lambda+m}{m}\right) \Pol_{2}\left(\frac{|\vec{p}|}{\Lambda+m}\right)
\end{split}
\een
Here we have also used the equation
\ben
\dot{C}^{\La}(k)=\frac{2}{\La^{3}}\, \e^{-\frac{k^{2}+m^{2}}{\La^{2}}}\, .
\een
In the second inequality of \eqref{ineq1st} we used the fact that $|\vec{p}|_{2n+2}\leq |\vec{p}|+|k|$, and we absorbed some numerical factors into the polynomials $\Pol_{i}$\footnote{We denote these new, slightly larger polynomials again by $\Pol_{i}$, by abuse of notation. This convention will be used regularly in the following.}. To arrive at the last line, we have used the exponential factor $\e^{-k^{2}/\Lambda^{2}}$ to bound the integral over powers of $|k|/\Lambda$ via
\ben\label{gammaloop}
\int_{0}^{\infty}\e^{-x^{2}} \, x^{n} = \Gamma\left(\frac{n+1}{2}\right)/2\, ,
\een
where the numerical factors have been absorbed into the (new) polynomials $\Pol_{i}$. The inequality \eqref{ineq1st} confirms that the contribution from the first term on the right hand side of the flow equation satisfies the claimed bound, \eqref{eqbound}, in the case $t=1$ and $x_{N}=0$. 

To verify the bound with $t=0$ (i.e. no $\La$-derivative), we have to integrate once more over $\La'$ between $\La$ and $\La_{0}$.
This integral can be estimated with the help of the following lemma:
\begin{lemma}\label{lemA}
Let $s\in\mathbb{N}_{0}$ and $r\in\mathbb{R}$ with $r>0$. Then
\ben
\int_{a}^{b}\d x\, x^{-r-1} (\log x)^{s} \leq  \frac{a^{-r}}{r} \bigg( (\log a)^{s} + \frac{s}{r} (\log a)^{s-1} + \frac{s (s-1)}{r^{2}}(\log a)^{s-2}+\ldots +\frac{s!}{r^{s}} \bigg)
\een
where $1\leq a \leq b$.
\end{lemma}
A proof is given in~\cite{Muller:2002he} for the case $r\in\mathbb{N}$, which generalizes straightforwardly to $r\in\mathbb{R}_{+}$. Combinig lemma \ref{lemA} with \eqref{ineq1st}, we find
\ben
\begin{split}
&\Big| \int_{\La}^{\La_{0}}\d\La'\, \left(\atop{2n+2}{2}\right)\int\d^{4}k \ \dot{C}^{\Lambda'}(k) \, \partial_{\vec{p}}^{w}\F_{D,2n+2,l-1}^{\Lambda',\Lambda_{0} }(\bigotimes_{i=1}^{N}\O_{A_{i}}; k,-k, p_{1},\ldots, p_{2n})\Big|\\
& \leq (\Lambda+m)^{D-2n-|w|+1-\frac{1}{r}}\, \m(\vec{x},D',D,w)\, \Pol_{1}\left(\log\frac{\Lambda+m}{m}\right) \Pol_{2}\left(\frac{|\vec{p}|}{\Lambda+m}\right) \, ,
\end{split}
\een
which is consistent with our inductive bound, \eqref{eqbound}, in the case $t=0$ and $x_{N}=0$.

\item \underline{\textsf{Second term on the r.h.s. of the flow equation:}}

Substituting the known bounds for the CAG's without insertion~\cite{Keller:1990ej,Kopper:1997vg,Muller:2002he, Kopper:2009um},
\ben\label{CAGnoinsert}
|\partial_{\vec{p}}^{w}\L^{\Lambda,\Lambda_{0}}_{2n,l}(  p_{1},\ldots,p_{2n-1})| \leq (\Lambda+m)^{4-2n-|w|}\ \Pol_{1}\left(\log\frac{\Lambda+m}{m}\right) \Pol_{2}\left(\frac{|\vec{p}|}{\Lambda+m}\right)
\een
as well as the inductive bound \eqref{eqbound}, we find (recall that we write $q=p_{2n_{1}}+\ldots+p_{2n}$)
\ben\label{ineq2nd}
\begin{split}
&\Big|\sum_{\substack{l_{1}+l_{2}=l \\ n_{1}+n_{2}=n+1\\ w_{1}+w_{2}+w_{3}=w }}\!\!\!\!\! 4n_{1}n_{2} c_{\{w_{j}\}} \ \partial_{\vec{p}}^{w_{1}}\F^{\Lambda,\Lambda_{0}}_{D,2n_{1},l_{1}}(\bigotimes_{i=1}^{N}\O_{A_{i}}; q,  p_{1},\ldots,p_{2n_{1}-1})\\
 &\qquad\times \partial_{\vec{p}}^{w_{2}}\dot{C}^{\Lambda}(q)\,   \partial_{\vec{p}}^{w_{3}}\L^{\Lambda,\Lambda_{0}}_{2n_{2},l_{2}}(  p_{2n_{1}},\ldots,p_{2n}) \Big| \\
&\leq \hspace{-.3cm} \sum_{\substack{l_{1}+l_{2}=l \\ n_{1}+n_{2}=n+1\\ w_{1}+w_{2}+w_{3}=w }}\!\!\!\!\! 4n_{1}n_{2} c_{\{w_{j}\}} (\Lambda+m)^{D-2n_{1}-|w_{1}|+1-\frac{1}{r}}\, \Pol_{1}\left(\log\frac{\Lambda+m}{m}\right) \Pol_{2}\left(\frac{|\vec{p}|}{\Lambda+m}\right) \m(\vec{x},D',D,w_{1}) \\
&\times   (\Lambda+m)^{-3-|w_{2}|}\Pol_{3}(\frac{|q|}{\Lambda+m}) \cdot (\Lambda+m)^{4-2n_{2}-|w_{3}|}\ \Pol_{4}\left(\log\frac{\Lambda+m}{m}\right) \Pol_{5}\left(\frac{|\vec{p}|}{\Lambda+m}\right)\\
&\leq\, (\Lambda+m)^{D-2n-|w|-\frac{1}{r}} \m(\vec{x},D',D,w) \Pol_{1}\left(\log\frac{\Lambda+m}{m}\right) \Pol_{2}\left(\frac{|\vec{p}|}{\Lambda+m}\right)
\end{split}
\een
Here we used the inequality
\ben
|\partial_{\vec{p}}^{w}\dot{C}^{\Lambda}(q)| \leq (\La+m)^{-3-|w|} \Pol\left(\frac{|q|}{\La+m}\right)
\een
as well as $|q|\leq |\vec{p}|$, and also the fact that
\ben\label{xiprop}
\m(\vec{x},D',D,w_{1})\leq \m(\vec{x},D',D,w)
\een 
for $|w_{1}|\leq |w|$. Recall also that, by a sight abuse of notation, we choose new polynomials $\Pol_{i}$ in different steps of these inequalities, which allows us to absorb numerical factors. We conclude that also the second term in the flow equation satisfies a bound consistent with \eqref{eqbound} in the case $t=1$, $x_{N}=0$.

Again, we verify the case $t=0$ by integrating the inequality \eqref{ineq2nd} over $\La'$ and using lemma \ref{lemA}, which yields the estimate
\ben
\begin{split}
&\Big|\int_{\La}^{\La_{0}}\d\La' \sum_{\substack{l_{1}+l_{2}=l \\ n_{1}+n_{2}=n+1\\ w_{1}+w_{2}+w_{3}=w }}\!\!\!\!\! 4n_{1}n_{2} c_{\{w_{j}\}} \ \partial_{\vec{p}}^{w_{1}}\F^{\Lambda',\Lambda_{0}}_{D,2n_{1},l_{1}}(\bigotimes_{i=1}^{N}\O_{A_{i}}; q,  p_{1},\ldots,p_{2n_{1}-1})\\
 &\qquad\times \partial_{\vec{p}}^{w_{2}}\dot{C}^{\Lambda'}(q)\,   \partial_{\vec{p}}^{w_{3}}\L^{\Lambda',\Lambda_{0}}_{2n_{2},l_{2}}(  p_{2n_{1}},\ldots,p_{2n}) \Big| \\
&\leq\, (\Lambda+m)^{D-2n-|w|+1-\frac{1}{r}}\,\m(\vec{x},D',D,w)\,   \Pol_{1}\left(\log\frac{\Lambda+m}{m}\right) \Pol_{2}\left(\frac{|\vec{p}|}{\Lambda+m}\right)
\end{split}
\een

We conclude that this contribution reproduces the inductive bound (with $x_{N}=0$ and $t=0$) as well.

\item \underline{\textsf{Third term on the r.h.s. of the flow equation:}}

In order to bound this \emph{source term}, we first derive the following estimate on the momentum integral. To keep formulas at a reasonable length, we will use the notation
\ben\label{vecpi}
\vec{p}_{i}=(p_{2n_{i-1}},\ldots,p_{2n_{i}-1})
\een
in the following, where $i$ takes values between $1$ and $N$ and where we set $p_{2n_{0}}:=p_{1}$ and $p_{2n_{N}-1}:=p_{2n}$.
\begin{lemma}\label{lem}
Let $n_{1}+\ldots+n_{N}=n+1$ and $l_{1}+\ldots+l_{N}=l$.
\ben\label{lemeq}
\begin{split}
&\bigg|\ \partial^{w}_{\vec{p}}\int_{k}\sum_{1\leq a<b\leq N} \L^{\Lambda,\Lambda_{0}}_{2n_{a},l_{a}}(\O_{A_{a}}(x_{a}); k,  \vec{p}_{a})\,  \dot{C}^{\Lambda}(k)\,  \L^{\Lambda,\Lambda_{0}}_{2n_{b},l_{b}}(\O_{A_{b}}(x_{b}); -k,   \vec{p}_{b})\\
&\hspace{2cm}\times \prod_{c\in\{1,\ldots,N\}\setminus\{a,b\}} \L^{\Lambda,\Lambda_{0}}_{2n_{c},l_{c}}(\O_{A_{c}}(x_{c});  \vec{p}_{c})\ \bigg|\\
&\leq \frac{\max\limits_{1\leq i\leq N}|x_{i}|^{|w|} \cdot (m\min\limits_{1\leq i<j\leq N}|x_{i}-x_{j}|)^{-s} }{\min\limits_{1\leq i<j\leq N}|x_{i}-x_{j}|^{D'-D+|w|-1+\frac{1}{r}}  } \ (\Lambda+m)^{D-2n-|w|-\frac{1}{r}}\Pol_{1}(\frac{\log(\Lambda+m)}{m})\Pol_{2}(\frac{|\vec{p}|}{\Lambda+m})
\end{split}
\een
where $D\leq D'=[A_{1}]+\ldots+[A_{N}]$ and where $s\in\mathbb{N}_{0}$ and $r\in\mathbb{N} $. 
\end{lemma}

\begin{proof}[Proof of Lemma \ref{lem}:]
Using the translation properties of the CAG's with one insertion, eq.\eqref{CAGtrans}, we obtain:
\ben\label{3rdirr1}
\begin{split}
&\bigg|\ \partial^{w}_{\vec{p}}\int_{k}\sum_{1\leq a<b\leq N} \L^{\Lambda,\Lambda_{0}}_{2n_{a},l_{a}}(\O_{A_{a}}(x_{a}); k,  \vec{p}_{a})\,  \dot{C}^{\Lambda}(k)\,  \L^{\Lambda,\Lambda_{0}}_{2n_{b},l_{b}}(\O_{A_{b}}(x_{b}); -k,   \vec{p}_{b})\\
&\hspace{2cm}\times \prod_{c\in\{1,\ldots,N\}\setminus\{a,b\}} \L^{\Lambda,\Lambda_{0}}_{2n_{c},l_{c}}(\O_{A_{c}}(x_{c});  \vec{p}_{c})\ \bigg|\\
&=\, \bigg| \int_{k }\sum_{\atop{1\leq a<b\leq N}{w_{1}+w_{2}=w}} c_{\{w_{i}\}} \pa^{w_{1}}_{\vecp} \e^{i (p_{1}+\ldots+p_{2n_{1}-1}) x_{1}+ \ldots + i (p_{2n_{N-1}}+\ldots+p_{2n_{N}-1}) x_{N}+ i k (x_{a}-x_{b})}\\
&\times \pa^{w_{2}}_{\vecp} \Big(\L^{\Lambda,\Lambda_{0}}_{2n_{a},l_{a}}(\O_{A_{a}}(0); k,  \vec{p}_{a})\,  \dot{C}^{\Lambda}(k)\,  \L^{\Lambda,\Lambda_{0}}_{2n_{b},l_{b}}(\O_{A_{b}}(0); -k,   \vec{p}_{b}) \prod_{c\neq a,b} \L^{\Lambda,\Lambda_{0}}_{2n_{c},l_{c}}(\O_{A_{c}}(0);   \vec{p}_{c}) \Big) \bigg|\\
&\leq\,2^{|w|} \bigg| \int_{k }\sum_{\atop{1\leq a<b\leq N}{w_{1}+w_{2}=w}} \frac{\max_{1\leq i\leq N}|x_{i}|^{|w_{1}|}}{||x_{a}-x_{b}||^{|w_{1}|+D'-D+s-1}}\, \e^{i (p_{1}+\ldots+p_{2n_{1}-1}) x_{1}+ \ldots + i k (x_{a}-x_{b})}\\
&\times \partial_{k_{\alpha}}^{|w_{1}|+D'-D+s-1}\pa^{w_{2}}_{\vecp} \Big(\L^{\Lambda,\Lambda_{0}}_{2n_{a},l_{a}}(\O_{A_{a}}(0); k,  \vec{p}_{a})\,  \dot{C}^{\Lambda}(k)\,  \L^{\Lambda,\Lambda_{0}}_{2n_{b},l_{b}}(\O_{A_{b}}(0); -k,   \vec{p}_{b})\\
&\hspace{4cm}\times \prod_{c\neq a,b} \L^{\Lambda,\Lambda_{0}}_{2n_{c},l_{c}}(\O_{A_{c}}(0);   \vec{p}_{c}) \Big) \bigg|\\
\end{split}
\een
Here we used the notation $||x||:=\max_{1\leq \mu\leq  4}|x_{\mu}|$. In the last line we used partial integration in $k_{\alpha}$, where the index $\alpha$ is defined via $||x_{a}-x_{b}||=:|(x_{a}-x_{b})_{\alpha}|$. Now transform the integration variables $k_{\mu}=(\tilde{k}_{\mu})^{r}$, where $r\in\mathbb{N}$ is odd.
%
%
\ben\label{3rdirr2}
\begin{split}
&\text{r.h.s. of \eqref{3rdirr1} } \leq 2^{|w|}\bigg| \int_{\tilde{k} }\, \left(\prod_{\mu=1}^{4} r (\tilde{k}_{\mu})^{r-1}\right)\sum_{\atop{1\leq a <b\leq N}{w_{1}+w_{2}=w} } \frac{\max_{1\leq i\leq N}|x_{i}|^{|w_{1}|}}{||x_{a}-x_{b}||^{|w_{1}|+D'-D+s-1}} \\
&\times \e^{i (p_{1}+\ldots+p_{2n_{1}-1}) x_{1}+ \ldots+ i \tilde{k}_{\mu}^{r} (x_{a}-x_{b})^{\mu}} \partial_{k_{\alpha}}^{|w_{1}|+D'-D+s-1}\pa^{w_{2}}_{\vecp} \Big(\L^{\Lambda,\Lambda_{0}}_{2n_{a},l_{a}}(\O_{A_{a}}(0); k,  \vec{p}_{a})   \\
&  \times \dot{C}^{\Lambda}(k) \L^{\Lambda,\Lambda_{0}}_{2n_{b},l_{b}}(\O_{A_{b}}(0); -k,   \vec{p}_{b}) \prod_{c\neq a,b} \L^{\Lambda,\Lambda_{0}}_{2n_{c},l_{c}}(\O_{A_{c}}(0);   \vec{p}_{c}) \Big)_{k_{\mu}=\tilde{k}_{\mu}^{r}} \bigg|
\end{split}
\een
We now integrate by parts in $\tilde{k}_{\alpha}$:
\ben\label{3rdirr3}
\begin{split}
&\text{r.h.s. of \eqref{3rdirr2} } \leq
2^{|w|}r^{4}  \int_{\tilde{k} } \bigg| \sum_{\atop{1\leq a<b\leq N}{w_{1}+w_{2}=w} }  \left(\int_{0}^{\tilde{k}_{\alpha}} \e^{i \tau^{r} (x_{a}-x_{b})_{\alpha}}\d \tau\right) \frac{  \max_{1\leq i\leq N}|x_{i}|^{|w_{1}|} }{||x_{a}-x_{b}||^{D'-D+s-1+|w_{1}|}}
 \\
&\times \partial_{\tilde{k}_{\alpha}}\bigg[  \left(\prod_{\mu=1}^{4}  ( \tilde{k}_{\mu})^{r-1}\right) \partial_{k_{\alpha}}^{D'-D+s-1+|w_{1}|}\pa^{w_{2}}_{\vecp}\bigg( \L^{\Lambda,\Lambda_{0}}_{2n_{a},l_{a}}(\O_{A_{a}}(0); k,  \vec{p}_{a})\,  \dot{C}^{\Lambda}(k)\\
&\hspace{4cm}\times\L^{\Lambda,\Lambda_{0}}_{2n_{b},l_{b}}(\O_{A_{b}}(0); -k,   \vec{p}_{b}) \prod_{c\neq a,b} \L^{\Lambda,\Lambda_{0}}_{2n_{c},l_{c}}(\O_{A_{c}}(0);   \vec{p}_{c})\bigg)_{k_{\mu}=\tilde{k}_{\mu}^{r}}\bigg]\bigg|
\end{split}
\een
We then make use of the following estimate:
\ben\label{estimate1}
\Big|\int_{0}^{\tilde{k}_{\alpha}} \e^{i \tau^{r} (x_{a}-x_{b})_{\alpha}}\d \tau \Big| = \Big|\int_{0}^{\tilde{k}_{\alpha}(x_{a}-x_{b})_{\alpha}^{\frac{1}{r}}} \e^{i \tau^{r} } (x_{a}-x_{b})_{\alpha}^{-\frac{1}{r}}\d \tau \Big|\leq \frac{3}{|(x_{a}-x_{b})_{\alpha}|^{\frac{1}{r}}}
\een
\begin{proof}[Proof of estimate \eqref{estimate1}:]
We decompose the integral as follows:
\ben
\int_{0}^{a} e^{i \tau^{r} }\d \tau = \int_{0}^{1} e^{i \tau^{r}}\d \tau + \int_{1}^{|a|} e^{i \tau^{r}}\d \tau
\een
The first integral on the r.h.s. can be estimated trivially. To estimate the second contribution, we make use of (a special case of) the \emph{van der Corput Lemma}~\cite{stein1993harmonic}:
\ben
\Big|  \int_{1}^{|a|} e^{i \tau^{r}}\d \tau \Big| \leq 3
\een
The inequality \eqref{estimate1} then follows immediately.
\end{proof}
We thus arrive at
\ben\label{3rdirr4}
\begin{split}
&\text{r.h.s. of \eqref{3rdirr3} } \leq
2^{|w|+2}r^{5} \bigg| \int_{\tilde{k} } \sum_{\atop{1\leq a<b\leq N}{w_{1}+w_{2}=w}} \bigg[ \frac{ \max_{1\leq i\leq N}|x_{i}|^{|w_{1}|}\left(\prod_{\mu=1}^{4}  (\tilde{k}_{\mu})^{r-1}\right)}{||x_{a}-x_{b}||^{D'-D+|w_{1}|+s+\frac{1}{r}-1}} \\
&\times \left( \tilde{k}_{\alpha}^{-1} +r \tilde{k}_{\alpha}^{r-1} \partial_{k_{\alpha}}  \right)  \partial_{k_{\alpha}}^{D'-D+|w_{1}|+s-1}\pa^{w_{2}}_{\vecp}\bigg( \L^{\Lambda,\Lambda_{0}}_{2n_{a},l_{a}}(\O_{A_{a}}(0); k,  \vec{p}_{a})\,  \dot{C}^{\Lambda}(k) \\
&\times \L^{\Lambda,\Lambda_{0}}_{2n_{b},l_{b}}(\O_{A_{b}}(0); -k,   \vec{p}_{b}) \prod_{c\neq a,b} \L^{\Lambda,\Lambda_{0}}_{2n_{c},l_{c}}(\O_{A_{c}}(0);   \vec{p}_{c})\bigg)_{k_{\mu}=\tilde{k}_{\mu}^{r}}\bigg]\bigg|
\end{split}
\een
Inserting the known bounds for the CAG's with one insertion~\cite{Keller:1991bz,Muller:2002he,Hollands:2011gf}
\ben\label{CAGoneinsert}
|\partial_{\vec{p}}^{w}\L^{\Lambda,\Lambda_{0}}_{2n,l}( \O_{A}, \vec{p})| \leq (\Lambda+m)^{[A]-2n-|w|}\ \Pol_{1}\left(\log\frac{\Lambda+m}{m}\right) \Pol_{2}\left(\frac{|\vec{p}|}{\Lambda+m}\right)
\een
and using the elementary estimate
\ben
|\partial^{w}\dot{C}^{\La}(k)| \leq \La^{-|w|-3}\, \Pol(\frac{|k|}{\La})\, \e^{-\frac{k^{2}+m^{2}}{\La^{2}}}
\een
  yields:
\ben\label{3rdirr5}
\begin{split}
&\text{r.h.s. of \eqref{3rdirr4} } \leq
 \bigg| \int_{\tilde{k} }  \bigg[ \frac{ \max_{1\leq i\leq N}|x_{i}|^{|w|}\left(\prod_{\mu=1}^{4}  (\tilde{k}_{\mu})^{r-1}\right)}{\min\limits_{1\leq  i < j\leq N}|x_{i}-x_{j}|^{D'-D+|w|+s+\frac{1}{r}-1}} \left( \tilde{k}_{\alpha}^{-1} +r \tilde{k}_{\alpha}^{r-1} /\Lambda  \right) \\
 &\times (\Lambda+m)^{D-2n-2} \Lambda^{-|w|-s-2} \Pol_{1}(\log\frac{\Lambda+m}{m})\Pol_{2}(\frac{|k|+|\vec{p}|}{\Lambda+m})\Pol_{3}(\frac{|k|}{\Lambda}) \e^{-\frac{k^{2}+m^{2}}{\Lambda^{2}}}\bigg]_{k_{\mu}=\tilde{k}_{\mu}^{r}}  \bigg|\\
&\leq  \bigg| \int_{(\tilde{k}/\Lambda^{\frac{1}{r}}) }  \bigg[ \frac{ \max_{1\leq i\leq N}|x_{i}|^{|w|}\left(\prod_{\mu=1}^{4}  (\tilde{k}_{\mu}/\Lambda^{\frac{1}{r}})^{r-1}\right)}{\min\limits_{1\leq  i < j\leq N}|x_{i}-x_{j}|^{D'-D+|w|+\frac{1}{r}-1+s}} \left( (\tilde{k}_{\alpha}/\Lambda^{\frac{1}{r}})^{-1} +r (\tilde{k}_{\alpha}/\Lambda^{\frac{1}{r}})^{r-1}  \right)\\
&\times  (\Lambda+m)^{D-2n-|w|-\frac{1}{r}}\, m^{-s}\, \Pol_{1}(\frac{\log(\Lambda+m)}{m})\Pol_{2}(\frac{|\vec{p}|}{\Lambda+m}) \Pol_{3}(\frac{|k|}{\La})\e^{-\frac{k^{2}}{\Lambda^{2}}}\bigg]_{k_{\mu}=\tilde{k}_{\mu}^{r}}  \bigg|\\
&\leq \frac{\max\limits_{1\leq i\leq N}|x_{i}|^{|w|}\, (\min\limits_{1\leq i<j\leq N}|x_{i}-x_{j}|\cdot m)^{-s} }{\min\limits_{1\leq  i < j\leq N}|x_{i}-x_{j}|^{D'-D+|w|+\frac{1}{r}-1}} \ (\Lambda+m)^{D-\frac{1}{r}-2n-|w|}\Pol_{1}(\frac{\log(\Lambda+m)}{m})\Pol_{2}(\frac{|\vec{p}|}{\Lambda+m})
\end{split}
\een
In the second inequality we used the exponential $\e^{-m^{2}/\La^{2}}$ in order to replace inverse powers of $\La$ by inverse powers of $\La+m$, i.e. for any $a\in\mathbb{N}$
\ben
\left(\frac{\La+m}{\La}\right)^{a}\e^{-m^{2}/\La^{2}} \leq 2^{a}\, \sqrt{a!}\, .
\een
To obtain the last inequality in \eqref{3rdirr5} we have estimated the loop integral via \eqref{gammaloop} and have absorbed this numerical factor into the (new) polynomial coefficients. This finishes the proof of the lemma for the case of $r$ odd. To see that the bound also holds for any even $r\in\mathbb{N}$, we note that by our previous discussions it holds for both $r+1$ and $r-1$. Depending on whether $(\La+m)/\min|x_{i}-x_{j}|$ is greater or smaller than one, one of these cases therefore implies the bound for even $r$.
\end{proof}
With lemma \ref{lem} at hand, we conclude that also the source terms (i.e. the third term on the r.h.s. of the flow equation) satisfies a bound that is consistent with our claim, \eqref{eqbound}, in the case $t=1$ and $x_{N}=0$. 

The case $t=0$ again follows by estimating the $\Lambda$-integral over the source terms in the irrelevant case,
\ben
\begin{split}
&\bigg|\ \int_{\Lambda}^{\Lambda_{0}}\d\Lambda'\,  \partial^{w}_{\vec{p}}\int_{k}\sum_{1\leq a<b\leq N} \L^{\Lambda',\Lambda_{0}}_{2n_{a},l_{a}}(\O_{A_{a}}(x_{a}); k,  \vec{p}_{a})\,  \dot{C}^{\Lambda'}(k)\,  \L^{\Lambda',\Lambda_{0}}_{2n_{b},l_{b}}(\O_{A_{b}}(x_{b}); -k,   \vec{p}_{b})\\
&\qquad\times \prod_{c\in\{1,\ldots,N\}\setminus\{a,b\}} \L^{\Lambda',\Lambda_{0}}_{2n_{c},l_{c}}(\O_{A_{c}}(x_{c});  \vec{p}_{c})\ \bigg|\\
&\leq\int_{\Lambda}^{\Lambda_{0}}\d\Lambda'\ \frac{\max\limits_{1\leq i\leq N}|x_{i}|^{|w|}\, (\min\limits_{1\leq i<j\leq N}|x_{i}-x_{j}|\cdot m)^{-s}  }{\min\limits_{1\leq  i < j\leq N}|x_{i}-x_{j}|^{D'-D+|w|+\frac{1}{r}-1}}  \\
&\qquad\times (\Lambda'+m)^{D-\frac{1}{r}-2n-|w|}\Pol_{1}(\frac{\log(\Lambda'+m)}{m})\Pol_{2}(\frac{|\vec{p}|}{\Lambda'+m})\\
&\leq \frac{\max\limits_{1\leq i\leq N}|x_{i}|^{|w|}\, (\min\limits_{1\leq i<j\leq N}|x_{i}-x_{j}|\cdot m)^{-s}  }{\min\limits_{1\leq  i < j\leq N}|x_{i}-x_{j}|^{D'-D+|w|+\frac{1}{r}-1}}  (\Lambda+m)^{D+1-\frac{1}{r}-2n-|w|}\Pol_{1}(\frac{\log(\Lambda+m)}{m})\Pol_{2}(\frac{|\vec{p}|}{\Lambda+m})\, ,
\end{split}
\een
where we applied lemma \ref{lemA} in the last line. We see that the estimate is consistent with the inductive bound \eqref{eqbound} in the case $t=0$, $x_{N}=0$.
\end{itemize}

\paragraph{Relevant terms ($2n+|w|\leq D$) at vanishing external momentum $\vec{p}=\vec{0}$:}

\begin{itemize}

\item \underline{\textsf{First term on the r.h.s. of the flow equation:}}

The case $t=1$ works just as before, i.e. we proceed as in the inequality \eqref{ineq1st}. To find a bound for the case $t=0$, we again have to perform a $\La'$ integral, but, due to the different boundary conditions in the relevant case, we now integrate between $0$ and $\La$:
\ben
\begin{split}
&\Big|\int_{0}^{\Lambda}\d\Lambda' \left(\atop{2n+2}{2}\right)\int\d^{4}k \ \dot{C}^{\Lambda'}(k) \, \partial_{\vec{p}}^{w}\F_{D,2n+2,l-1}^{\Lambda',\Lambda_{0} }(\bigotimes_{i=1}^{N}\O_{A_{i}}; k,-k, 0,\ldots, 0)\Big|\\
&\leq \int_{0}^{\Lambda}\, (\Lambda'+m)^{D-2n-|w|-\frac{1}{r}}\, \m(\vec{x},D',D,w)\, \Pol\left(\log\frac{\Lambda'+m}{m}\right) \\ &\leq (\Lambda+m)^{D-2n-|w|+1-\frac{1}{r}}\,\m(\vec{x},D',D,w)\, \Pol\left(\log\frac{\Lambda+m}{m}\right)
\end{split}
\een
in agreement with the inductive bound \eqref{eqbound}. Here we have used \eqref{ineq1st} to obtain the first inequality. To arrive at the last line, we made use of the following estimate:
\begin{lemma}\label{lemB}
Let $s\in\mathbb{N}_{0}$ and $r\in\mathbb{R}_{+}$. Then
\ben
\int_{1}^{b}\d x\, x^{r-1} (\log x)^{s} < \frac{s!}{r^{s}} + \frac{b^{r}}{r} \bigg| (\log b)^{s} - \frac{s}{r} (\log b)^{s-1} + \frac{s (s-1)}{r^{2}}(\log b)^{s-2}-\ldots +(-1)^{s}\frac{s!}{r^{s}} \bigg|
\een
where $1 \leq b$.
\end{lemma}
A proof of this lemma for $r\in\mathbb{N}_{0}$ can be found in~\cite{Muller:2002he}. The generalization to non-integer $r$ is again straightforward.

\item \underline{\textsf{Second term on the r.h.s. of the flow equation:}}

The case $t=1$ again follows from \eqref{ineq2nd}. To obtain a bound for the case $t=0$ we combine \eqref{ineq2nd} with lemma \ref{lemB}:
\ben
\begin{split}
&\Big|\int_{0}^{\Lambda}\d\Lambda'\sum_{\substack{l_{1}+l_{2}=l \\ n_{1}+n_{2}=n+1\\ w_{1}+w_{2}+w_{3}=w }}\!\!\!\!\! 4n_{1}n_{2} c_{\{w_{j}\}} \ \partial_{\vec{p}}^{w_{1}}\F^{\Lambda',\Lambda_{0}}_{D,2n_{1},l_{1}}(\bigotimes_{i=1}^{N}\O_{A_{i}}; \vec{0})\,  \partial_{\vec{p}}^{w_{2}}\dot{C}^{\Lambda'}(0)\,   \partial_{\vec{p}}^{w_{3}}\L^{\Lambda',\Lambda_{0}}_{2n_{2},l_{2}}(  \vec{0}) \Big| \\
&\leq\, \int_{0}^{\Lambda}\d\Lambda' (\Lambda'+m)^{D-2n-|w|-\frac{1}{r}} \, \m(\vec{x},D',D,w)\, \Pol\left(\log\frac{\Lambda'+m}{m}\right) \\
&\leq (\Lambda+m)^{D-2n-|w|+1-\frac{1}{r}}\, \m(\vec{x},D',D,w)\, \Pol\left(\log\frac{\Lambda+m}{m}\right)
\end{split}
\een
This is again consistent with \eqref{eqbound} for $t=0$ and $x_{N}=0$.

\item \underline{\textsf{Third term on the r.h.s. of the flow equation:}}

The case $t=1$ follows from lemma \ref{lem}, and the case $t=0$ is verified with the help lemma \ref{lem} and lemma \ref{lemB}:
\ben
\begin{split}
&\bigg|\ \int_{0}^{\Lambda}\d\Lambda'\,  \partial^{w}_{\vec{p}}\int_{k}\sum_{1\leq a<b\leq N} \L^{\Lambda',\Lambda_{0}}_{2n_{a},l_{a}}(\O_{A_{a}}(x_{a}); k,  \vec{0})\,  \dot{C}^{\Lambda'}(k)\,  \L^{\Lambda',\Lambda_{0}}_{2n_{b},l_{b}}(\O_{A_{b}}(x_{b}); -k,   \vec{0})\\
&\qquad\times \prod_{c\in\{1,\ldots,N\}\setminus\{a,b\}} \L^{\Lambda',\Lambda_{0}}_{2n_{c},l_{c}}(\O_{A_{c}}(x_{c});  \vec{0})\ \bigg|\\
&\leq\int_{0}^{\Lambda}\d\Lambda'\ \frac{\max\limits_{1\leq i\leq N}|x_{i}|^{|w|}\, (m\cdot \min\limits_{1\leq i<j\leq N}|x_{i}-x_{j}|)^{-s}  }{\min\limits_{1\leq  i < j\leq N}|x_{i}-x_{j}|^{D'-D+|w|+\frac{1}{r}-1}}  \ (\Lambda'+m)^{D-\frac{1}{r}-2n-|w|}\Pol(\frac{\log(\Lambda'+m)}{m})\\
&\leq \frac{\max\limits_{1\leq i\leq N}|x_{i}|^{|w|}\, (m\cdot \min_{1\leq i<j\leq N}|x_{i}-x_{j}|)^{-s}  }{\min\limits_{1\leq  i < j\leq N}|x_{i}-x_{j}|^{D'-D+|w|+\frac{1}{r}-1}}  \ (\Lambda+m)^{D+1-\frac{1}{r}-2n-|w|}\Pol(\frac{\log(\Lambda+m)}{m})\, .
\end{split}
\een
\end{itemize}

\paragraph{Relevant case ($2n+|w|\leq D$) at non-vanishing momentum $\vec{p}\neq0$:} 

We proceed to arbitrary external momenta $p_{1},\ldots,p_{2n}$ with the help of the Taylor expansion formula
\ben\label{Taylornonzero}
\begin{split}
&|\partial_{\vec{p}}^{w}\F^{\Lambda,\Lambda_{0}}_{D,2n,l}(\bigotimes_{i=1}^{N}\O_{A_{i}}; \vec{p})|\, =\, \Big|\! \!\! \!\sum_{|\tilde{w}| \le D-2n-|w|}\! \! \frac{{\vec p}^{\,\tilde{w}}}{\tilde{w}!}\,
\,\partial_{\vec p }^{\tilde{w}+w} \F^{\Lambda,\Lambda_{0}}_{D,2n,l}(\bigotimes_{i=1}^{N}\O_{A_{i}};{\vec 0})\\
&+\hspace{-.5cm}\sum_{|\tilde{w}|=D+1-2n-|w|} \hspace{-.5cm}
{\vec p}^{\tilde{w}} \int_0^1\d\tau \frac{|\tilde{w}|}{\tilde{w}!}\,(1-\tau)^{|\tilde{w}|-1}
\,\partial_{\tau\vec p }^{\tilde{w}+w}
\F^{\Lambda,\Lambda_{0}}_{D,2n,l}(\bigotimes_{i=1}^{N}\O_{A_{i}};\tau {\vec p})\ \Big|\, .
\end{split}
\een
In view of the estimates derived above, and also using the property
\ben\label{xiprop2}
\m(\vec{x},D',D,w+\tilde{w})=\m(\vec{x},D',D,w)\quad \text{ for }|w|+|\tilde{w}|\leq D+1\, ,
\een
 it is not hard to check that the r.h.s. of \eqref{Taylornonzero} satisfies a bound consistent with our induction hypothesis \eqref{eqbound}, which finishes the proof of theorem~\ref{thmbound} for the case $x_{N}=0$. 
 
 As mentioned earlier, we can proceed to $x_{N}\neq 0$ with the help of the translation formula \eqref{Ftrans}, i.e.
\ben
\partial^{w}_{\vecp}\F^{\Lambda,\Lambda_{0}}_{D,n,l}(\bigotimes_{i=1}^{N}\O_{A_{i}}(x_{i}); p_{1},\ldots,p_{n})=\partial^{w}_{\vecp}\left[\e^{i x_{N} (p_{1}+\ldots+p_{n})}\F^{\Lambda,\Lambda_{0}}_{D,n,l}(\bigotimes_{i=1}^{N}\O_{A_{i}}(x_{i}-x_{N}); p_{1},\ldots,p_{n})\right]\, .
\een 
Distributing the momentum derivatives over the two factors on the right hand side and using the previously established bound for the case $x_{N}=0$, the inequality \eqref{eqbound} is verified easily.
\end{proof}
%
%

\subsection{Regularization on a partial diagonal}\label{appartB}

Our aim in the following is to prove the scaling properties of the partially regularized AG's claimed in \eqref{SDpartbad}, \eqref{SDpart} and \eqref{Hinfrared}. All these estimates follow from the following bound:
\begin{thm}\label{thmboundH}
Let $M<N$ and $t\in\{0,1\}$. For any $D\leq D'=[A_{1}]+\ldots+[A_{M}]$, any $s\in\mathbb{N}_{0}$ and any $r\in\mathbb{N}$, the bound
\ben\label{eqboundH}
\begin{split}
&\Big|\,\partial_{\La}^{t}\partial_{\vecp}^{w}\H^{\Lambda,\Lambda_{0}}_{2n,l}([\bigotimes_{i=1}^{M}\O_{A_{i}}(x_{i})]_{D};\bigotimes_{j=M+1}^{N}\O_{A_{j}}(x_{j}); \vec{p})\, \Big| \leq  \frac{ (\Lambda+m)^{-t-\frac{1}{r}-2n-|w|} \, m^{-s} } {\min\limits_{
 \substack{ 1\leq i\leq M \\ M+1\leq j \leq N } } |x_{i}-x_{j}|^{D+s} 
}\\
& \times\ \frac{1 }{
\min\limits_{ \substack{ 1\leq i\leq N, i\neq j \\ 1\leq j \leq M }}|x_{i}-x_{j}|^{[A_{1}]+\ldots+[A_{M}]-D-1+\frac{1}{r}}
\min\limits_{ \substack{ 1\leq i\leq N, i\neq j \\ M+1\leq j \leq N } }|x_{i}-x_{j}|^{[A_{M+1}]+\ldots+[A_{N}]+1}  }
\\
&\times  \left(\frac{\max_{1\leq i\leq M}|x_{i}-x_{M}|}{\min_{1\leq i<j\leq M}|x_{i}-x_{j}|}\cdot\frac{\max_{M+1\leq i\leq N}|x_{i}-x_{N}|}{\min_{M+1\leq i<j\leq N}|x_{i}-x_{j}|}\right)^{D'+|w|+s+1}\, \sup\Big(1,(\Lambda+m)|x_{M}|\Big)^{|w|}\\
&\times  \left( \frac{\max\limits_{1\leq i\leq N}|x_{i}-x_{M}|}{\min\limits_{1\leq i\leq M<j\leq N}|x_{i}-x_{j}|} \right)^{|w|} \Pol_{1}\left(\log\frac{\Lambda+m}{m}\right) \Pol_{2}\left(\frac{|\vec{p}|}{\Lambda+m}\right)
\end{split}
\een
holds, where $\Pol_{i}(x)$ are polynomials in $x$ with positive coefficients.
\end{thm}
\begin{rem}
A version of this bound with a stronger control over the numerical factors in the polynomials can be found in~\cite{Holland:2013we}. The bound clearly implies the scaling identity \eqref{Hinfrared} by choosing $s$ large enough and $w=0$.
\end{rem}
Before we come the proof of the theorem, we also note the following corollary:
\begin{cor}\label{corboundH}
Let $M<N$ and $t\in\{0,1\}$. For any $D\leq D'=[A_{1}]+\ldots+[A_{M}]$, and any $r\in\mathbb{N}$, the bound
\ben\label{eqboundHcor}
\begin{split}
&\Big|\,\partial_{\La}^{t}\partial_{\vecp}^{w}\G^{\Lambda,\Lambda_{0}}_{2n,l}([\bigotimes_{i=1}^{M}\O_{A_{i}}(x_{i})]_{D};\bigotimes_{j=M+1}^{N}\O_{A_{j}}(x_{j}); \vec{p})\, \Big| \leq  \frac{ (\Lambda+m)^{-t-\frac{1}{r}-2n-|w|}  } {\min\limits_{
 \substack{ 1\leq i\leq M \\ M+1\leq j \leq N } } \Big[ |x_{i}-x_{j}|, 1/(\La+m)\Big]^{D} 
}\\
& \times\ \frac{\left( \frac{\max\limits_{1\leq i\leq N}|x_{i}-x_{M}|}{\min\limits_{1\leq i\leq M<j\leq N}|x_{i}-x_{j}|} \right)^{|w|} \sup\Big(1,(\Lambda+m)(|x_{M}|+|x_{N}|)\Big)^{|w|} }{
\min\limits_{ \substack{ 1\leq i\leq N, i\neq j \\ 1\leq j \leq M }}\Big[ |x_{i}-x_{j}|, 1/(\La+m)\Big]^{[A_{1}]+\ldots+[A_{M}]-D-1+\frac{1}{r}}
\min\limits_{ \substack{ 1\leq i\leq N, i\neq j \\ M+1\leq j \leq N } }\Big[ |x_{i}-x_{j}|, 1/(\La+m)\Big]^{[A_{M+1}]+\ldots+[A_{N}]+1}  }
\\
&\times  \left(\frac{\max_{1\leq i\leq M}|x_{i}-x_{M}|}{\min_{1\leq i<j\leq M}|x_{i}-x_{j}|}\cdot\frac{\max_{M+1\leq i\leq N}|x_{i}-x_{N}|}{\min_{M+1\leq i<j\leq N}|x_{i}-x_{j}|}\right)^{D'+|w|+1}   \Pol_{1}\left(\log\frac{\Lambda+m}{m}\right) \Pol_{2}\left(\frac{|\vec{p}|}{\Lambda+m}\right)
\end{split}
\een
holds, where $\Pol_{i}(x)$ are polynomials in $x$ with positive coefficients.
\end{cor}
\begin{rem}
We make the following observations:
\begin{enumerate} 
\item Scaling the spacetime arguments by $\varepsilon>0$, one can check that the bound scales as $\varepsilon^{-[A_{1}]-\ldots-[A_{N}]-\frac{1}{r}}$ for small $\varepsilon$. Thus, we verify
\ben
\operatorname{sd}\left(G^{\Lambda,\Lambda_{0}}\left([\bigotimes_{i=1}^{M}\O_{A_{i}}]_{D}\bigotimes_{M+1}^{N}\O_{A_{i}}\right)\right)\leq [A_{1}]+\ldots+[A_{N}]\, .
\een

\item  Scaling only the variables $x_{1},\ldots,x_{M}$ by $\varepsilon>0$, we note that the bound scales as\\ $\varepsilon^{-[A_{1}]-\ldots-[A_{M}]+D+1-\frac{1}{r}}$, which implies
\ben
\operatorname{sd}_{\{1,\ldots, M\}}\left(G^{\Lambda,\Lambda_{0}}\left([\bigotimes_{i=1}^{M}\O_{A_{i}}]_{D}\bigotimes_{M+1}^{N}\O_{A_{i}}\right)\right)\leq [A_{1}]+\ldots+[A_{M}]-D-1\, ,
\een
as claimed in \eqref{SDpart}.
\end{enumerate}
\end{rem}

\noindent
The corollary follows simply by combining the bounds of theorem \ref{thmboundH} and corollary \ref{corbound}. We continue with the proof of theorem \ref{thmboundH}.

\begin{proof}[Proof of theorem \ref{thmboundH}:]
We follow the same inductive scheme as in the proof of theorem~\ref{thmbound}. Since many parts of the proof are very similar to the one of theorem \ref{thmbound}, we will be relatively brief here. The strategy is to integrate the flow equation
\begin{small}
\ben\label{HFE2}
\begin{split}
&\hspace{2cm}\partial_{\Lambda}\partial_{\vec{p}}^{w}\H^{\Lambda,\Lambda_{0}}_{2n,l}\left([\bigotimes_{i=1}^{M}\O_{A_{i}}]_{D}; \bigotimes_{j=M+1}^{N}\O_{A_{j}} ; p_{1},\ldots,p_{2n}\right)=\\
\vspace{0.4cm}\\
&= \left(\atop{2n+2}{2}\right) \, \int_{k} \dot{C}^{\Lambda}(k)\ \partial_{\vec{p}}^{w}\H^{\Lambda,\Lambda_{0}}_{2n+2,l-1}\left([\bigotimes_{i=1}^{M}\O_{A_{i}}]_{D}; \bigotimes_{j=M+1}^{N}\O_{A_{j}} ; k, -k,  p_{1},\ldots,p_{2n}\right)\\
 &- \mathbb{S}\, \Bigg[ \partial_{\vec{p}}^{w} \sum_{ \substack{l_{1}+l_{2}=l \\ n_{1}+n_{2}=n+1}}\!\!\!\!\! 4n_{1}n_{2}
 \H^{\Lambda,\Lambda_{0}}_{D,2n_{1},l_{1}}\left([\bigotimes_{i=1}^{M}\O_{A_{i}}]_{D}; \bigotimes_{j=M+1}^{N}\O_{A_{j}}; q,  p_{1},\ldots,p_{2n_{1}-1}\right)\\
 & \hspace{4cm} \times \dot{C}^{\Lambda}(q)\,   \L^{\Lambda,\Lambda_{0}}_{2n_{2},l_{2}}(  p_{2n_{1}},\ldots,p_{2n}) \\
 &-\sum_{\substack{l_{1}+l_{2}=l-2 \\ n_{1}+n_{2}=n+1}}\!\!\!\!\! 4n_{1}n_{2}\partial^{w}_{\vec{p}}\int_{k} \F^{\Lambda,\Lambda_{0}}_{D,2n_{1},l_{1}}\left(\bigotimes_{i=1}^{M}\O_{A_{i}}; k,  p_{1},\ldots,p_{2n_{1}-1}\right)\\
 &\hspace{4cm}\times \dot{C}^{\Lambda}(k)\,   \F^{\Lambda,\Lambda_{0}}_{2n_{2},l_{2}}\left(\bigotimes_{i=M+1}^{N}\O_{A_{i}}; -k, p_{2n_{1}},\ldots,p_{2n}\right)   \\
&-\hspace{-.5cm}\sum_{\substack{l_{1}+\ldots+l_{M+1}=l-1 \\ n_{1}+\ldots+n_{M+1}=n+1\\1\leq a\leq M }}\hspace{-.5cm} 4n_{a}n_{M+1} \partial_{\vec{p}}^{w} \int_{k}\, \L^{\Lambda,\Lambda_{0}}_{2n_{a},l_{a}}(\O_{A_{a}}; k,  \vec{p}_{a})  \prod_{c\in\{1,\ldots,M\}\setminus\{a\}}  \L^{\Lambda,\Lambda_{0}}_{2n_{c},l_{c}}(\O_{A_{c}};  \vec{p}_{c})   \\
& \hspace{4cm}\times \dot{C}^{\Lambda}(k)  \F^{\Lambda,\Lambda_{0}}_{2n_{M+1},l_{M+1}}\left(\bigotimes_{i=M+1}^{N}\O_{A_{i}}; -k, p_{2n_{M}},\ldots,p_{2n}\right)  \\
&-\hspace{-.5cm}\sum_{\substack{l_{M}+\ldots+l_{N}=l-1 \\ n_{M}+\ldots+n_{N}=n+1\\M+1\leq a\leq N }}\hspace{-.5cm} 4n_{a}n_{M} \partial_{\vec{p}}^{w} \int_{k}\, \L^{\Lambda,\Lambda_{0}}_{2n_{a},l_{a}}(\O_{A_{a}}; k,  \vec{p}_{a}) \prod_{c\in\{M+1,\ldots,N\}\setminus\{a\}} \L^{\Lambda,\Lambda_{0}}_{2n_{c},l_{c}}(\O_{A_{c}};   \vec{p}_{c})   \\
& \hspace{4cm}\times \dot{C}^{\Lambda}(k)  \F^{\Lambda,\Lambda_{0}}_{D,2n_{M},l_{M}}\left(\bigotimes_{i=1}^{M}
\O_{A_{i}}; -k, p_{1},\ldots,p_{2n_{M}-1}\right)  \\
&-\hspace{-.5cm}\sum_{\substack{l_{1}+\ldots+l_{N}=l \\ n_{1}+\ldots+n_{N}=n+1\\1\leq a\leq M < b\leq N }}\hspace{-.75cm} 4n_{a}n_{b} \partial_{\vec{p}}^{w} \int_{k}\! \L^{\Lambda,\Lambda_{0}}_{2n_{a},l_{a}}(\O_{A_{a}}; k,  \vec{p}_{a})\,  \dot{C}^{\Lambda}(k)  \L^{\Lambda,\Lambda_{0}}_{2n_{b},l_{b}}(\O_{A_{b}}; -k,  \vec{p}_{b})\hspace{-.8cm}  \prod_{c\in\{1,\ldots,N\}\setminus\{a,b\}} \hspace{-.8cm} \L^{\Lambda,\Lambda_{0}}_{2n_{c},l_{c}}(\O_{A_{c}};   \vec{p}_{c}) \Bigg]
\end{split}
\een
\end{small}
and to bound each of the six terms on the right hand side separately. A simplification arises from the fact that we do not have to consider \emph{relevant} contributions here, since the boundary conditions for the functionals $H^{\Lambda,\Lambda_{0}}$ are always given at $\Lambda=\Lambda_{0}$. Here we first verify the bound for the case $x_{M}=0$ before we proceed to the general case $x_{M}\neq 0$, using the translation property \eqref{Htrans}.

\begin{itemize}

\item \underline{\textsf{First and second term on the r.h.s. of the flow equation:}}

The first two terms on the r.h.s. of the flow equation, which are linear in $\H^{\Lambda,\Lambda_{0}}$, can be estimated using the same inductive scheme used in the proof of theorem \ref{thmbound}. As mentioned above, we only have to consider the part of the induction which refers to \emph{irrelevant contributions}.  Collecting all the $\vec{x}$ dependent terms in $\bar{\m}$ (recall that we set $x_{M}=0$ at this point),
\ben
\begin{split}
\bar{\m}(\vec{x},D',D,w):= &\frac{\left(\frac{\max_{1\leq i\leq M}|x_{i}|}{\min_{1\leq i<j\leq M}|x_{i}-x_{j}|}\cdot\frac{\max_{M+1\leq i\leq N}|x_{i}-x_{N}|}{\min_{M+1\leq i<j\leq N}|x_{i}-x_{j}|}\right)^{D'+s+|w|+1} }{\min\limits_{
 \substack{ 1\leq i\leq M \\ M+1\leq j \leq N } } |x_{i}-x_{j}|^{D+s}\, m^{s}  \min\limits_{ \substack{ 1\leq i\leq N, i\neq j \\ 1\leq j \leq M }}|x_{i}-x_{j}|^{[A_{1}]+\ldots+[A_{M}]-D-1+\frac{1}{r}} } \\
 \times& \frac{1 } {
\min\limits_{ \substack{ 1\leq i\leq N, i\neq j \\ M+1\leq j \leq N } }|x_{i}-x_{j}|^{[A_{M+1}]+\ldots+[A_{N}]+1}  }\, \left( \frac{\max\limits_{1\leq i\leq N}|x_{i}|}{\min\limits_{1\leq i\leq M<j\leq N}|x_{i}-x_{j}|} \right)^{|w|}
\end{split}
\een
we can in fact follow exactly the same steps as in the proof of theorem \ref{thmbound} (our boundary conditions correspond to the case $D=-1$ there). Note that, crucially, $\bar{\m}$ satisfies the condition \eqref{xiprop}, and that the condition \eqref{xiprop2}, which is not satisfied by $\bar{\m}$, is not used in the part of the proof dealing with irrelevant contributions. We will not repeat the lengthy estimates here.

\item \underline{\textsf{Third term on the r.h.s. of the flow equation:}}

In order to estimate the third term on the r.h.s. of the flow equation, we will make use of he following lemma.
\begin{lemma}\label{lemmakintH1}
Let $x_{M}=0$ and  $n_{1}+n_{2}=n+1$ and $l-2=l_{1}+l_{2}\geq 0$. Further, let $\vec{p}_{1}=(p_{1},\ldots,p_{2n_{1}-1})$ and $\vec{p}_{2}=(p_{2n_{1}},\ldots,p_{2n})$. Then we have for any $D\leq [A_1]+\ldots+[A_M]$, $r\in\mathbb{N}$ and $s\in\mathbb{N}_{0}$
\ben\label{eq:lemH1}
\begin{split}
&\bigg|\ \partial^{w}_{\vec{p}}\int_{k}\, \F^{\Lambda,\Lambda_{0}}_{D,2n_{1},l_{1}}(\bigotimes_{i=1}^M\O_{A_{i}}(x_{i}); k, \vec{p}_{1})\,  \dot{C}^{\Lambda}(k)\,  \F^{\Lambda,\Lambda_{0}}_{2n_{2},l_{2}}(\bigotimes_{j=M+1}^{N}\O_{A_{j}}(x_{j}); -k, \vec{p}_{2})\ \bigg|\\
&\leq\quad   (\Lambda+m)^{-2n-|w|-1-\frac{1}{r}}\,   \Pol_{1}\left(\log\frac{\Lambda+m}{m}\right) \Pol_{2}\left(\frac{|\vec{p}|}{\Lambda+m}\right) \\
&\times  \frac{ m^{-s}   \left(\frac{\max_{1\leq i\leq M}|x_{i}|}{\min_{1\leq i<j\leq M}|x_{i}-x_{j}|}\cdot \frac{\max_{M+1\leq i\leq N}|x_{i}-x_{N}|}{\min_{M+1\leq i<j\leq N}|x_{i}-x_{j}|}\right)^{D+s+|w|+1} 
}{\min\limits_{1\leq i<j\leq M}|x_{i}-x_{j}|^{[A_{1}]+\ldots+[A_{M}]-D-1+\frac{1}{r}}\min\limits_{M+1\leq i<j\leq N}|x_{i}-x_{j}|^{[A_{M+1}]+\ldots+[A_{N}]+1} |x_{N}|^{D+s} }  
\end{split}
\een
where $\Pol_{i}(x)$ are polynomials in $x$ with positive coefficients.
\end{lemma}
\begin{proof}
Using the translation properties of the $F$-functionals, eq.\eqref{Ftrans}, the integral can be written as
\ben
\begin{split}
&\left|\ \partial^{w}_{\vec{p}}\int_{k} \F^{\Lambda,\Lambda_{0}}_{D,2n_{1},l_{1}}\left(\bigotimes_{i=1}^{M}\O_{A_{i}}(x_{i}); k,  \vec{p}_{1}\right)\,  \dot{C}^{\Lambda}(k)\,   \F^{\Lambda,\Lambda_{0}}_{2n_{2},l_{2}}\left(\bigotimes_{i=M+1}^{N}\O_{A_{i}}(x_{i}); -k, \vec{p}_{2} \right)\ \right|\\
&\leq \bigg|\ \int_{k}\sum_{w_{1}+w_{2}+w_{3}=w} c_{\{w_{i}\}}\, \partial_{\vec{p}}^{w_{3}} \e^{ ix_{N}(p_{2n_{1}}+\ldots+p_{{2n}})-ik x_{N}}\\
&\quad\times  \partial_{\vecp}^{w_{1}} \F^{\Lambda,\Lambda_{0}}_{D,2n_{1},l_{1}}\left(\bigotimes_{i=1}^{M}\O_{A_{i}}(x_{i}); k,  \vec{p}_{1}\right)\,  \dot{C}^{\Lambda}(k)  \\
&\qquad\qquad\qquad\qquad\times\partial_{\vecp}^{w_{2}}\F^{\Lambda,\Lambda_{0}}_{2n_{2},l_{2}}\left(\bigotimes_{i=M+1}^{N}\O_{A_{i}}(x_{i}-x_{N}); -k, \vec{p}_{2}\right)\, \bigg|
\end{split}
\een 
The momentum derivatives on the exponential can be turned into $k$-derivatives and moved onto the moments of the $F$-functionals via partial integration. Just as in the proof of lemma \ref{lem}, we now introduce $D+s$ additional $k$-derivatives at the cost of a factor $(2/|x_{N}|)^{D+s}$, in order to obtain the desired dependence on $\Lambda$. Hence
\ben\label{thmHboundlem1}
\begin{split}
&\left|\ \partial^{w}_{\vec{p}}\int_{k} \F^{\Lambda,\Lambda_{0}}_{D,2n_{1},l_{1}}\left(\bigotimes_{i=1}^{M}\O_{A_{i}}(x_{i}); k,  \vec{p}_{1}\right)\,  \dot{C}^{\Lambda}(k)\,   \F^{\Lambda,\Lambda_{0}}_{2n_{2},l_{2}}\left(\bigotimes_{i=M+1}^{N}\O_{A_{i}}(x_{i}); -k, \vec{p}_{2}\right)\ \right|\\
&\leq\bigg|\ \int_{k}\sum_{w_{1}+w_{2}+w_{3}=w} c_{\{w_{i}\}}    \e^{ ix_{N}(p_{2n_{1}}+\ldots+p_{{2n}})- ik x_{N}} \ \left(\frac{2}{|x_{N}|}\right)^{D+s} \\
&\times  \partial_{k_{\alpha}}^{D+s}\partial_{k}^{w_{3}} \bigg( \partial_{\vecp}^{w_{1}} \F^{\Lambda,\Lambda_{0}}_{D,2n_{1},l_{1}}\left(\bigotimes_{i=1}^{M}\O_{A_{i}}(x_{i}); k,  \vec{p}_{1}\right)\,  \dot{C}^{\Lambda}(k)\\  
&\hspace{5cm}\times\partial_{\vecp}^{w_{2}} \F^{\Lambda,\Lambda_{0}}_{2n_{2},l_{2}}\left(\bigotimes_{i=M+1}^{N}\O_{A_{i}}(x_{i}-x_{N}); -k, \vec{p}_{2}\right)\bigg)\, \bigg|
\end{split}
\een 
where $\alpha\in\{1,\ldots, 4\}$ corresponds to the maximal component of $x_{N}$, i.e. $||x_{N}||=:|x_{N,\alpha}|$. Distributing the $k$-derivatives over the three factors, substituting our bounds for the $F$-functionals [see theorem \ref{thmbound}] and estimating the $k$-integral as in \eqref{gammaloop} then yields the lemma after some straightforward estimates.
\end{proof}
Continuing the proof of theorem \ref{thmboundH}, we see that this lemma confirms the validity of the claimed bound, \eqref{eqboundH}, in the case $t=1$, $x_{M}=0$.
 For $t=0$, we now estimate the $\Lambda$-integral over the third term on the r.h.s. of the flow equation combining lemma \ref{lemmakintH1} and lemma \ref{lemA}. One thereby verifies that also this contribution is  compatible with the claimed bound.

\item \underline{\textsf{Fourth and fifth term on the r.h.s. of the flow equation:}}

The following lemma will help us to estimate both these contributions:
\begin{lemma}\label{lemmakintH2}
Let  $n_{1}+\ldots+n_{M+1}=n+1$ and $l-1=l_{1}+\ldots+l_{M+1}$.
Then we have for any $D\leq [A_{M+1}]+\ldots+[A_N]$ and $a\in\{1,\ldots, M\}$ 
\ben\label{eq:lemH2}
\begin{split}
&\bigg|\ \partial_{\vec{p}}^{w} \int_{k}\, \L^{\Lambda,\Lambda_{0}}_{2n_{a},l_{a}}(\O_{A_{a}}; k, \vec{p}_{a})  \prod_{c\in\{1,\ldots,M\}\setminus\{a\}} \L^{\Lambda,\Lambda_{0}}_{2n_{c},l_{c}}(\O_{A_{c}};  \vec{p}_{c})   \\
&\qquad \times \dot{C}^{\Lambda}(k)  \F^{\Lambda,\Lambda_{0}}_{D,2n_{M+1},l_{M+1}}\left(\bigotimes_{i=M+1}^{N}\O_{A_{i}}; -k, p_{2n_{M}},\ldots,p_{2n}\right) \bigg|\\
&\leq\,  \Lambda^{-2n-|w|-\frac{1}{r_{1}}-\frac{1}{r_{2}}}\,\Pol_{1}\left(\log\frac{\Lambda+m}{m}\right) \Pol_{2}\left(\frac{|\vec{p}|}{\Lambda+m}\right)  \\
&\times  \frac{ m^{-s} \left(\frac{\max_{M+1\leq i\leq N}|x_{i}-x_{N}|}{\min_{M+1\leq i<j\leq N}|x_{i}-x_{j}|}\right)^{D+s+[A_{1}]+\ldots+[A_{M}]+|w|+1} \left(\frac{\max_{1\leq i\leq M}|x_{i}|}{|x_{N}-x_{a}|}\right)^{|w|}}{|x_{N}-x_{a}|^{[A_{1}]+\ldots+[A_{M}]+D+s+\frac{1}{r_{1}}}\min_{M+1\leq i<j\leq N}|x_{i}-x_{j}|^{[A_{M+1}]+\ldots+[A_{N}]-D-1+\frac{1}{r_{2}}} }   
\end{split}
\een
where $r_{1},r_{2}\in\mathbb{N}$ and $s\in\mathbb{N}_{0}$ and where $\Pol_{i}(x)$ are polynomials in $x$ with positive coefficients.
\end{lemma}
\begin{proof}
Using the same strategy as in the proof of lemma \ref{lemmakintH1}, we can estimate the l.h.s. of eq.\eqref{eq:lemH2} as
\ben\label{lem9eq}
\begin{split}
&\bigg|\ \partial_{\vec{p}}^{w} \int_{k}\, \L^{\Lambda,\Lambda_{0}}_{2n_{a},l_{a}}(\O_{A_{a}}; k,  \vec{p}_{a}) \prod_{c\in\{1,\ldots,M\}\setminus\{a\}} \L^{\Lambda,\Lambda_{0}}_{2n_{c},l_{c}}(\O_{A_{c}};   \vec{p}_{c} )  \\
&\qquad \times  \dot{C}^{\Lambda}(k)  \F^{\Lambda,\Lambda_{0}}_{D,2n_{M+1},l_{M+1}}\left(\bigotimes_{i=M+1}^{N}\O_{A_{i}}; -k, p_{2n_{M}},\ldots,p_{2n}\right)  \ \bigg|\\
&\leq \sum_{w_{1}+w_{2}=w} c_{\{w_{i}\}} \left(\frac{\max_{1\leq i \leq M}|x_{i}|}{|x_{a}-x_{N}|/2}\right)^{|w_{1}|}\frac{1}{(|x_{a}-x_{N}|/2)^{[A_{1}]+\ldots+[A_{M}]+D+s}} \\
&\times \bigg| \int_{k} \e^{i k (x_{a}-x_{N})} \partial_{k_{\alpha}}^{[A_{1}]+\ldots+[A_{M}]+D+s+|w_{1}|}\partial_{\vecp}^{w_{2}}\bigg( \L^{\Lambda,\Lambda_{0}}_{2n_{a},l_{a}}(\O_{A_{a}}(0); k,  \vec{p}_{j}) \, \dot{C}^{\Lambda}(k)  \!\!\!\!\! \prod_{c\in\{1,\ldots,M\}\setminus\{a\}}  \\ 
&\quad\times \L^{\Lambda,\Lambda_{0}}_{2n_{c},l_{c}}(\O_{A_{c}}(0);   \vec{p}_{c})  \F^{\Lambda,\Lambda_{0}}_{D,2n_{M+1},l_{M+1}}\left(\bigotimes_{i=M+1}^{N}\O_{A_{i}}(x_{i}-x_{N}); -k, p_{2n_{M}},\ldots,p_{2n}\right) \bigg)\bigg|
\end{split}
\een
Using the partial integration ''trick'' as in \eqref{3rdirr2}, i.e. transforming the integration variable to $\tilde{k}_{\mu}^{r_{1}}=k_{\mu}$ and integrating by parts once more, inserting the known bounds for the CAG's and $F^{\Lambda,\Lambda_{0}}$-functionals and estimating the loop integral as in \eqref{gammaloop}, we verify the lemma after some lengthy but straightforward estimates.
%
%
\end{proof}

This lemma allows us to find a bound for the fourth and the fifth term on the r.h.s. of the flow equation, \eqref{HFE2}. Using the case $D=-1$ and $r_{2}=1$ in the lemma, one immediately checks that the fourth term on the r.h.s. of the flow equation satisfies the bound \eqref{eqboundH} for $t=1$, $x_{M}=0$, and combining this bound with lemma \ref{lemA} one also verifies the case $t=0$. Similarly, for the fifth term on the r.h.s. of the flow equation we use lemma \ref{lemmakintH2} again, but we exchange the role of the indices $(1,\ldots, M)\leftrightarrow (M+1,\ldots,N)$. This way, it is not hard to verify that the claimed bound is satisfied by the contributions in question.

\item \underline{\textsf{Sixth term on the r.h.s. of the flow equation:}}

The loop integral in last term on the r.h.s. of the flow equation \eqref{HFE2} can be estimated with the help of lemma \ref{lem}, and the subsequent $\Lambda$-integral can be estimated with the help of lemma \ref{lemA}.

\end{itemize}

In summary, we have found that all six terms on the r.h.s. of the flow equation satisfy the bound claimed in theorem \ref{thmboundH} for $x_{M}=0$. To finish the proof, we combine the translation formula \eqref{Htrans} with the previously established bound for the $x_{M}=0$ in order to obtain an estimate for $x_{M}\neq 0$.
\end{proof}

\section[Proof of proposition 4]{Proof of proposition \ref{thmAGAP}}\label{ap:propgder}

Here we are going to prove that both sides of the equation
\ben\label{AGAPrepeat}
\begin{split}
&\hbar\, \partial_{g}G^{\Lambda,\Lambda_{0}}(\bigotimes_{i=1}^{N}\O_{A_{i}}(x_{i}))=\\
&\frac{-1}{4!}\int\d^{4} y\Big[ G^{\Lambda,\Lambda_{0}}(\bigotimes_{i=1}^{N}\O_{A_{i}}(x_{i})\otimes \varphi^{4}(y))-G^{\Lambda,\Lambda_{0}}(\bigotimes_{i=1}^{N}\O_{A_{i}}(x_{i})) L^{\Lambda,\Lambda_{0}}(\varphi^{4}(y))\\
&- \sum_{j=1}^{N}\sum_{[C]\leq [A_{j}]} \C_{\In A_{j}}^{C}(y,x_{j}) G^{\Lambda,\Lambda_{0}}(\bigotimes_{\atop{i=1}{i\neq j}}^{N}\O_{A_{i}}(x_{i})\otimes\O_{C}(x_{j})) \Big]
\end{split}
\een
indeed satisfy the same flow equations and boundary conditions. We make use of the decomposition $G^{\Lambda,\Lambda_{0}}(\bigotimes_{i=1}^{N}\O_{A_{i}})= \hbar F^{\Lambda,\Lambda_{0}}(\bigotimes_{i=1}^{N}\O_{A_{i}})+ \prod_{i=1}^{N}L^{\Lambda,\Lambda_{0}}(\O_{A_{i}})$ and study the $g$-derivative of the expressions on the right side of this equation. To begin with, consider the derivative of the factorized contribution to $\hbar\partial_{g}G^{\Lambda,\Lambda_{0}}(\bigotimes_{i=1}^{N}\O_{A_{i}}(x_{i}))$.
Using proposition \ref{thmAP2} we find
\ben\label{factorgder}
\hbar\,\partial_{g}\prod_{i=1}^{N}L^{\Lambda,\Lambda_{0}}(\O_{A_{i}})=\frac{\hbar}{4!}\int\d^{4}y\, \sum_{j=1}^{N}L^{\Lambda,\Lambda_{0}}_{D=[A_{j}]}(\O_{A_{j}}\otimes\varphi^{4}(y))\prod_{\atop{i=1}{i\neq j}}^{N}L^{\Lambda,\Lambda_{0}}(\O_{A_{i}})
\een
We find a similar contribution on the r.h.s. of eq.\eqref{AGAPrepeat}:
\ben\label{AGAPcancel}
\begin{split}
&\frac{1}{4!}\int\d^{4}y\, \sum_{j=1}^{N}\sum_{[C]\leq [A_{j}]} \C_{\In A_{j}}^{C}(y,x_{j})L^{\Lambda,\Lambda_{0}}(\O_{C}(y)) \prod_{\atop{i=1}{i\neq j}}^{N}L^{\Lambda,\Lambda_{0}}(\O_{A_{i}})\\
=&\frac{\hbar}{4!}\int\d^{4}y\, \sum_{j=1}^{N}\left(L^{\Lambda,\Lambda_{0}}_{D=[A_{j}]}(\O_{A_{j}}\otimes\varphi^{4}(y))-L^{\Lambda,\Lambda_{0}}(\O_{A_{j}}\otimes\varphi^{4}(y))\right)\prod_{\atop{i=1}{i\neq j}}^{N}L^{\Lambda,\Lambda_{0}}(\O_{A_{i}})
\end{split}
\een
In the second line we made use of equation \eqref{RGD}. We see that the first term on the r.h.s. of \eqref{AGAPcancel} coincides with the r.h.s. of equation \eqref{factorgder}, which means that these terms cancel in equation \eqref{AGAPrepeat}. Note also that the factorized contributions (i.e. terms containing only CAG's with one insertion) to the first two terms on the r.h.s of equation \eqref{AGAPrepeat} cancel each other. Let us now come to the various contributions from the $F^{\Lambda,\Lambda_{0}}$-functionals. Consider first the $g$-derivative of the flow equation for $F^{\Lambda,\Lambda_{0}}$
\ben\label{AGAPfe5}
\begin{split}
&\partial_{\Lambda}\partial_{g}F^{\Lambda,\Lambda_{0}}(\bigotimes_{i=1}^{N}\O_{A_{i}})\\
=&\frac{\hbar}{2}\bra \varp , \dot{C}^{\Lambda}\star\varp \ket\, \partial_{g}F^{\Lambda,\Lambda_{0}}(\bigotimes_{i=1}^{N}\O_{A_{i}})- \bra \varp \partial_{g}F^{\Lambda,\Lambda_{0}}(\bigotimes_{i=1}^{N}\O_{A_{i}}), \dot{C}^{\Lambda}\star\varp L^{\Lambda,\Lambda_{0}} \ket\\
-&\bra \varp F^{\Lambda,\Lambda_{0}}(\bigotimes_{i=1}^{N}\O_{A_{i}}), \dot{C}^{\Lambda}\star\varp \frac{1}{4!}\int\d^{4} y\, L^{\Lambda,\Lambda_{0}}(\varphi^{4}(y)) \ket\\
+&\sum_{1\leq i<j\leq N} \bra \varp \frac{1}{4!}\int\d^{4} y\,  L_{[A_{i}]}^{\Lambda,\Lambda_{0}}(\O_{A_{i}}\otimes\varphi^{4}(y)), \dot{C}^{\Lambda}\star\varp L^{\Lambda,\Lambda_{0}}(\O_{A_{j}})   \ket \!\!\!\!\!\!\!\!\!\!\!\! \prod_{r\in\{1,\ldots,N\}\setminus\{i,j\}}\hspace{-.6cm} L^{\Lambda,\Lambda_{0}}(\O_{A_{r}})\\
+&\sum_{1\leq i<j\leq N} \bra \varp L^{\Lambda,\Lambda_{0}}(\O_{A_{i}}), \dot{C}^{\Lambda}\star\varp L^{\Lambda,\Lambda_{0}}(\O_{A_{j}})   \ket \\
&\quad\times\sum_{k\in\{1,\ldots, N\}\setminus\{i,j\}}\frac{1}{4!}\int\d^{4} y\, L_{[A_{k}]}^{\Lambda,\Lambda_{0}}(\O_{A_{k}}\otimes\varphi^{4}(y)),  \!\!\!\!\!\! \prod_{r\in\{1,\ldots,N\}\setminus\{i,j,k\}} L^{\Lambda,\Lambda_{0}}(\O_{A_{r}})\, ,
\end{split}
\een
where we made use of proposition \ref{thmAP2}. The boundary conditions read
\ben\label{AGAPbc}
\partial^{w}_{\vec{p}} \partial_{g}\F^{\Lambda_{0},\Lambda_{0}}_{n,l}(\bigotimes_{i=1}^{N}\O_{A_{i}}; \vec{p})=0\quad \text{for all }n,l,w.
\een
We want to compare this to the flow equations for the terms on the r.h.s. of eq.\eqref{AGAPrepeat}. To start with, we have
\ben\label{AGAPfe1}
\begin{split}
&\partial_{\Lambda}\frac{-1}{4!}\int\d^{4} y\, F^{\Lambda,\Lambda_{0}}(\bigotimes_{i=1}^{N}\O_{A_{i}}\otimes \varphi^{4}(y))\\
&=\frac{\hbar}{2}\bra \varp , \dot{C}^{\Lambda}\star\varp \ket\, \frac{-1}{4!}\int\d^{4} y\, F^{\Lambda,\Lambda_{0}}(\bigotimes_{i=1}^{N}\O_{A_{i}}\otimes \varphi^{4}(y))\\
&- \bra \varp \frac{-1}{4!}\int\d^{4} y\, F^{\Lambda,\Lambda_{0}}(\bigotimes_{i=1}^{N}\O_{A_{i}}\otimes \varphi^{4}(y)), \dot{C}^{\Lambda}\star\varp L^{\Lambda,\Lambda_{0}} \ket\\
&-\!\!\!\!\hspace{-.2cm}\sum_{1\leq i<j\leq N}\hspace{-.2cm}
 \bra \varp  L^{\Lambda,\Lambda_{0}}(\O_{A_{i}}), \dot{C}^{\Lambda}\star\varp L^{\Lambda,\Lambda_{0}}(\O_{A_{j}})   \ket \!\!\!\!\!\!\!\!\hspace{-.2cm} \prod_{r\in\{1,\ldots,N\}\setminus\{i,j\}}\!\!\!\!\!\!\!\hspace{-.3cm} L^{\Lambda,\Lambda_{0}}(\O_{A_{r}}) \int\frac{\d^{4} y}{4!} L^{\Lambda,\Lambda_{0}}(\varphi^{4}(y))\\
&-\!\!\!\sum_{ j \in\{1,\ldots, N\}} \bra \varp  L^{\Lambda,\Lambda_{0}}(\O_{A_{j}}), \dot{C}^{\Lambda}\star\varp \frac{1}{4!}\int\d^{4} y L^{\Lambda,\Lambda_{0}}(\varphi^{4}(y))   \ket \ \prod_{\atop{r=1}{r\neq j}}^{N} L^{\Lambda,\Lambda_{0}}(\O_{A_{r}}) \\
\end{split}
\een
Next, we have
\ben\label{AGAPfe2}
\begin{split}
&\partial_{\Lambda}\left( F^{\Lambda,\Lambda_{0}}(\bigotimes_{i=1}^{N}\O_{A_{i}})\frac{1}{4!}\int\d^{4} y\, L^{\Lambda,\Lambda_{0}}(\varphi^{4}(y))\right) \\
&=\frac{\hbar}{2}\bra \varp , \dot{C}^{\Lambda}\star\varp \ket\, F^{\Lambda,\Lambda_{0}}(\bigotimes_{i=1}^{N}\O_{A_{i}})\frac{1}{4!}\int\d^{4} y\, L^{\Lambda,\Lambda_{0}}(\varphi^{4}(y))\\
&- \bra \varp F^{\Lambda,\Lambda_{0}}(\bigotimes_{i=1}^{N}\O_{A_{i}})\frac{1}{4!}\int\d^{4} y\, L^{\Lambda,\Lambda_{0}}(\varphi^{4}(y)), \dot{C}^{\Lambda}\star\varp L^{\Lambda,\Lambda_{0}} \ket\\
&- \hbar\bra \varp F^{\Lambda,\Lambda_{0}}(\bigotimes_{i=1}^{N}\O_{A_{i}}), \dot{C}^{\Lambda}\star\varp \frac{1}{4!}\int\d^{4} y\, L^{\Lambda,\Lambda_{0}}(\varphi^{4}(y)) \ket\\
&+\!\!\!\hspace{-.2cm}\sum_{1\leq i<j\leq N}\hspace{-.2cm}
 \bra \varp  L^{\Lambda,\Lambda_{0}}(\O_{A_{i}}), \dot{C}^{\Lambda}\star\varp L^{\Lambda,\Lambda_{0}}(\O_{A_{j}})   \ket \!\!\!\!\!\!\!\!\hspace{-.2cm} \prod_{r\in\{1,\ldots,N\}\setminus\{i,j\}}\!\!\!\!\!\!\!\hspace{-.3cm} L^{\Lambda,\Lambda_{0}}(\O_{A_{r}}) \int\frac{\d^{4} y}{4!} L^{\Lambda,\Lambda_{0}}(\varphi^{4}(y))\\
\end{split}
\een
Finally, for the last term in eq.\eqref{AGAPrepeat}:
\ben\label{AGAPfe3}
\begin{split}
&\partial_{\Lambda}\left( \frac{1}{4!}\int\d^{4} y\sum_{j=1}^{N}\sum_{[C]\leq [A_{j}]} \C_{\In A_{j}}^{C}(y,x_{j}) F^{\Lambda,\Lambda_{0}}(\bigotimes_{\atop{i=1}{i\neq j}}^{N}\O_{A_{i}}(x_{i})\otimes\O_{C}(x_{j}))\right) \\
&=\frac{\hbar}{2}\bra \varp , \dot{C}^{\Lambda}\star\varp \ket\, \frac{1}{4!}\int\d^{4} y\sum_{j=1}^{N}\sum_{[C]\leq [A_{j}]} \C_{\In A_{j}}^{C}(y,x_{j}) F^{\Lambda,\Lambda_{0}}(\bigotimes_{\atop{i=1}{i\neq j}}^{N}\O_{A_{i}}(x_{i})\otimes\O_{C}(x_{j}))\\
&- \bra \varp \int \frac{\d^{4} y}{4!} \sum_{j=1}^{N}\sum_{[C]\leq [A_{j}]}\hspace{-.2cm} \C_{\In A_{j}}^{C}(y,x_{j}) F^{\Lambda,\Lambda_{0}}(\bigotimes_{\atop{i=1}{i\neq j}}^{N}\O_{A_{i}}(x_{i})\otimes\O_{C}(x_{j})), \dot{C}^{\Lambda}\star\varp L^{\Lambda,\Lambda_{0}} \ket\\
&-\hbar\!\!\! \sum_{1\leq i<j\leq N} \bra \varp  L^{\Lambda,\Lambda_{0}}(\O_{A_{i}}), \dot{C}^{\Lambda}\star\varp L^{\Lambda,\Lambda_{0}}(\O_{A_{j}})   \ket\\
&\hspace{3cm}\times\sum_{k\in\{1,\ldots,N\}\setminus\{i,j\}} \frac{1}{4!} \int \d^{4} y L^{\Lambda,\Lambda_{0}}(\O_{A_{k}}\otimes\varphi^{4}(y)) \!\!\!\!\!\!\!\! \prod_{r\in\{1,\ldots,N\}\setminus\{i,j,k\}}\!\!\!\!\!\!\! L^{\Lambda,\Lambda_{0}}(\O_{A_{r}}) \\
&+\hbar\!\!\!\sum_{1\leq i<j\leq N} \bra \varp  L^{\Lambda,\Lambda_{0}}(\O_{A_{i}}), \dot{C}^{\Lambda}\star\varp L^{\Lambda,\Lambda_{0}}(\O_{A_{j}})   \ket\\
&\hspace{3cm}\times\sum_{k\in\{1,\ldots,N\}\setminus\{i,j\}} \frac{1}{4!}\int\d^{4} y L^{\Lambda,\Lambda_{0}}_{[A_{k}]}(\O_{A_{k}}\otimes\varphi^{4}(y)) \!\!\!\!\!\!\!\! \prod_{r\in\{1,\ldots,N\}\setminus\{i,j,k\}}\!\!\!\!\!\!\! L^{\Lambda,\Lambda_{0}}(\O_{A_{r}}) \\
&-\hbar\!\!\!\sum_{1\leq i<j\leq N} \bra \varp  L^{\Lambda,\Lambda_{0}}(\O_{A_{i}}), \dot{C}^{\Lambda}\star\varp \frac{1}{4!}\int\d^{4} y L^{\Lambda,\Lambda_{0}}(\O_{A_{j}}\otimes\varphi^{4}(y))    \ket \!\!\!\!\!\!\!\! \prod_{r\in\{1,\ldots,N\}\setminus\{i,j\}}\!\!\!\!\!\!\! L^{\Lambda,\Lambda_{0}}(\O_{A_{r}}) \\
&+\hbar\!\!\!\sum_{1\leq i<j\leq N} \bra \varp  L^{\Lambda,\Lambda_{0}}(\O_{A_{i}}), \dot{C}^{\Lambda}\star\varp \frac{1}{4!}\int\d^{4} y L^{\Lambda,\Lambda_{0}}_{[A_{j}]}(\O_{A_{j}}\otimes\varphi^{4}(y))    \ket \!\!\!\!\!\!\!\! \prod_{r\in\{1,\ldots,N\}\setminus\{i,j\}}\!\!\!\!\!\!\! L^{\Lambda,\Lambda_{0}}(\O_{A_{r}}) 
\end{split}
\een
Also recall the remaining term from eq.\eqref{AGAPcancel}, which satisfies the flow equation
\ben\label{AGAPfe4}
\begin{split}
&\partial_{\Lambda} \left(\frac{-1}{4!}\int\d^{4}y\, \sum_{j=1}^{N}L^{\Lambda,\Lambda_{0}}(\O_{A_{j}}\otimes\varphi^{4}(y))\prod_{\atop{i=1}{i\neq j}}^{N}L^{\Lambda,\Lambda_{0}}(\O_{A_{i}})\right)\\
&=\frac{\hbar}{2}\bra \varp , \dot{C}^{\Lambda}\star\varp \ket\, \frac{-1}{4!}\int\d^{4}y\, \sum_{j=1}^{N}L^{\Lambda,\Lambda_{0}}(\O_{A_{j}}\otimes\varphi^{4}(y))\prod_{\atop{i=1}{i\neq j}}^{N}L^{\Lambda,\Lambda_{0}}(\O_{A_{i}})\\
&- \bra \varp \frac{-1}{4!}\int\d^{4}y\, \sum_{j=1}^{N}L^{\Lambda,\Lambda_{0}}(\O_{A_{j}}\otimes\varphi^{4}(y))\prod_{\atop{i=1}{i\neq j}}^{N}L^{\Lambda,\Lambda_{0}}(\O_{A_{i}}), \dot{C}^{\Lambda}\star\varp L^{\Lambda,\Lambda_{0}} \ket\\
&+\hbar\!\!\!\sum_{1\leq i<j\leq N} \bra \varp  L^{\Lambda,\Lambda_{0}}(\O_{A_{i}}), \dot{C}^{\Lambda}\star\varp L^{\Lambda,\Lambda_{0}}(\O_{A_{j}})   \ket\\
&\hspace{3cm}\times\sum_{k\in\{1,\ldots,N\}\setminus\{i,j\}} \frac{1}{4!}\int\d^{4} y L^{\Lambda,\Lambda_{0}}(\O_{A_{k}}\otimes\varphi^{4}(y)) \!\!\!\!\!\!\!\!\!\!\! \prod_{r\in\{1,\ldots,N\}\setminus\{i,j,k\}}\!\!\!\!\!\!\!\!\!\!\! L^{\Lambda,\Lambda_{0}}(\O_{A_{r}}) \\
&+\hbar\!\!\!\!\!\sum_{1\leq i<j\leq N} \bra \varp  L^{\Lambda,\Lambda_{0}}(\O_{A_{i}}), \dot{C}^{\Lambda}\star\varp \frac{1}{4!}\int\d^{4} y L^{\Lambda,\Lambda_{0}}(\O_{A_{j}}\otimes\varphi^{4}(y))    \ket \!\!\!\!\!\!\!\!\!\!\! \prod_{r\in\{1,\ldots,N\}\setminus\{i,j\}}\!\!\!\!\!\!\! L^{\Lambda,\Lambda_{0}}(\O_{A_{r}}) \\
&+\!\!\!\sum_{ j \in\{1,\ldots, N\}} \bra \varp  L^{\Lambda,\Lambda_{0}}(\O_{A_{j}}), \dot{C}^{\Lambda}\star\varp \frac{1}{4!}\int\d^{4} y L^{\Lambda,\Lambda_{0}}(\varphi^{4}(y))   \ket \ \prod_{\atop{r=1}{r\neq j}}^{N} L^{\Lambda,\Lambda_{0}}(\O_{A_{r}})\, .
\end{split}
\een
Now, summing up equations \eqref{AGAPfe1},\eqref{AGAPfe2},\eqref{AGAPfe3} and \eqref{AGAPfe4}, one can check that these contributions satisfy the same flow equation as $\hbar\,\partial_{g}F^{\Lambda,\Lambda_{0}}(\bigotimes_{i=1}^{N}\O_{A_{i}})$, see eq.\eqref{AGAPfe5}. Also note that all these contributions satisfy the boundary conditions of the form \eqref{AGAPbc}. It follows that the left and right hand side of equation \eqref{AGAPrepeat} are equal. 

\vspace{.5cm}

\paragraph{ Acknowledgements:} We would like to thank Ch. Kopper for numerous useful 
conversations and comments. The research of S.H. is partly supported by ERC starting grant QC \& C 259562. 

\addcontentsline{toc}{chapter}{Bibliography}
\providecommand{\href}[2]{#2}\begingroup\raggedright\endgroup

\end{document}